\documentclass[11pt]{article}
\pdfoutput=1

\usepackage{html,amsmath,amssymb,amsfonts,epsfig,cancel,cite,longtable,siunitx,array,booktabs,setspace,arydshln}

\usepackage{graphicx, float, xspace, bbm}
\usepackage{arydshln}
\usepackage[english]{babel}
\usepackage[all]{xy}

\usepackage[T1]{fontenc}
\usepackage{dsfont}
\usepackage{mathrsfs}
\usepackage[mathscr]{euscript}
\usepackage{calligra}

\usepackage{mathtools}  
\usepackage{textcomp}
\usepackage{array}
\usepackage[usenames, dvipsnames]{color}
\usepackage{cite}
\usepackage[utf8]{inputenc}
\usepackage{amsthm}
\usepackage[multiple]{footmisc}

\usepackage{caption}
\usepackage{pifont}

\renewcommand\baselinestretch{1.4}
\numberwithin{equation}{section}

\newcommand{\bean}{\begin{eqnarray*}}
\newcommand{\eean}{\end{eqnarray*}}

\newcommand{\fref}[1]{Figure~\ref{#1}}

\newcommand{\capt}[3]{\parbox{#1}{\renewcommand{\baselinestretch}{1.2}
                                                           \caption{\label{#2}\small\it #3}}}

\newcommand{\ind}{\mathop{{\rm ind}}}

\newcommand{\IP}{\mathbb{P}}

\newcommand{\cO}{{\cal O}}
\newcommand{\cN}{{\cal N}}
\newcommand{\cM}{{\cal M}}

\newcommand{\cA}{{\cal A}}
\newcommand{\cB}{{\cal B}}
\newcommand{\cC}{{\cal C}}

\newcommand{\cV}{{\cal V}}

\def\cjn1{{\cA, \cC^*\otimes \wedge^j \cN^*}}
\def\bjn1{{\cA, \cB^*\otimes \wedge^j \cN^*}}
\def\vjn1{{\cA, \cV^*\otimes \wedge^j \cN^*}}
\def\cjn2{{\cA, \cC\otimes \wedge^j \cN^*}}
\def\bjn2{{\cA, \cB\otimes \wedge^j \cN^*}}
\def\vjn2{{\cA, \cV\otimes \wedge^j \cN^*}}

\newcolumntype{C}{>{$}c<{$}} 		
\newcolumntype{L}{>{$}l<{$}} 		

\newcommand{\floor}[1]{\left\lfloor{#1}\right\rfloor}
\newcommand{\ceil}[1]{\left\lceil{#1}\right\rceil}

\newcommand{\mc}{\mathcal}
\newcommand{\mbb}{\mathbb}

\newcommand{\dd}{\mathrm{d}}

\newcommand{\smlhdg}[1]{\bigskip\noindent\underline{\textbf{#1}}\medskip}
\newcommand{\smlhdgnogap}[1]{\noindent\underline{\textbf{#1}}\medskip}

\newcommand{\be}{\begin{equation}}
\newcommand{\ee}{\end{equation}}
\newcommand*{\nnbe}{\begin{equation}}
\newcommand*{\nnee}{\end{equation}}
\newcommand{\bea}{\begin{eqnarray}}
\newcommand{\eea}{\end{eqnarray}}
\newcommand{\ba}{\begin{align}}
\newcommand{\ea}{\end{align}}
\newcommand{\bi}{\begin{itemize}}
\newcommand{\ei}{\end{itemize}}
\newtheorem{thm}{Theorem}[section]
\newtheorem{defn}[thm]{Definition}
\newtheorem{prp}[thm]{Proposition}
\newtheorem{crl}[thm]{Corollary}
\newtheorem{lma}[thm]{Lemma}
\newtheorem{rmk}[thm]{Remark}
\newtheorem*{nndefn}{Definition}
\newtheorem*{nnthm}{Theorem}
\newtheorem*{nnprp}{Proposition}

\newsavebox{\overlongequation}
\newenvironment{longequation}
 {\begin{displaymath}\begin{lrbox}{\overlongequation}$\displaystyle}
 {$\end{lrbox}\makebox[0pt]{\usebox{\overlongequation}}\end{displaymath}}

\newcommand{\moricn}{\mathcal{M}}	
\newcommand{\nefcn}{\mathcal{N}}		
\newcommand{\dual}[1]{#1^*}			
\newcommand{\divcls}[1]{\left[{#1}\right]}			
\newcommand{\cls}[1]{\left|{#1}\right|}			

\newcommand{\numeq}{\equiv_\mathrm{num.}}		
\newcommand{\lineq}{\equiv}						

\newcommand{\surf}{S}				

\newcommand{\indiv}{D^-}							

\newcommand{\base}[1]{B_{#1}}						
\newcommand{\lb}{L} 								
\newcommand{\lbb}{\mc{L}}						
\newcommand{\fibprj}{\pi}							
\newcommand{\fibsec}{\sigma}						
\newcommand{\fibcls}{F}							

\newcommand{\divtf}{\mc{D}}						
\newcommand{\divb}{D}							
\newcommand{\crvtf}{\mc{C}}						

\newcommand{\cy}[1]{X_{#1}}						

\newcommand{\phifl}{\phi_Z^\downarrow}
\newcommand{\phicl}{\phi_Z^\uparrow}

\begin{document}

\title{{\Large \bf$~$\\[-21pt]
Cohomology Chambers on Complex Surfaces and \\ Elliptically Fibered Calabi-Yau Three-folds 
}}

\vspace{3cm}

\author{
Callum R.~Brodie${}^{1,2,4}$ and Andrei Constantin${}^{1,2,3,4}$ \\
}
\date{}
\maketitle
\thispagestyle{empty}
\begin{center} {        
       {\it 
       ${}^1$Rudolf Peierls Centre for Theoretical Physics, University of Oxford, \\ 
       Clarendon Laboratory, Parks Rd, Oxford OX1 3PU, United Kingdom\\[0.3cm]}
       }
      
         ${}^2$ {\it 
         Brasenose College, University of Oxford, \\ Radcliffe Square, OX1 4AJ, United Kingdom 
	} 
       
       ${}^3$ {\it 
       Mansfield College, University of Oxford,\\ Mansfield Road, OX1 3TF, United Kingdom
	}

	${}^4$ {\it 
	Pembroke College, University of Oxford,\\ St.~Aldates, OX1 1DW,  United Kingdom
	}
\end{center}

\vspace{12pt}
\abstract
\noindent
We determine several classes of smooth complex projective surfaces on which Zariski decomposition can be combined with vanishing theorems to yield cohomology formulae for all line bundles. The obtained formulae express cohomologies in terms of divisor class intersections, and are adapted to the decomposition of the effective cone into Zariski chambers. In particular, we show this occurs on generalised del Pezzo surfaces, toric surfaces, and K3 surfaces. In the second part we use these surface results to derive formulae for all line bundle cohomology on a simple class of elliptically fibered Calabi-Yau three-folds. Computing such quantities is a crucial step in deriving the massless spectrum in string compactifications.

\vspace{100pt}
\noindent\rule{4cm}{0.4pt}

\noindent callum.brodie@pmb.ox.ac.uk, andrei.constantin@physics.ox.ac.uk

\newpage

\tableofcontents

\newpage
\section{Introduction and Summary}

Vector bundle cohomology is an essential tool for string theory, being related to the degrees of freedom (particles) present in the low energy field theory limit. 
However, its computation is notoriously difficult and has been a major obstacle for progress in string phenomenology from its very beginning. In the last decade several computer implementations have been written to cope with this technical hurdle, automating laborious calculations that would otherwise be impossible to carry out in any practicable time \cite{cicypackage, cohomCalg:Implementation, cicytoolkit}. These codes primarily deal with holomorphic line bundle cohomology on complex manifolds, since line bundles feature in many important contexts in string theory and moreover can be used as building blocks for higher rank vector bundles. Though extremely useful for practical purposes, such implementations remain limited in two respects. First, the algorithms become increasingly slow and eventually unworkable for manifolds with a large Picard number (say, greater than $5$, for a rough estimate) as well as for line bundles with large first Chern class integers. For string model building this imposes a significant limitation in the exploration of the string landscape of solutions. Second, all algorithmic computations of cohomology give very little insight into the results and provide virtually no information about the cohomology of other line bundles, thus rendering the string model building effort unmanageable, ultimately a `trial and error' feat. 

A novel approach to the problem has recently emerged through the observation that for many classes of complex manifolds of interest in string theory, line bundle cohomology is described by simple, often locally polynomial, functions \cite{Buchbinder:2013dna, Constantin:2018otr, Constantin:2018hvl}. To date, this observation has been checked to hold true for the zeroth as well as all higher cohomologies on several classes of two and three-dimensional complex manifolds which include certain complete intersections in products of projective spaces, toric varieties and hypersurfaces therein, all del Pezzo and all Hirzebruch surfaces \cite{Constantin:2018hvl, Klaewer:2018sfl, Larfors:2019sie, Brodie:2019ozt, Brodie:2019dfx, Brodie:2019pnz}. The existence of simple closed-form expressions for cohomology is an interesting mathematical question in itself. For Physics, these provide an unexpected shortcut to incredibly hard computations needed for connecting String Theory to Particle Physics, making feasible the implementation of what is known in string model building as the `bottom-up approach'. This involves working out the topology and geometry of the compactification space by starting from physical data, such as the number of quark and lepton families, and the number of vector-like matter states, which get encoded in the compactification data as dimensions of certain vector bundle cohomologies. The context in which cohomology formulae are, perhaps, the most relevant for attempting a bottom-up string model building approach is that of heterotic string compactifications on smooth Calabi-Yau three-folds with abelian internal fluxes described by sums of line bundles (see for instance Refs.~\cite{Blumenhagen:2005ga, Blumenhagen:2006ux, Blumenhagen:2006wj, Anderson:2011ns, Anderson:2012yf, Anderson:2013xka, He:2013ofa, Anderson:2014hia}).  

The existence of cohomology formulae has been discovered through a combination of direct observation \cite{Constantin:2018otr, Buchbinder:2013dna, Constantin:2018hvl, Larfors:2019sie, Brodie:2019ozt} and machine learning \cite{Klaewer:2018sfl, Brodie:2019dfx} of line bundle cohomology dimensions computed algorithmically. A common feature of these formulae is that they involve a decomposition of the Picard group into disjoint regions, in each of which the cohomology function is polynomial or very close to polynomial. This pattern has been observed for the zeroth as well as all higher cohomologies, with a different region structure emerging for each type of cohomology. The number of regions often increases dramatically with the Picard number of the space. 
The origin of these formulae has been elucidated for certain complex surfaces in Refs.~\cite{Brodie:2019ozt, Brodie:2019pnz}.

\subsection{Simple example}
\label{sec:int_ex}

A central aim in the present paper is to give a general understanding of the appearance of functions describing the zeroth cohomology of line bundles on certain classes of non-singular complex projective surfaces. In dimension two it suffices to understand the zeroth cohomology function since this implies the existence of formulae for the first and second cohomologies by Serre duality and the Hirzebruch-Riemann-Roch theorem. 
We begin with a simple example. 

Consider a del Pezzo surface of degree $7$, obtained by blowing-up $\IP^2$ at two generic points, denoted as ${\rm dP}_2$ in the Physics literature. 
Within the cone of effective line bundles (divisor classes), one finds \cite{Brodie:2019pnz} that the zeroth cohomology is given by the value of a piecewise polynomial function. Outside of the effective cone the zeroth cohomology is trivial. \fref{dp2_piclat_numb_shifts} depicts the chambers, within each of which a single polynomial describes the zeroth cohomology. Region 0 corresponds to the nef cone, its interior being the K\"ahler cone.

\begin{figure}[ht]
  \begin{center}
  \includegraphics[scale=0.5]{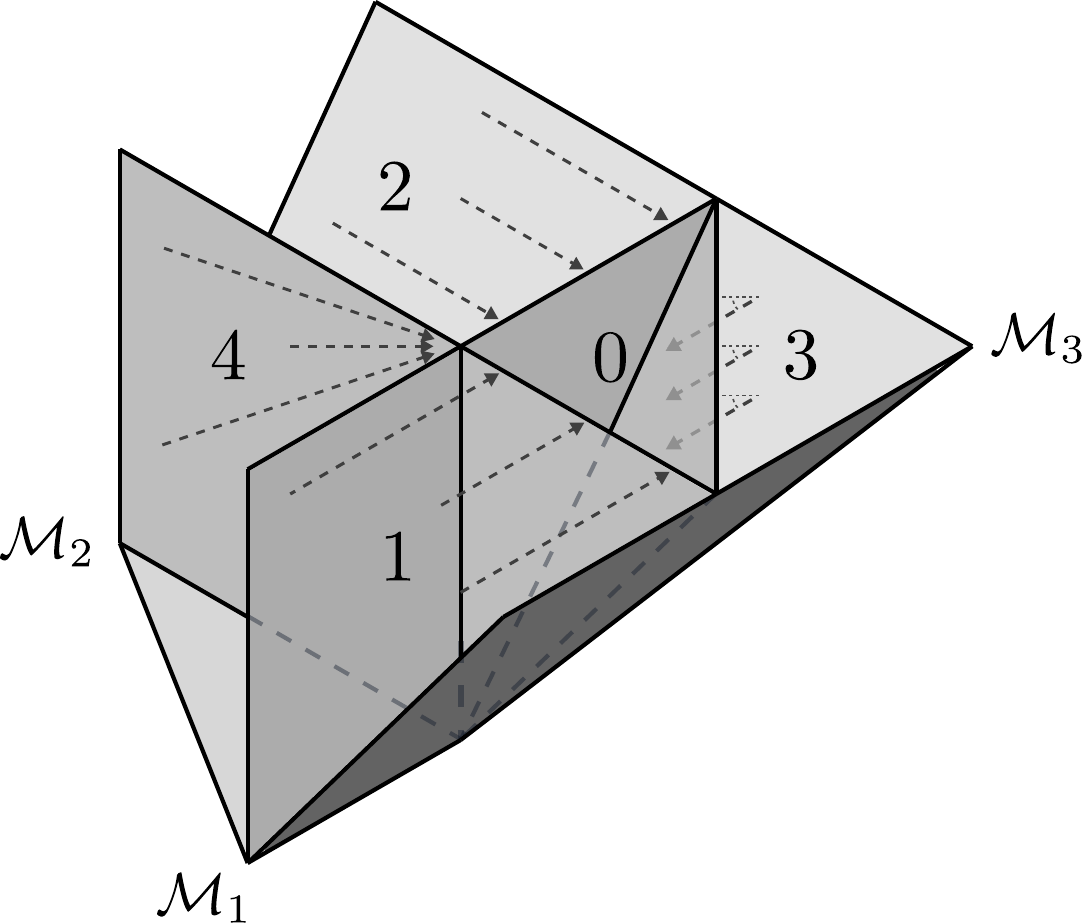}
  \end{center}
\caption{\itshape\small Zeroth cohomology chamber structure of the effective cone of dP$_2$.}
  \label{dp2_piclat_numb_shifts}
\end{figure}

The Picard lattice of dP$_2$ is spanned by the hyperplane class $H$ of $\IP^2$ and the two exceptional divisor classes $E_1$ and $E_2$ resulting from the two blow-ups. The effective cone (Mori cone) is generated by $\cM_1 = E_1$, $\cM_2 = E_2$, and $\cM_3 = H-E_1-E_2$. All three generators are rigid, satisfying $\cM_i^2=-1$.

In the nef cone, a vanishing theorem due to Kawamata and Viehweg implies that all higher cohomologies are trivial and hence the zeroth cohomology is given by the index (the Euler characteristic), which is a polynomial function of degree 2. In the other regions, it turns out that the zeroth cohomology is given by the index of a shifted divisor. More explicitly, for an effective line bundle associated with a divisor class $D$, one has the following locally polynomial formula.
\begin{equation}
\begin{tabular}{ L | L}
 		&~ h^0\big(\mathrm{dP}_2,\mc{O}_{\mathrm{dP}_2}(D)\big) \\
\hline
\textrm{Region 0} &~ \ind\big({\rm dP}_2, \cO_{\mathrm{dP}_2}(D)\big) \\
\textrm{Region 1} &~ \ind\big({\rm dP}_2, \cO_{\mathrm{dP}_2}(D-(D\cdot \cM_1)\,\cM_1)\big) \\
\textrm{Region 2} &~ \ind\big({\rm dP}_2, \cO_{\mathrm{dP}_2}(D-(D\cdot \cM_2)\,\cM_2)\big) \\
\textrm{Region 3} &~ \ind\big({\rm dP}_2, \cO_{\mathrm{dP}_2}(D-(D\cdot \cM_3)\,\cM_3)\big) \\
\textrm{Region 4} &~ \ind\big({\rm dP}_2, \cO_{\mathrm{dP}_2}(D-(D\cdot \cM_1)\,\cM_1-(D\cdot \cM_2)\,\cM_2)\big) \\
\end{tabular}
\label{eq:dp2_formulae}
\end{equation}
Equivalently, one can capture this locally polynomial function in the single expression,
\begin{equation*}
h^0\left(\mathrm{dP}_2,\mc{O}_{\mathrm{dP}_2}(D)\right) = \mathrm{ind}\Big(\mathrm{dP}_2, \mc{O}_{\mathrm{dP}_2}\big(D + \sum_{i=1}^3 \theta( - D \cdot \cM_i ) \, (D \cdot \cM_i)\cM_i\big) \Big) \,.
\label{eq:dp2_h0_alt}
\end{equation*}
where $\theta(\,\cdot\,)$ equals one for $x \geq 0$ and zero otherwise.

\subsection{Summary of results}

The appearance of formulae as in Equation~\eqref{eq:dp2_formulae}, which is a particularly simple example of a more general phenomenon, can be explained by combinining Zariski decomposition with vanishing theorems for cohomology, as we now briefly explain.

If $D$ is an effective divisor, a theorem due to Zariski ensures that it can be uniquely decomposed as $D=P+N$, where $P$ is nef and $N$ effective, and $P$ intersects no components in the curve decomposition of $N$. In general $P$ and $N$ are rational rather than integral divisors. In the case of an integral divisor $D$, which defines an effective line bundle, the importance of Zariski decomposition for cohomology arises from the relation
\begin{equation}\label{eq:dP2}
h^0\big(S, \mc{O}_S(D)\big) = h^0\big(S, \mc{O}_S(\floor{P})\big) \,,
\end{equation}
which holds for any smooth projective surface $S$. Here the round-down divisor $\floor{P}$ is the maximal integral subdivisor of $P$, which being integral defines an effective line bundle.

A line bundle $\mc{O}_\surf(D)$ depends up to isomorphism only on the divisor class $\divcls{D}$. Different representatives $D'$ in the class will have different positive parts $P'$. However, importantly, the classes $\divcls{P}$ and $\divcls{\floor{P}}$ depend only on the class $\divcls{D}$, and, crucially, can be computed purely from intersection properties of $D$. In particular this computation requires knowledge of the Mori cone and the intersection form.

If the cohomology on the right-hand side of Equation~\eqref{eq:dP2} can be computed more easily than the left, the relation becomes practically important. An obvious example is when a theorem ensures the vanishing of the higher cohomologies of $\mc{O}_\surf(\floor{P})$, so that the zeroth cohomology is computed by the index $h^0\big(S, \mc{O}_S(\floor{P})\big) = \ind\big(S, \mc{O}_S(\floor{P})\big)$. The latter, importantly, is straightforward to compute due to the Hirzebruch-Riemann-Roch theorem. The availability of such theorems depends on the surface in question.

When $D$ is nef, $N$ is trivial and $D=P$. When $D$ is outside the nef cone, the positive part $P$ always lies on the boundary of the nef cone. In the latter case, the prescription for Zariski decomposition implies that an effective divisor gets projected $D \to P$ to a face of the nef cone. Grouping divisors according to the face onto which they get projected gives rise to `Zariski chambers', which are locally polyhedral subcones of the effective cone. Within a Zariski chamber, the support of $N$ is fixed and the Zariski decomposition takes a fixed form. The chamber structure induced on the interior of the effective cone (the big cone) by Zariski decomposition is a fairly recent result established in Ref.~\cite{Bauer04}.

If the image of a Zariski chamber under the map $\divcls{D} \to \divcls{\floor{P}}$ is covered by a vanishing theorem, then the index function can be `pulled back' to give a single function for zeroth cohomology throughout the Zariski chamber. In this case the Zariski chamber becomes also a `cohomology chamber'.

Promoting Zariski chambers to cohomology chambers requires a vanishing theorem that interacts well with the flooring. While the positive part $P$ in a Zariski decomposition is nef, there is a round-down operation in the relation $h^0(S,\cO_\surf(D))=h^0(S,\cO_\surf(\floor{P}))$, so that it is not sufficient for a vanishing theorem to apply to the nef cone.  Additionally, most vanishing theorems involve a twist by the canonical bundle, which may push $\floor{P}$ even further away from the region covered by the vanishing theorems.

\smlhdg{Cohomology formulae for complex surfaces}

\noindent While Zariski chambers exist for every smooth complex projective surface, whether these become cohomology chambers depends on the presence of appropriate vanishing theorems. Hence in this paper we consider several classes of surfaces on which there exist such vanishing theorems.

\medskip

On all generalised del Pezzo surfaces and all projective toric surfaces, we prove that the zeroth line bundle cohomology is described throughout the Picard lattice by closed-form expressions. In the case of toric surfaces, it is possible to utilise the Demazure vanishing theorem.
\begin{nnthm}
Let $\surf$ be a smooth projective toric surface, and $D$ an effective $\mbb{Z}$-divisor with Zariski decomposition $D = P + N$. Then
\be
h^0 \big(\surf, \mc O_\surf(D) \big) = \ind\big(\surf,\mc{O}_\surf(\floor{P})\big) \,.
\ee
Hence every Zariski chamber is upgraded to a cohomology chamber.
\end{nnthm}
\noindent On generalised del Pezzo surfaces, we show one can use the Kawamata-Viehweg vanishing theorem. Here we find that one should instead use the round-up $\ceil{P}$ of the positive part, rather than the round-down $\floor{P}$.
\begin{nnthm}
Let $\surf$ be a smooth generalised del Pezzo surface, and $D$ an effective $\mbb{Z}$-divisor with Zariski decomposition $D = P + N$. Then
\be
h^0 \big(\surf, \mc O_\surf(D) \big) = \ind\big(\surf,\mc{O}_\surf(\ceil{P})\big) \,.
\ee
Hence every Zariski chamber is upgraded to a cohomology chamber.
\end{nnthm}

On K3 surfaces, we find one can again use the Kawamata-Viehweg vanishing theorem. However, in this case, the combination of Zariski decomposition with a vanishing theorem gives cohomology formulae only in the interior of the effective cone. The cohomologies of those line bundles lying on the boundary are generally not determined by our methods, and require a separate treatment that we do not attempt.
\begin{nnthm}
Let $\surf$ be a smooth projective complex K3 surface, and $D$ an effective $\mbb{Z}$-divisor not on the boundary of the Mori cone with Zariski decomposition $D = P + N$. Then
\be
h^0 \big(\surf, \mc O_\surf(D) \big) = \ind\big(\surf,\mc{O}_\surf(\ceil{P})\big) \,.
\ee
Hence every Zariski chamber, excluding its intersection with the boundary of the Mori cone, is upgraded to a cohomology chamber.

On the boundary, one can at least say for the subset of integral divisors $D'$ whose support has negative definite intersection matrix that the positive part is trivial $P'=0$ so that $h^0\big(\surf,\mc{O}_\surf(D')\big) = h^0\big(\surf,\mc{O}_\surf\big) = 1$. In general this determines the cohomology on a number of faces of the Mori cone but not the entire boundary.
\end{nnthm}

The expressions for $P$, and hence $\floor{P}$ and $\ceil{P}$, can be made very explicit, given knowledge of the Mori cone and the intersection form, and in particular are determined purely from intersection properties. Since the index is also computed from intersections, this means the calculation of any zeroth cohomology involves only intersection computations.

Concretely, note the prescription for constructing Zariski chambers is that every face $F$ of the nef cone not contained in the boundary of the Mori cone gives rise to a Zariski chamber $\Sigma_F$, by translating the face $F$ along the Mori cone generators which have zero intersection with divisors on the face (with respect to the intersection form). Then one has the following.

\begin{nnprp}
Let $D$ be an effective divisor, within a Zariski chamber $\Sigma_{i_1,\ldots i_n}$ obtained by translating a codimension $n$ face of the nef cone along the set of  Mori cone generators $R = \{\cM_{i_1},\cM_{i_2},\ldots \cM_{i_n}\}$ orthogonal (with respect to the intersection form) to the face. The positive part $P$ in the Zariski decomposition of $D$ is given by 
\be
P = D - \sum_{k=1}^n\, (-D\cdot \cM_{i_k , R}^\vee) \,\cM_{i_k}  \,,
\ee
where the dual $\cM_{i_k,R}^\vee$ is an effective divisor with support $R$ defined such that $ \cM_{i_k}^\vee\cdot  \cM_{i_m} = -\delta_{km}$, $\forall \cM_{i_m} \in R$.
\end{nnprp}
\noindent Note that the dual divisor $\cM_{i_k,R}^\vee$ is computed with respect to the set $R = \{\cM_{i_1},\cM_{i_2},\ldots \cM_{i_n}\}$ and so can take different forms in different Zariski chambers. When $D$ is integral as in the case of line bundles, the round-up and round-down are then given by the following simple expressions
\be
\floor{P} = D - \sum_{k=1}^n\, \ceil{-D\cdot \cM_{i_k,R}^\vee} \,\cM_{i_k}  \quad \mathrm{and} \quad \ceil{P} = D - \sum_{k=1}^n\, \floor{-D\cdot \cM_{i_k,R}^\vee} \,\cM_{i_k}  \,.
\ee

We also show that, alternatively, one can write a single expression for $P$ throughout the effective cone. Let $\mc{I}(\surf)$ be the set of rigid curves on the surface $\surf$, which is a subset of the set of Mori cone generators. And let $\mc{R}(\surf)$ be the set of subsets of $\mc{I}(\surf)$ with negative definite intersection form. Every subset $R\in\mc{R}(\surf)$ corresponds to a set of generators of the Mori cone orthogonal to a face of the nef cone. In a given subset $R \in \mc{R}(\surf)$, for any element $\cM_i \in R$ one can define a unique effective dual divisor $\cM_{i,R}^\vee$ as above. Each element $R \in \mc{R}(\surf)$ with $\cM_i \in R$ determines a dual divisor $\cM_{i,R}^\vee$. Defining $\mc{G}_i(D) = \{-\cM_{i,R}^\vee \cdot D \, | \, R \in \mc{R}(\surf) \} \cup \{ 0 \}$, one then has the following.
\begin{nnprp}
Let $D$ be an effective divisor on $\surf$ with Zariski decomposition $D = P + N$ Then
\be
P = D - \sum_{\cM_i \in \mc{I}(S)} \mathrm{max}\, (\mc{G}_i(D)) \, \cM_i \,.
\ee
\end{nnprp}

At a practical level, to determine Zariski decompositions one requires knowledge of the subsets of the Mori cone generators on which the intersection form restricts to a negative definite matrix. These are the subsets of generators orthogonal to those faces of the nef cone that intersect the interior of the Mori cone, and hence directly determine the Zariski chambers. The subsets are straightforward to compute given knowledge of the Mori cone and the intersection form.

While for generalised del Pezzo surfaces and toric surfaces the Mori cone data can be computed algorithmically, in general this is not an easy matter. In the cases where the Mori cone data is not easily available, one can attempt to use the cohomology formulae described above `backwards'. The proposal is that one would start with some partial knowledge of the zeroth cohomology, as determined from algorithmic methods, and then attempt to fit these results to the formulae in order to infer the Mori cone data.   

\medskip

We note again that, while the above framework applies only to the zeroth cohomology, formulae for the first and the second cohomology follow immediately via the index formula and Serre duality. 

\smlhdg{Cohomology formulae for elliptically fibered Calabi-Yau three-folds}

\noindent Elliptically fibered Calabi-Yau three-folds are of particular significance in string theory, especially in the study of heterotic/F-theory duality (see Ref.~\cite{Braun:2017feb, Braun:2018ovc} for some recent work on this duality involving line bundles). Thus, in the second part of the paper we consider smooth elliptic Calabi-Yau three-folds realised as generic Weierstrass models with a single section over smooth compact two dimensional bases. The aim is to lift the cohomology formulae obtained for surfaces to the corresponding three-folds.

On such a three-fold $\cy{3}$, the cohomology of any line bundle $\lb$ can be computed in terms of the cohomology of the pushforward bundle $\fibprj_\ast\lb$ and the higher direct image $R^1\fibprj_\ast\lb$ under the projection map $\fibprj: \cy{3} \to \base{2}$ to the base $\base{2}$, by use of the Leray spectral sequence. We show that this sequence degenerates in our context, so that the lift of cohomology on the base to the three-fold is simply
\be
\begin{aligned}
h^0(\cy{3},\lb) &= h^0(\base{2},\fibprj_* \lb) \,, \\
h^1(\cy{3},\lb) &= h^1(\base{2},\fibprj_* \lb) + h^0(\base{2},R^1\fibprj_* \lb) \,, \\
h^2(\cy{3},\lb) &= h^2(\base{2},\fibprj_* \lb) + h^1(\base{2},R^1\fibprj_* \lb) \,, \\
h^3(\cy{3},\lb) &= h^2(\base{2},R^1\fibprj_* \lb) \,.
\end{aligned}
\label{eq:3fold_coh}
\ee
The pushforward and higher direct image are simple sums of line bundles, written explicitly in Equation~\eqref{eq:hidirim_list}.

From these formulae, one can expect that the cohomology chambers of the base give rise on the three-fold to regions in which the cohomology function has a closed form. We study this phenomenon in detail for an elliptic fibration over the simplest base, $\IP^2$.
On the one hand, we show that it is indeed possible to determine regions and corresponding formulae describing all line bundle cohomologies on the Calabi-Yau three-fold. On the other hand, we make the point that this procedure is intricate, and not immediately transparent. Nevertheless, this provides the first proofs of cohomology formulae for three-folds of this kind.

\newpage
\section{Zariski decomposition}

In this section we give a pedagogical introduction to Zariski decomposition. The reader familiar with the terminology and the basic ideas can safely skip to the following section. 

\label{sec:zar_dec}

\subsection{Basic notions}
\label{sec:bas_defs}

\smlhdgnogap{Divisors}

\noindent We start by reviewing some definitions involving divisors. Since we are dealing only with smooth projective surfaces, we will not distinguish between Weil and Cartier divisors. 
The group of divisors on a surface $S$ is denoted by ${\rm Div}(S)$. A divisor $D\in {\rm Div}(S)$ is a $\mathbb Z$-linear combination of irreducible codimension one subvarieties (irreducible curves), that is a finite sum $D=\sum_i n_i C_i$ with $n_i \in \mbb{Z}$, and the group operation is addition. The set of curves $\{C_i\}$ is called the support of $D$, which we denote by ${\rm Supp}(D)$. $D$~is said to be effective if $n_i\geq 0$ for all $i$. A subdivisor $P$ of $D$ is a divisor such that $D-P$ is effective.

Two divisors $D$ and $D'$ are said to be linearly equivalent $D \lineq D'$ if they differ by the divisor of a meromorphic function $f$, i.e.\ $D-D' = \mathrm{div}\,f = \sum_i \mathrm{ord}_i C_i$ where $\mathrm{ord}_i$ is the vanishing order (positive) or the pole order (negative) of $f$ on the curve $C_i$. Note that the divisor of a product of meromorphic functions is $\mathrm{div}(fg) = \mathrm{div}(f) + \mathrm{div}(g)$. The class of a divisor $D$ modulo linear equivalence is denoted by~$\divcls{D}$. A linear equivalence class is said to be effective if it contains effective representatives.

There is also the related notion of a complete linear system of a divisor, denoted $\cls{D}$, which is the set of all effective divisors linearly equivalent to $D$, which can of course be empty. If $D$ is effective, one can think of $\cls{D}$ as the family of deformations of $D$, and its dimension $\mathrm{dim}\cls{D}$ as the number of parameters of the family. If $D$ is effective but the only element in its complete linear system, then $\mathrm{dim}\cls{D}=0$ and $D$ is called `rigid'. The set of points common to every element of the complete linear system is called the base locus.

The group ${\rm Div}(S)$ can be extended to ${\rm Div}(S)\otimes \mathbb Q$, whose elements are called $\mathbb Q$-divisors. These are rational linear combinations of curves. Two $\mbb{Q}$-divisors $D_1$ and $D_2$ are said to be linearly equivalent if there exists an integer $n$ such that $nD_1$ and $nD_2$ are integral and linearly equivalent. Elements of ${\rm Div}(S)$ will be referred to as integral or $\mathbb Z$-divisors. An $\mathbb R$-divisor is a $\mathbb Q$-divisor multiplied by some real number.

\smlhdg{Divisors and line bundles}

\noindent A divisor $D\in {\rm Div}(S)$ determines a line bundle $\mathcal O_S(D)$ such that $D$ is a rational section of $\mathcal O_S(D)$. If two divisors are linearly equivalent, their associated line bundles are isomorphic. Hence the group of divisors modulo linear equivalence is isomorphic to the group of line bundles up to bundle isomorphisms, which is called the Picard group $\mathrm{Pic}(S)$. Note in particular that the operation of adding divisors corresponds to taking the tensor product of the line bundles, $\mc{O}_\surf(D_1 + D_2) = \mc{O}_\surf(D_1) \otimes \mc{O}_\surf(D_2)$. Below we will be interested only in line bundles up to isomorphism, so we will simply refer to a `line bundle' when we mean a line bundle up to isomorphism, and we will write $\mc{O}_S(D)$ rather than $\mc{O}_S(\divcls{D})$.

Particularly important for our purposes is the simple relationship between the zeroth cohomology of the line bundle $\mc{O}_\surf(D)$ and the complete linear system of the divisor $D$, specifically $\cls{D} \cong \mbb{P}\left(H^0\big(\surf,\mc{O}_\surf(D)\big)\right)$, where $\mbb{P}(\,\cdot\,)$ denotes the projectivisation.

\smlhdg{Intersections}

\noindent If two curves $C$ and $C'$ intersect transversely, there is a natural intersection product given by the number $\# \{ C \cap C' \}$ of intersection points of the two subvarieties. More generally, if two curves do not share connected components, the geometric interpretation is still valid if the intersection points are weighted by the local intersection multiplicities (greater than 1 for non-transversal intersections).  The product can be extended to include curves sharing connected components, by requesting that the following conditions are met:
\begin{enumerate}
\item Consistency with the natural case: $C \cdot C' = \# \{C \cap C'\}$ if $C$ and $C'$ intersect transversely.
\item Symmetry: $C \cdot C' = C' \cdot C$.
\item Linearity: $C \cdot (C'+C'') = C \cdot C' + C \cdot C''$.
\item Invariance under linear equivalence: $C \cdot C' = C \cdot C''$ if $C' \lineq C''$.
\end{enumerate}
These conditions give a unique intersection product $C \cdot C' \in\mathbb Z$. In particular, the intersection of two curves $C$ and $C'$ sharing a connected component is understood by replacing $C'$ by a linearly equivalent sum of curves that share no connected components with $C$. In this way, negative intersections naturally occur. Suppose a curve $C$ is linearly equivalent to a distinct curve or an effective sum $D$ of curves that shares no connected components with $C$. Then its self-intersection is nonnegative, $C^2 = C \cdot D \geq 0$. Hence, conversely, if $C^2 < 0$, then there must be no distinct effective divisor linearly equivalent to $C$. So one can conclude that there are no other elements in the complete linear system $\cls{C}$, i.e.\ $C$ is rigid. For this reason we also refer to a rigid curve as a `negative' curve.

\medskip

As divisors are linear combinations of curves, the above defines intersections between $\mbb{R}$-divisors. This gives rise to an important equivalence relation on $ {\rm Div}(S)\otimes \mathbb R$: two divisors $D_1$ and $D_2$ are called `numerically equivalent' if $D_1\cdot C = D_2\cdot C$ for every curve $C$ in $S$. Note that by the third condition above, linearly equivalent divisors are also numerically equivalent, so numerical equivalence is in general a weaker condition. In particular, note there is an intersection pairing $\mathrm{Pic}(S) \times \mathrm{Pic}(S) \to \mbb{Z}$.

On many common spaces, linear equivalence and numerical equivalence coincide. For instance, this is true on all compact toric varieties (see Proposition~6.3.15 in Ref.~\cite{cox2011toric}), on all generalised del Pezzo surfaces, and on all projective Calabi-Yau manifolds of dimension greater than one (where $X$ being Calabi-Yau is understood in the strict sense of having no holomorphic $k$-forms for $0 < k < \mathrm{dim}(X)$), and hence on all spaces we discuss explicitly below. Counter-examples to this include the elliptic curve, and products of curves of large genus.

\medskip

An $\mbb{R}$-divisor $D \in {\rm Div}(S)\otimes \mathbb R$ is said to be nef if $D\cdot C\geq 0$ for every curve $C$ in $S$. It follows that $D$ is nef if and only if $D\cdot C_j\geq 0$ for every $C_j\in {\rm Supp}(D)$ since the intersection of distinct curves is non-negative. For a divisor $D=\sum x_i C_i$, with $C_i$ irreducible components, its intersection matrix $I(D)$ is defined as the symmetric matrix with $(i,j)$ entry $(C_i\cdot C_j)$, hence $D\cdot D = {\bf x}^T I(D){\bf x}$. 

\smlhdg{Cones}

\noindent The natural arena for defining several important objects is the space of divisors modulo numerical equivalence. This is called the N\'{e}ron-Severi group $\mathrm{NS}(\surf)$, and we can define the corresponding real vector space $\mathrm{NS}(\surf)_\mathbb{R} = \mathrm{NS}(\surf) \otimes \mathbb{R}$. Within this vector space, the set of nef divisors naturally forms a cone $\mathrm{Nef}(\surf)$. To any cone one can associate a dual cone, which is the set of points having non-negative intersection with every element in the cone. The dual of the nef cone is the closure of the cone of effective divisors, and is called the Mori cone or the cone of pseudo-effective divisors, and is denoted by $\overline{\mathrm{NE}}(\surf)$. The interior of the Mori cone is the big cone, whose elements are big divisors. 

We note it is easy to see that a rigid curve must be a generator of the Mori cone as follows. Let $C$ be a rigid curve and consider the hyperplane in the N\'{e}ron-Severi group corresponding to zero intersection with $C$. Any other Mori cone generator, being a distinct curve, must have non-negative intersection with $C$, and hence must lie on the hyperplane or be on the positive side of it. But $C$ is on the negative side since $C^2<0$. Since $C$ is effective, this is impossible unless $C$ is also a generator.

Since linear equivalence and numerical equivalence coincide on the spaces we will discuss, the N\'{e}ron-Severi group and the Picard group are isomorphic. Hence integral points in the above cones can be identified with line bundles up to isomorphism.

\smlhdg{Detection of rigid divisors}

\noindent An important idea in relation to Zariski decomposition is that of detecting via intersections rigid parts of a complete linear system. Suppose an effective divisor $D$ has negative intersection $D \cdot C_j < 0$ with an irreducible curve $C_j$. In the intersection
\be
D \cdot C_j = \bigg( \sum_i x_i C_i \bigg) \cdot C_j \,,
\ee
the only possible negative contribution is from the self-intersection term $x_j (C_j \cdot C_j)$. Hence $C_j$, which must be a rigid curve, must be in the divisor expansion of $D$, i.e.\ $x_j > 0$. More strongly, there is clearly a lower bound on the coefficient $x_j$ of $C_j$,
\be
x_j \geq \frac{D \cdot C_j}{C_j \cdot C_j} \,.
\ee
Any linearly equivalent divisor has the same intersection with $C_j$, and hence for any effective $D' \lineq D$, i.e.\ any element of the linear complete system $\cls{D}$, the same lower bound applies. In particular, removing this much of $C_j$ from every divisor in the complete linear system $|D|$ leads to a linear system of equal size
\be
{\rm dim}\cls{D} ~=~ {\rm dim} \cls{D- \frac{D \cdot C_j }{C_j \cdot C_j}C_j} \,.
\ee

More generally, if $C_j$ is an irreducible negative divisor and $\tilde{D}_j$ an effective divisor that (1) intersects $C_j$ negatively and (2) has non-negative intersection with all other irreducible curves, then $\tilde{D}_j$ can be used in order to detect the presence of $C_j$ in the expansion of $D$, provided that $D \cdot \tilde{D}_j < 0$. As before, it follows that
\be
{\rm dim}\cls{D} ~=~ {\rm dim} \cls{D- \frac{D \cdot \tilde{D}_j }{C_j \cdot \tilde{D}_j}C_j} \,.
\ee

\subsection{Zariski decomposition}
\label{sec:zar_dec_subsec}

\noindent In Ref.~\cite{Zariski}, Zariski established the following result.

\vspace{5pt}
\begin{thm}[Zariski decomposition]
\label{thm:Zariski}
Let $D$ be an effective $\mathbb Q$-divisor on a smooth projective surface $S$. Then $D$ has a unique decomposition $D = P + N$, where $P$ and $N$ are $\mathbb Q$-divisors such that 
\begin{itemize}
\item[1.] $P$ is nef. 
\item[2.] $N$ is effective and if $N \neq 0$ then it has negative definite intersection matrix
\item[3.] $P \cdot C_i = 0$ for every irreducible component $C_i$ of $N$.
\vspace{-4pt}
\end{itemize}
\end{thm}
\vspace{8pt}

\noindent Zariski decomposition was extended to pseudo-effective divisors by Fujita in Ref.~\cite{fujita1979}. While $N$ is effective, $P$ is only pseudo-effective in general. Moreover, recalling that an $\mathbb R$-divisor $D_{\mathbb R}$ is a $\mathbb Q$-divisor $D_{\mathbb Q}$ multiplied by some real number $a$, the Zariski decomposition of $D_{\mathbb R}$ can be defined by $D_{\mathbb R} = a\, P_{D_{\mathbb Q}} +a\, N_{D_{\mathbb Q}}$.

\begin{defn}
The subdivisors $P$ and $N$ in Theorem~\ref{thm:Zariski} are called the `positive' and `negative' parts of the divisor $D$, respectively. 
\end{defn}

\begin{prp}\label{prop:Z}
The following properties hold in Zariski decomposition.
\bi
\item[(Z1)] If $D$ and $D'$ are numerically equivalent, $D \numeq D'$, then $N = N'$. 
\item[(Z2)] If $D$ and $D'$ are linearly equivalent, $D \lineq D'$, then $P \lineq P'$.
\ei
\end{prp}
\begin{proof}
Property~(Z1) is clear as follows. Since $D - N = P$ and $D \numeq D'$, we have $D' - N \numeq P$. But then $D' = (D' - N) + N$ satisfies the requirements to be a Zariski decomposition of $D'$. Since the decomposition is unique, $N = N'$. 
Property~(Z2) follows in a similar way: since linear equivalence implies numerical equivalence, by the same argument $D \lineq D'$ implies $N=N'$ and further $P \lineq P'$.
\end{proof}

\noindent The latter Property~(Z2) implies that Zariski decomposition determines a map between linear equivalence classes of effective divisors, $\divcls{D} \to \divcls{P}$. The following table summarises the extent to which the linear equivalence class $\divcls{D}$ and numerical equivalence class $\divcls{D}_\mathrm{num.}$ of the divisor $D$ determine the negative and positive parts in its Zariski decomposition.
\be
\begin{tabular}{L |C C C C}
						& N 			& \divcls{P}_\mathrm{num.} 	& \divcls{P}		 	& ~~~P~~~ \\
\hline
\divcls{D}_\mathrm{num.} 	& \checkmark & \checkmark 				& \text{\sffamily x} 	& \text{\sffamily x} \\
\divcls{D}			 	& \checkmark & \checkmark 				& \checkmark 			& \text{\sffamily x} \\
D			 			& \checkmark & \checkmark 				& \checkmark 			& \checkmark \\
\end{tabular}
\ee
While in general the numerical equivalence class $\divcls{D}_\mathrm{num.}$ of $D$ does not determine the linear equivalence class $\divcls{P}$ of $P$, on the classes of surfaces that we consider linear and numerical equivalence coincide, so $\divcls{D}_\mathrm{num.}$ does determine $\divcls{P}$.

\smlhdg{A pedagogical algorithm}

\noindent In Section~\ref{sec:zar_ch} we will present a way to implement Zariski decomposition with a simple formula, which will be the basis of the cohomology discussion in Section~\ref{sec:bun_coh}. In the present section, we present a pedagogical iterative algorithm, based on a classical proof for Zariski's theorem - see for example Theorem~14.14 in Ref.~\cite{badescu2001algebraic}.

The algorithm begins with a naive guess of the support for the negative part $N$, as detected by negative intersections. A candidate Zariski decomposition is then constructed. However there then appear new negative intersections, so the process is iterated.

The following steps lead to the unique Zariski decomposition of an effective $\mbb{Q}$-divisor $D$. Let $\mathcal I(S)$ denote the set of all irreducible negative divisors on $S$. And set $I = \varnothing$.
\begin{enumerate}
\item Determine the set of curves $\{ C\in \mathcal I(S) \, | \, C \cdot D < 0 \}$. This set is non-empty, unless $D$ is nef, in which case its Zariski decomposition is trivial. Incorporate these into the set $\tilde{I} = I \cup \{ C\in \mathcal I(S) \, | \, C \cdot D < 0 \}$.
\item Construct the unique, effective $\mbb{Q}$-divisor $\tilde{N}$ with support $\tilde{I}$ such that $\tilde N \cdot C_i = D \cdot C_i$ for all $C_i \in \tilde{I}$.
\item Define $\tilde{P} := D-\tilde{N}$. If this is nef, take $P = \tilde{P}$ and $N = \tilde{N}$. Otherwise, repeat the first two steps with $I = \tilde{I}$ and $D = \tilde{P}$. 
\end{enumerate}

\noindent The algorithm must terminate because each iteration increases the size of the set $\tilde{I}$, while $\tilde{I}$ is finite since $\mathrm{Supp}(D)$ is finite and $\tilde{I} \subset \mathrm{Supp}(N) \subset \mathrm{Supp}(D)$. The uniqueness and effectiveness of $\tilde{N}$ at each stage follow respectively from Lemmas~14.12 and 14.9 of Ref.~\cite{badescu2001algebraic}. 

\smlhdg{Example}

\noindent Consider the Gorenstein Fano toric surface $F_8$, whose ray diagram is depicted in \fref{fig:F8}. 

\begin{figure}[h]
\begin{center}
\raisebox{0in}{\includegraphics[width=3cm]{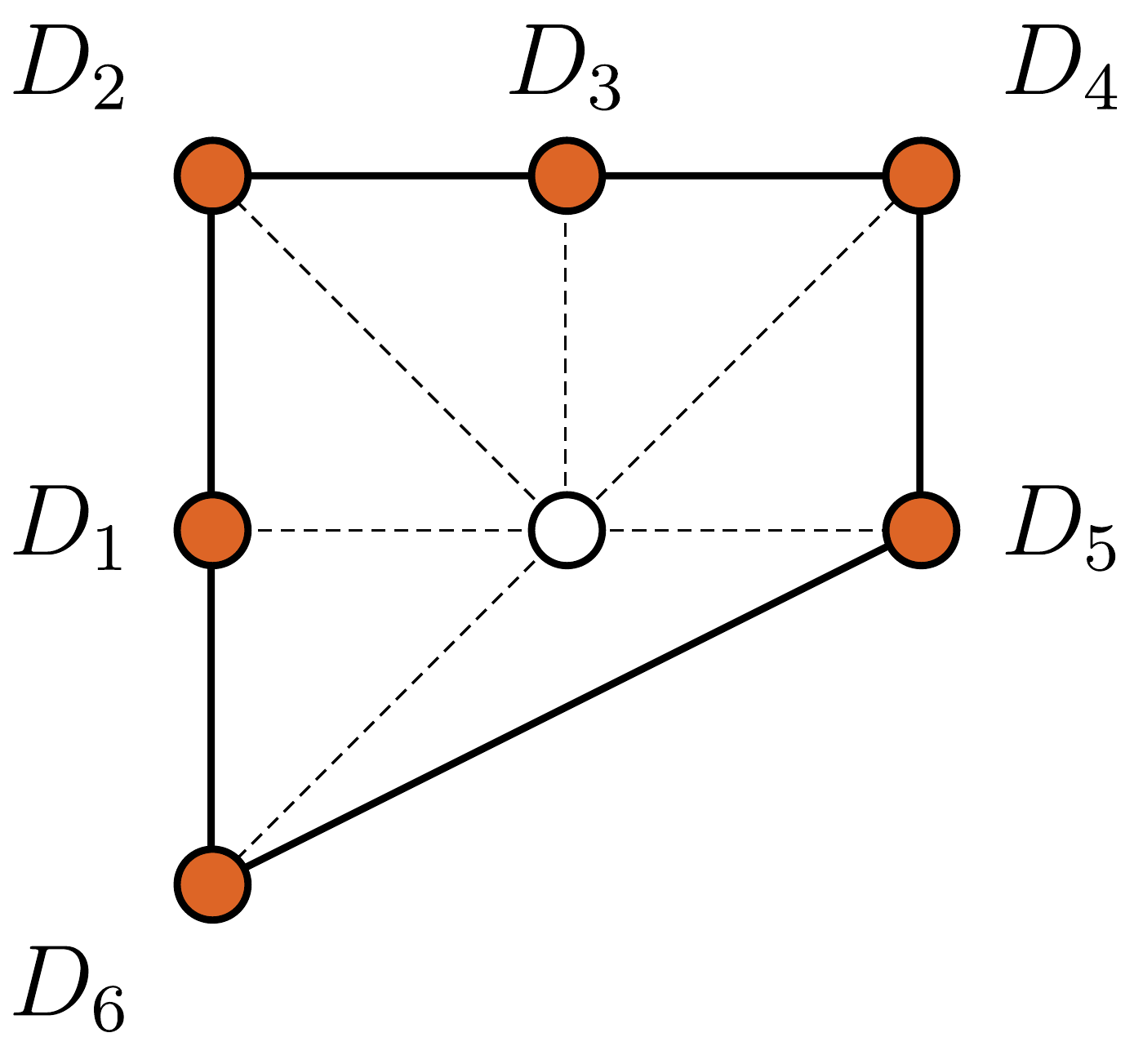}}
\caption{\itshape\small The ray diagram for the Gorenstein Fano toric surface $F_8$.}
\label{fig:F8}
\end{center}
\end{figure}

In Appendix~\ref{app:tor_surf} we provide a reminder on how to compute important properties of a toric surface from its ray diagram, which we apply here for the case of $F_8$. The diagram shows $6$ rays ${\bf v}_1,\dots, {\bf v}_6$, corresponding to toric divisors $D_1,\ldots, D_6$ labelled in the  order
\be
D_1: (-1, 0) \,, ~~ D_2: (-1, 1) \,, ~~ D_3: (0, 1) \,, ~~ D_4: (1, 1) \,, ~~ D_5: (1, 0) \,, ~~ D_6: (-1, -1) \,.
\ee
There are four linear relations between the six rays leading to the following weight system: 
\begin{equation}
\begin{array}{ccc}
\begin{array}{l}
\\[2pt]
{\bf v}_1 + {\bf v}_5 = {\bf 0} \\[0pt]
{\bf v}_2+ 2{\bf v}_5+{\bf v}_6 = {\bf 0} \\[0pt]
{\bf v}_3+{\bf v}_5+{\bf v}_6 = {\bf 0} \\[0pt]
{\bf v}_4+{\bf v}_6 = {\bf 0}
\end{array} & 
~~~~~~~~~~&
\begin{array}{cccccc}
D_1~&~D_2~&~D_3~&~D_4~&~D_5~&~D_6\\[2pt]
\hline
1& 0 &0&0&1&0\\[0pt]
0& 1 &0&0&2&1\\[0pt]
0& 0 &1&0&1&1\\[0pt]
0& 0 &0&1&0&1
\end{array}
\end{array}
\end{equation}
The toric divisors $D_1, D_2,D_3,D_4$ can be used as a basis for the Picard lattice. In terms of these, the expressions for $D_5$ and $D_6$, read off from the weight system, are $D_5 = D_1 + 2D_2 + D_3$ and $D_6 = D_2 + D_3 + D_4$. The self-intersections are given by:
\be
D_1^2 = -2 \,, \quad D_2^2 = -1 \,, \quad D_3^2 = -2 \,, \quad D_4^2 = -1 \,, \quad D_5^2 = 0 \,, \quad D_6^2 = 0 \,.
\ee
In this basis we will write divisors as $D = (\,\cdot\,,\,\cdot\,,\,\cdot\,,\,\cdot\,)$. The intersection form (see Appendix~\ref{app:tor_surf} for details about how to infer intersections from the toric data) is
\be
G := (D_i \cdot D_j) = 
\begin{pmatrix}
-2 &  ~~1 &  ~~0 &  ~~0~ \\
 ~~1 & -1 &  ~~1 &  ~~0~ \\
~~ 0 &  ~~1 & -2 &  ~~1~ \\
~~ 0 &  ~~0 &  ~~1 & -1~
\end{pmatrix}
\,.
\ee
The Mori cone generators $\moricn_i$ and the dual nef cone generators $\nefcn_j$ are given by
\be
\begin{aligned}
\{\moricn_i\} &= \{ (1,0,0,0) \,, ~ (0,1,0,0) \,, ~ (0,0,1,0) \,, ~ (0,0,0,1) \} \,, \\[4pt]
\{\nefcn_i\} &= \{(0, 1, 1, 1)\,, ~(1, 2, 2, 2)\,, ~(1, 2, 1, 1)\,, ~(1, 2, 1, 0)\} \,.
\end{aligned}
\ee
For applying the algorithm above, we need the set of rigid irreducible curves. In the present case these are simply the Mori cone generators.

With $D = 2D_1+D_2+D_3$, let us apply the above algorithm to find its Zariski decomposition. Applying the steps of the algorithm presented above we have the following.
\begin{enumerate}
\item We note $D \cdot \{ D_1 \,, \ldots \,, D_4 \} = \{ -3, 2, -1, 1 \}$, so that $\mathrm{Supp}(\tilde{N}) = \tilde{I} = \{D_1 \,, D_3 \}$.
\item The conditions $\tilde{N} \cdot D_i = D \cdot D_i$ for $D_i \in \mathrm{Supp}(\tilde{N})$ then uniquely fix $\tilde{N} = \frac{3}{2}D_1 + \frac{1}{2}D_3$.
\item Define $\tilde{D} := D - \tilde{N} = \frac{1}{2}D_1 + D_2 + \frac{1}{2}D_3$. Noting $\tilde{D} \cdot \{ D_1 \,, \ldots \,, D_4 \} = \{ 0, 0, 0, \frac{1}{2} \}$, we are done as $\tilde{D}$ is nef.
\end{enumerate}
The Zariski decomposition of $D$ is hence
\be\label{eq:DF8example}
D = 2D_1 + D_2 + D_3 = \left(\frac{1}{2}D_1 + D_2 + \frac{1}{2}D_3\right) + \left(\frac{3}{2}D_1 + \frac{1}{2}D_3\right) \equiv P + N \,.
\ee

\subsection{Zariski decomposition for line bundles}
\label{sec:zar_lb}

\smlhdgnogap{Map on line bundles}

\noindent We know from property (Z2) in Proposition~\ref{prop:Z} that Zariski decomposition determines a map between effective divisor classes, $\divcls{D} \to \divcls{P}$. However, while for a line bundle the corresponding class $\divcls{D}$ is integral,  the class $ \divcls{P}$ is in general not integral. In view of the cohomology result~\eqref{eq:zar_coh_rel} below, we define the round-down version of the $\mbb{Q}$-divisor $P$ as the $\mathbb Z$-divisor $\floor{P}$, obtained by rounding down each coefficient in the divisor expansion of $P$. The round-up $\ceil{P}$ of a divisor is defined analogously. We prove the following result.

\begin{prp}
Let $\surf$ be a smooth projective surface. If $D$ and $D'$ are two linearly equivalent integral divisors on $\surf$, the round-down versions of their positive parts, $\floor P$ and $\floor {P'}$ are also linearly equivalent. The same is true for the round-up versions $\ceil P$ and $\ceil {P'}$
\end{prp}
\begin{proof}
Let $D=P+N$ and $D' = P'+N'$ be the  Zariski decompositions of the two given linearly equivalent divisors. Since $D$ and $D'$ are integral, any floor or ceiling operations have no effect. Hence 
$\floor{P} = D - \ceil{N}$ and $\floor{P'} = D' - \ceil{N'}$.
Since $D$ and $D'$ are linearly equivalent, $N = N'$ by (Z1), and hence $\ceil{N} = \ceil{N'}$, which implies $\floor{P}$ and $\floor{P'}$ are linearly equivalent. Clearly the same argument applies for $\ceil{P}$ and $\ceil{P'}$.
\end{proof}

The result implies that it is possible to define maps
\begin{equation}
\begin{gathered}
\begin{aligned}
\phifl\colon\overline{{\rm NE}}(S)\cap {\rm NS}(S)&\rightarrow \overline{{\rm NE}}(S)\cap {\rm NS}(S) \\
\phifl\colon [D]&\mapsto[\floor{P}]
\end{aligned}
\vspace{.2cm} \\
\begin{aligned}
\phicl\colon \overline{{\rm NE}}(S)\cap {\rm NS}(S)&\rightarrow \overline{{\rm NE}}(S)\cap {\rm NS}(S)\\
\phicl\colon [D]&\mapsto[\ceil{P}]
\label{eq:phi_map}
\end{aligned}
\end{gathered}
\end{equation}
between effective integral linear equivalence classes (line bundles), where the classes $\divcls{\floor{P}}$ and $\divcls{\ceil{P}}$ are constructed by choosing any integral effective representative $D'$ of the class $[D]$, followed by taking its positive Zariski part $P'$, then rounding down or up to $\floor{P'}$ or $\ceil{P'}$, and finally going to the linear equivalence class~$\divcls{\floor{P'}}$ or $\divcls{\ceil{P'}}$.

\smlhdg{Preservation of zeroth cohomology}

\noindent Let $D=P+N$ be an effective integral divisor and $\cO_S(D)$ the line bundle associated to $[D]$. Since the negative part of $D$ is determined by intersection properties alone, $N$ is a subdivisor of every element of the complete linear system $|D|$, which implies that $\cls{D}$ and $\cls{P}$ have the same dimension. Moreover, since the complete linear system of $D$ contains only integral effective divisors, it follows that the round-up, $\ceil{N}$, must be a subdivisor of every element of $\cls{D}$. Consequently, 
\be
\cls{D} \cong \cls{\floor{P}} \,.
\ee
Equivalently, since in the divisor line bundle correspondence there is the isomorphism $\cls{D} \cong \mbb{P}\left(H^0\big(\surf,\mc{O}_\surf(D)\big)\right)$, we can say that Zariski decomposition provides a map on line bundles that preserves the zeroth cohomology. This is summarised in the following theorem.
\begin{thm}
\label{thm:zar_coh_rel}
Let $\surf$ be a smooth projective surface, and let $D$ be an effective $\mbb{Z}$-divisor with Zariski decomposition $D = P + N$. Then
\be
H^0\big(S, \mc{O}_S(D)\big) \cong H^0\big(S, \mc{O}_S(\floor{P})\big) \,.
\label{eq:zar_coh_rel}
\ee
\end{thm}
\begin{proof}
See Proposition~2.3.21 in Ref.~\cite{lazarsfeld2004positivity1}.
\end{proof}

\vspace{21pt}
In fact, it is straightforward to see that the same result  applies in the case of the round-up $\ceil{P}$.
\begin{crl}\label{crl:zar_coh_rel_cl}%
Let $\surf$ be a smooth projective surface, and let $D$ be an effective $\mbb{Z}$-divisor with Zariski decomposition $D = P + N$. Then
\be
H^0\big(\surf,\mc{O}_\surf(D)\big) \cong H^0\big(\surf,\mc{O}_\surf(\ceil{P})\big) \,.
\ee%
\end{crl}%
\begin{proof}
The ceiling $\ceil{P}$ of the positive part is related to $P$ by a fractional effective divisor, i.e.\
\be
\ceil{P} = P + \sum_i f_i D_i \equiv P + \Delta \,,
\ee
where $0 \leq f_i < 1$. Importantly, since $D$ is integral, $\mathrm{Supp}(\Delta) \subseteq \mathrm{Supp}(N)$. From the properties of Zariski decomposition, recalled in Section~\ref{sec:zar_dec_subsec}, it is then trivial to verify that the above expression for $\ceil{P}$ is in fact a Zariski decomposition, with positive part $P$ and negative part $\Delta$, so that in particular $D$ and $\ceil{P}$ have the same positive part $P$. But Theorem~\ref{thm:zar_coh_rel} then implies both $H^0\big(S, \mc{O}_S(\ceil{P})\big) \cong H^0\big(S, \mc{O}_S(\floor{P})\big)$ and $H^0\big(S, \mc{O}_S(D)\big) \cong H^0\big(S, \mc{O}_S(\floor{P})\big)$, which together establish the claim.
\end{proof}

\smlhdg{Iteration of Zariski decomposition and divisor rounding}

\noindent While the positive part $P$ in the Zariski decomposition $D=P+N$ is nef, the same is not in general true of the round-down $\floor{P}$ or the round-up $\ceil{P}$. That is, the maps $\phifl$ and $\phicl$, defined above, do not in general output in the nef cone. So it may be possible to perform a subsequent Zariski decomposition.

For simplicity we focus on the round-down $\floor{P}$. Denoting $\floor{P} = D^{(1)}$, its Zariski decomposition can be written as
\vspace{-4pt}
\begin{equation*}
D^{(1)} = P^{(1)} + N^{(1)} \,.
\vspace{-4pt}
\end{equation*}
This process can be iterated until for some $n$, $\floor{P^{(n)}}$ is nef, which includes the possibility of being zero. Equivalently, the iteration takes place as long as the Zariski decomposition gives a non-trivial negative part. The fact that such an $n$ must exist is clear, since at every iteration Zariski decomposition and flooring reduce at least one of the coefficients in the divisor expansion. 
It is useful to see a real example, which we choose from among the $16$ Gorenstein Fano toric surfaces.

\smlhdg{The $F_8$ example, once again}

\noindent Consider the Gorenstein Fano toric surface $F_8$, whose ray diagram is depicted in \fref{fig:F8} and whose properties we recalled in Section~\ref{sec:zar_dec_subsec}.

We also take again as our initial divisor $D = 2D_1+D_2+D_3$, which being integral defines a line bundle. In Section~\ref{sec:zar_dec_subsec}, we determined the Zariski decomposition of $D$ to be
\be
D = 2D_1 + D_2 + D_3 = \left(\frac{1}{2}D_1 + D_2 + \frac{1}{2}D_3\right) + \left(\frac{3}{2}D_1 + \frac{1}{2}D_3\right) \equiv P + N \,.
\ee
Since $P$ is not an integral divisor, $P \neq \floor{P}$. In particular, $\floor{P} = D_2$. Noting the intersection properties
\be
\floor{P} \cdot \{ D_1 \,, \ldots \,, D_4 \} = \{ 1 \,, -1 \,, 1 \,, 0 \} \,,
\ee
we see that $\floor{P}$ is not nef. Hence we look for another Zariski decomposition, $\floor{P} = P^{(1)} + N^{(1)}$. Applying again the algorithm of Section~\ref{sec:zar_dec_subsec}, we find straightforwardly
$\floor{P} = 0 + D_2 \equiv P^{(1)} + N^{(1)} \,.$
In this particular case, $\floor{P^{(1)}} = P^{(1)}$. Since $\floor{P^{(1)}} = 0$ is nef, the iteration process terminates here.

In terms of line bundles, the map $D \to \floor{P^{(1)}}$ of integral divisors becomes
$
\mc{O}_{F_8}(D) \to \mc{O}_{F_8} \big(\lfloor P^{(1)} \rfloor \big) = \mc{O}_{F_8} \,,
$
i.e.\ the final bundle is the trivial line bundle, which is nef. For the preserved zeroth cohomology, we have
\be
H^0\big( \mc{O}_{F_8} (D) \big) \cong H^0\big( \mc{O}_{F_8} \big( \lfloor P^{(1)} \rfloor \big) \big) \cong \mbb{C} \,,
\ee
as well as intermediate isomorphisms. It is easy to check within toric geometry that the complete linear system $\cls{D}$ indeed contains only one element.

\begin{figure}[h]
\begin{center}
\raisebox{0in}{\includegraphics[width=5.8cm]{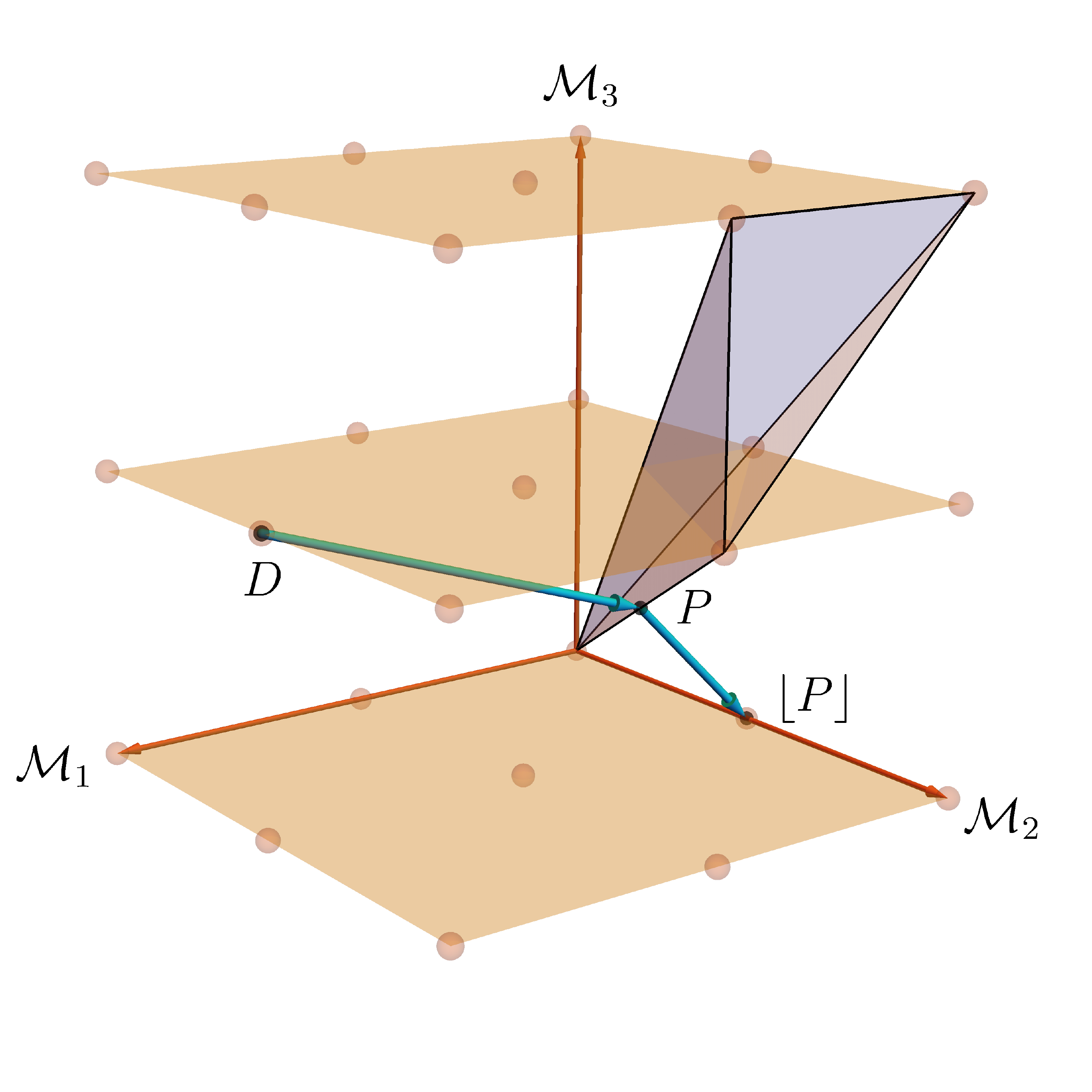}}
\hspace{1cm}
\raisebox{0.2in}{\includegraphics[width=5.cm]{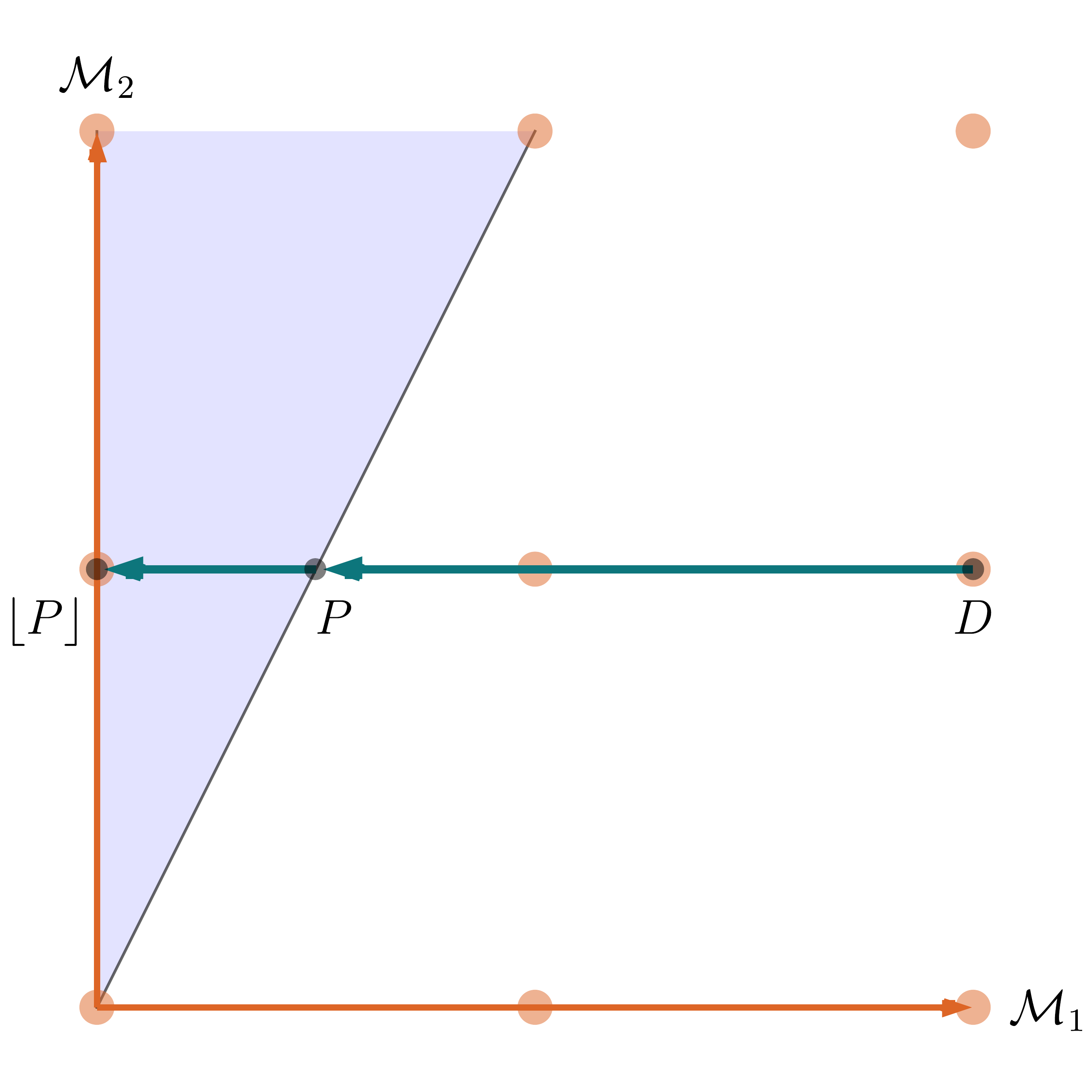}}
\capt{5.9in}{fig:flooring}{Illustration of Zariski decomposition $D \to P + N$ followed by a round-down $P \to \floor{P}$ on the Gorenstein Fano toric surface $F_8$. The blue regions corresponds to the nef cone. Lattice points correspond to effective integral divisors. Left image: projection onto the 3d subspace $(\cdot,\cdot,\cdot,0)$. Right image: projection onto the 2d subspace $(\cdot,\cdot,0,0)$.}
\end{center}
\end{figure}

\newpage
\section{Zariski chambers}
\label{sec:zar_ch}

Let $D = P + N$ be a Zariski decomposition. By varying the coefficients in $N$ while holding the support fixed, yielding effective divisors $\tilde{N}$, one obtains divisors $\tilde{D}$ whose Zariski decompositions are $\tilde{D} = P + \tilde{N}$. By keeping $P$ fixed and adding various $\tilde{N}$ with fixed support, one performs what might be called a `Zariski composition'.

 If $\mathrm{Supp}(N) \neq \varnothing$, then the positive part $P$ lies on a boundary of the nef cone. To see this, note $P \cdot C = 0$ for all curves $C \in \mathrm{Supp}(N)$. As the $C$ are rigid and hence generators of the Mori cone, $C \cdot ( \ldots ) = 0$ specifies a hyperplane which meets the nef cone along a boundary. But $P$ is nef by definition, so $P$ must lie on this boundary. Hence in a Zariski composition, one begins at a point on a boundary of the nef cone. One can then imagine varying the starting point across the entire boundary. The region reached by all such compositions will then be given by translating the entire boundary along the elements in $\mathrm{Supp}(N)$.

This perspective was formalised in Ref.~\cite{Bauer04}.  The authors showed that on any smooth projective surface the interior of the effective cone, which is the big cone, can be decomposed into rational locally polyhedral subcones called `Zariski chambers' such that in each region the support of the negative part of the Zariski decomposition of the divisors is constant. Moreover, these subcones are in one-to-one correspondence with faces of the nef cone that intersect the big cone.
\begin{defn}
Let $\surf$ be a smooth projective surface, and let $F$ denote a face of the nef cone which intersects the big cone. The Zariski chamber $\Sigma_F$ associated with $F$ is the subcone of the effective cone constructed by translating $F$ along all negative curves that are orthogonal to $F$ with respect to the intersection form, where the boundary between two chambers belongs to the chamber whose corresponding face has higher dimension.  
\end{defn}
\noindent Note that in general a Zariski chamber is a cone which is neither open nor closed. 

To construct these chambers, one requires knowledge of at least two out of three of the Mori cone, nef cone, and intersection form, which may in general be non-trivial to determine. Moreover, note that since the possible supports for the negative part of a Zariski decomposition are in one-to-one correspondence with the collections $\{C_A\}$ of rigid curves which have negative definite intersection matrix, the same is true of the faces of the nef cone that intersect the big cone. The Zariski chambers are determined by knowledge of the set of such collections, which we denote $\mc{R}(\surf)$ on a surface $\surf$. 

In Ref.~\cite{Bauer04} Zariski chambers were defined only in the interior of the effective cone, since the authors were interested in the volume properties of big line bundles. For our purposes we do not need to make this restriction. As such, we can extend Zariski chambers to the closure of the effective cone. 

\smlhdg{Map for fixed $\text{Supp}(N)$}

\noindent As we now explain, within a Zariski chamber the form of the Zariski decomposition is fixed. Let $D$ be an effective divisor with curve decomposition 
\begin{equation*}
D = \sum_{i=1}^r a_i C_i \,,
\end{equation*}
and Zariski decomposition
\begin{equation*}
D ~=~ P + N ~=~ \sum_{i=1}^r x_i C_i ~~+\!\!\! \sum_{\substack{A=1\\[1pt]C_A\in \mathcal{I}(S)}}^m y_A C_A \,,
\end{equation*}
where $C_A$ are rigid curves and the intersection matrix $(C_A\cdot C_B)$ is negative definite. When the support of $N$ is known, as throughout a Zariski chamber, the coefficients $y_A$ can be straightforwardly obtained as follows.

\begin{lma}
Let $D = P + N$ be the Zariski decomposition of an effective divisor. Then for every $C_A \in \mathrm{Supp}(N)$, the coefficient $y_A$ of $C_A$ in $N$ is given by
\begin{equation}
y_A = - D\cdot C_{A,\mathrm{Supp}(N)}^\vee \in\mathbb{Q} \,,
\label{eq:yi}
\end{equation}
where $C_{A,\mathrm{Supp}(N)}^\vee$ is the unique divisor with $\mathrm{Supp}(C_{A,\mathrm{Supp}(N)}^\vee) \subseteq \mathrm{Supp}(N)$ satisfying $C_{A,\mathrm{Supp}(N)}^\vee\cdot C_B = -\delta_{AB}$ for all $C_B \in \mathrm{Supp}(N)$.
\label{lem:ya_coeffs}
\end{lma}
\begin{proof}
Since the intersection matrix $(C_A\cdot C_B)$ is non-degenerate, it follows that $C_{A,\mathrm{Supp}(N)}^\vee$ exists and is unique. From the defining property $C_{A,\mathrm{Supp}(N)}^\vee\cdot C_B = -\delta_{AB}$, it follows that $N \cdot C_{A,\mathrm{Supp}(N)}^\vee = - y_A$. Additionally, since $\mathrm{Supp}(C_{A,\mathrm{Supp}(N)}^\vee) \subseteq \mathrm{Supp}(N)$ and $C_A \cdot P = 0$ for $C_A \in \mathrm{Supp}(N)$, one has $P \cdot C_{A,\mathrm{Supp}(N)}^\vee = 0$.
\end{proof}

The divisor $C_{A,\mathrm{Supp}(N)}^\vee$ should be read `the dual of $C_A$ with respect to the support of $N$'. We note it is a classic result in the context of Zariski decomposition that the divisor $C_A^\vee$ in Lemma~\eqref{lem:ya_coeffs} is effective (see for instance Lemma~14.9 of Ref.~\cite{badescu2001algebraic}). This lemma immediately gives the following formula. Note the rigid curves are a subset of the Mori cone generators, so we write $\cM_i$ for the elements in $\mathrm{Supp}(N)$. 

\begin{prp}
For a divisor $D$ belonging to a Zariski chamber $\Sigma_{i_1,\ldots i_n}$ obtained by translating a codimension $n$ face of the nef cone along the Mori cone generators $\cM_{i_1},\ldots \cM_{i_n}$ orthogonal to the face, the negative part of the Zariski decomposition reads
\begin{equation}
N ~=~ \sum_{k=1}^n\, (-D\cdot \cM_{i_k,\{i_1,\ldots,i_n\}}^\vee) \,\cM_{i_k} \,,
\label{eq:n_gen}
\end{equation}
where the notation $\cM_{i_k,\{i_1,\ldots,i_n\}}^\vee$ for the dual of $\cM_{i_k}$ indicates that it is computed with respect to the set $\{\cM_{i_1},\ldots \cM_{i_n}\}$.
\label{prp:n_gen}
\end{prp}
\begin{proof}
This follows immediately from Lemma~\eqref{lem:ya_coeffs}.
\end{proof}

\smlhdg{Example: Zariski chambers for the $F_6$ surface}

\noindent The space $F_6$ is a toric surface that is not isomorphic to a Hirzebruch or del Pezzo surface, and in fact it is the lowest Picard number surface of this kind among the Gorenstein Fano toric surfaces. It is isomorphic to a blow-up of the Hirzebruch surface $\mbb{F}_2$. We show the toric diagram with labelled toric divisors and the weight system:
\begin{figure}[H]
\begin{center}
{
\begin{minipage}[t]{3.9in}
\raisebox{0.in}{$\begin{tabular}{C C C C C}
D_1 	& D_2 	& D_3 	& D_4 	& D_5 	\\
\hline
 1		& 0 		& 0		& 1	& 0 		\\
0		& 1		& 0		& 2	& 1 		\\
0		& 0 		& 1		& 1	& 1 	\\
\end{tabular} $} 
\hfill \hspace*{0pt}
\raisebox{-.57in}{\includegraphics[width=3cm]{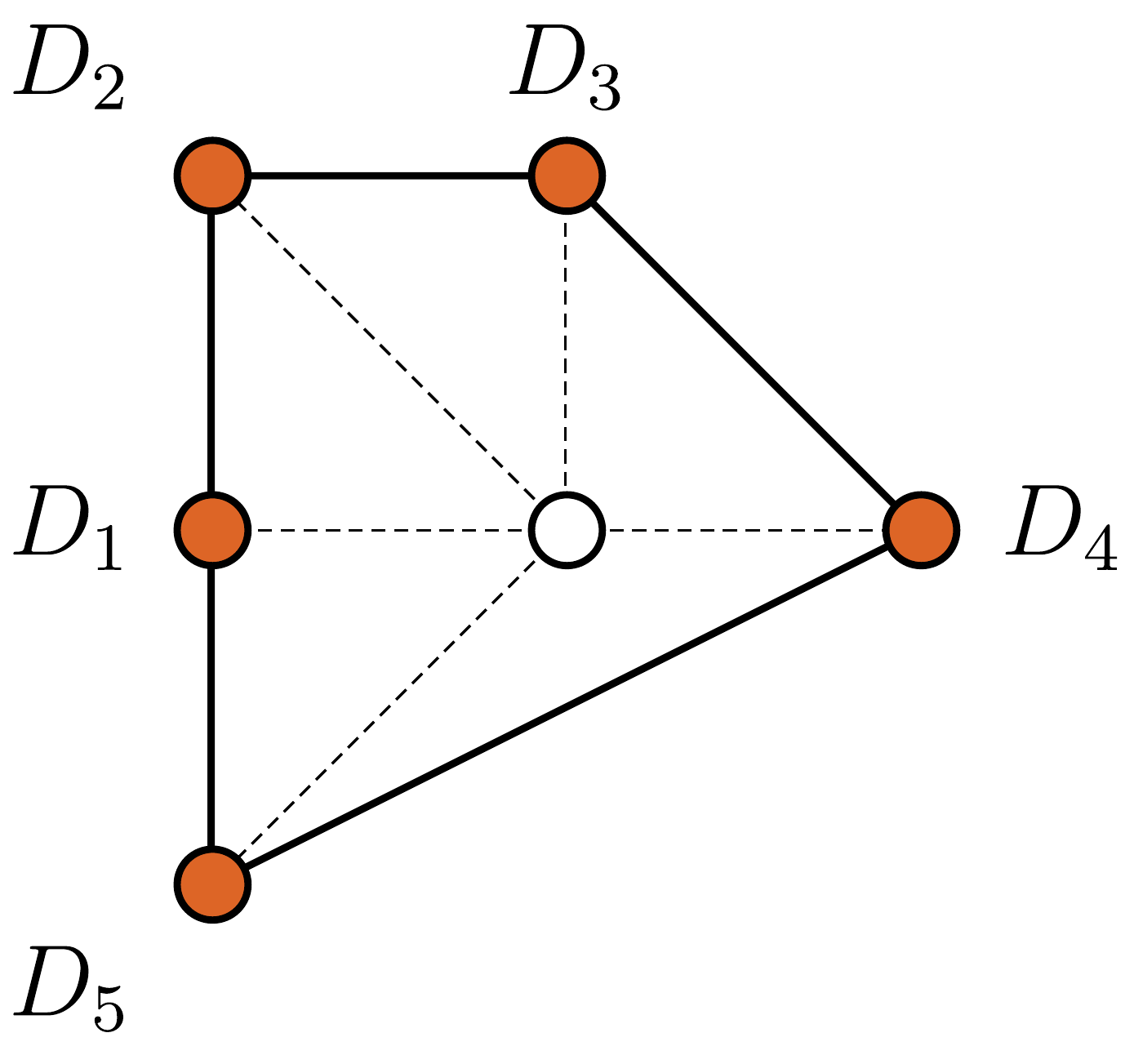}}
\end{minipage}}
\end{center}
\vspace{-12pt}
\end{figure}
\noindent We can take as a divisor basis $\{D_1, D_2, D_3\}$, in terms of which $D_4 = D_1 + 2D_2 + D_3$ and $D_5 = D_2 + D_3$. The self-intersections are given by 
\be
D_1^2 = -2 \,, \quad D_2^2 = -1 \,, \quad D_3^2 = -1 \,, \quad D_4^2 = 1 \,, \quad D_5^2 = 0 \,,
\ee
and the intersection form in the above basis is
\be
\left(D_i \cdot D_j \right) = 
\begin{pmatrix}
-2 & ~~1 & ~~0~ \\
~~1 & -1 & ~~1~ \\
~~0 & ~~1 & -1~ \\
\end{pmatrix} \,.
\label{eq:f6_intmat}
\ee
As before, we write divisors in the chosen basis as $D = (\,\cdot\,,\,\cdot\,,\,\cdot\,)$. The anti-canonical divisor $-K$ is the sum of the toric divisors, $-K = (2,4,3)$. The Mori cone generators and the dual nef cone generators are given by
\be
\begin{aligned}
\moricn_1 &= (1,0,0) \,, \quad \moricn_2 = (0,1,0) \,, \quad \moricn_3 = (0,0,1) \,,\\
\nefcn_1 &= (0,1,1) \,, \quad~ \nefcn_2 = (1,2,2) \,, \quad ~\nefcn_3 = (1,2,1) \,,
\end{aligned}
\ee
where $\moricn_i \cdot \nefcn_j = \delta_{ij}$. The rigid irreducible curves are simply the Mori cone generators.

\begin{figure}[H]
\begin{center}
\includegraphics[width=7cm]{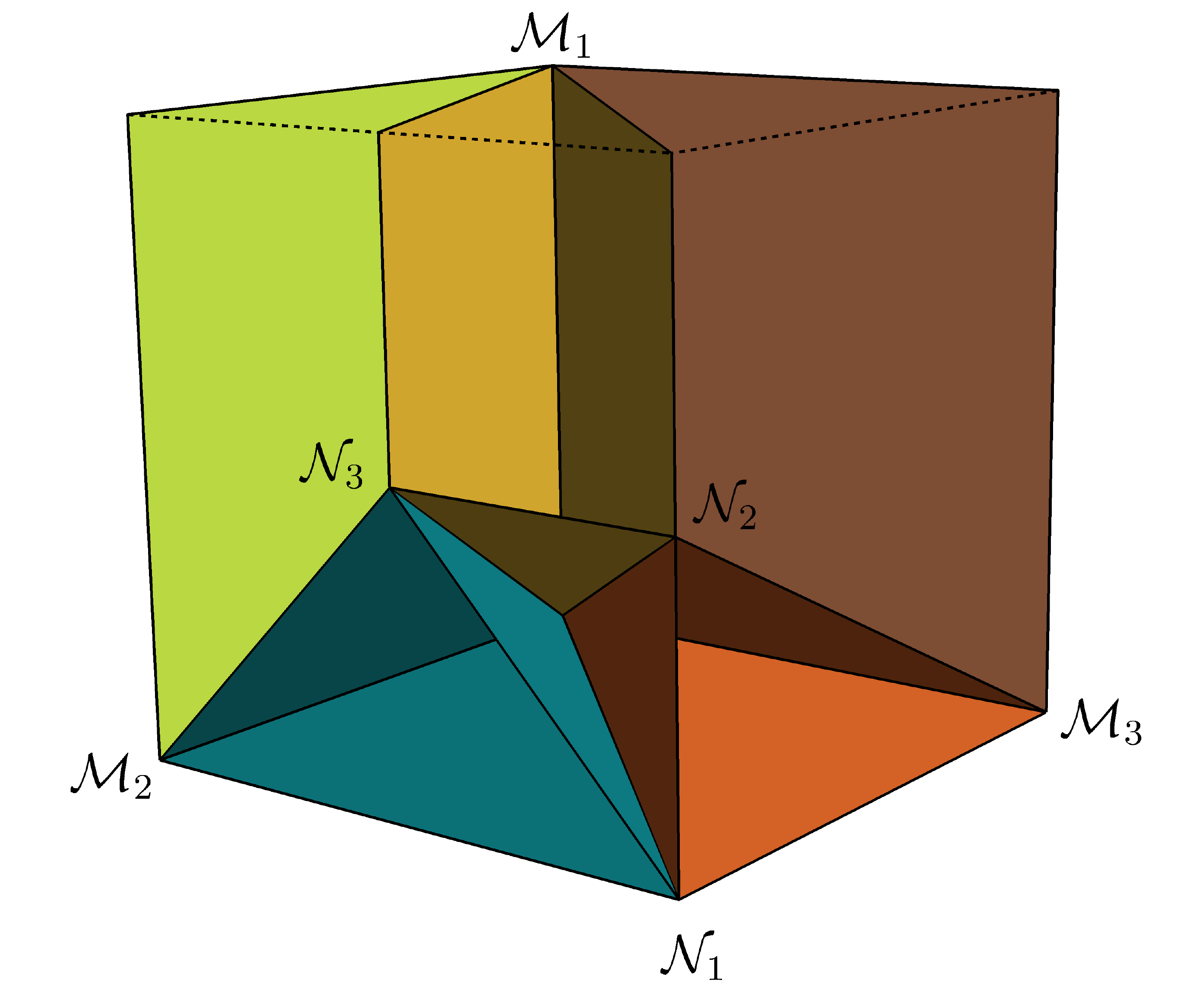}
\capt{5.8in}{fig:f6_regs}{The effective cone of the Gorenstein Fano toric surface $F_6$ splits into five different Zariski chambers outside of the nef cone. We have labelled the rays of the Mori cone generators $\moricn_i$ and the nef cone generators $\nefcn_j$.}
\end{center}
\end{figure}

Following the prescription outlined above, we determine the set $\mc{R}(F_6)$ of collections of rigid curves with negative definite intersection matrix. In the present case, the rigid curves are precisely the Mori cone generators, so intersections between rigid curves are given by the matrix in Equation~\eqref{eq:f6_intmat}. Intersections between a subset of the rigid curves are given by restricting the matrix. For example, restricting to $\{\cM_1,\cM_3\}$ gives
\be
I(\{\cM_1,\cM_3\}) = 
\begin{pmatrix}
-2 & ~~0~ \\
~~0 & -1~ \\
\end{pmatrix} \,,
\ee
which is negative definite, so that $\{\cM_1,\cM_3\} \in \mc{R}(F_6)$. In total there are five collections with negative definite intersection matrix,
\be
\mc{R}(F_6) = \{ \{\cM_1\} \,, \{\cM_2\} \,, \{\cM_3\} \,, \{\cM_1 , \cM_2 \} \,, \{\cM_1 , \cM_3 \} \} \,.
\ee
Correspondingly, one can note from \fref{fig:f6_regs} that the nef cone has three codimension 1 faces and two codimension 2 faces that have a non-vanishing intersection with the big cone.

The Zariski chamber $\Sigma_R$ for $R \in \mc{R}(\surf)$ is given by translating along the elements in $R$ the boundary of the nef cone spanned by generators orthogonal to all elements in $R$ up to boundaries -- which we recall belong to the subcone whose corresponding face has higher dimension.  Hence, the five Zariski chambers of $F_6$, which are sub-cones of the effective cone in addition to the nef cone, are
\begin{align*}
\Sigma_1 = \langle \cM_1, \cN_2, \cN_3 \rangle~,~~\Sigma_2 = \langle \cN_1&, \cM_2, \cN_3 \rangle~,~~\Sigma_3 = \langle \cN_1, \cN_2, \cM_3 \rangle~,\\[4pt]
\Sigma_{1,2} = \langle \cM_1, \cM_2, \cN_3 \rangle\setminus(\Sigma_1\cup\Sigma_2)~&,~~\Sigma_{1,3} = \langle \cM_1, \cN_2,  \cM_3 \rangle\setminus(\Sigma_1\cup\Sigma_3) \,,
\end{align*}
which are depicted in \fref{fig:f6_regs}. For a Zariski chamber $\Sigma_i$ corresponding to translation by a single Mori cone generator, the duals are simply $\cM_{i,\{i\}}^\vee = -\cM_i / \cM_i^2$. For the chambers with two Mori cone generators we have the following duals:
\begin{align*}
\Sigma_{1,2}: ~ \cM_{1,\{1,2\}}^\vee =  \cM_1+\cM_2 ~&\text{ and }~\cM_{2,\{1,2\}}^\vee = \cM_1+2\cM_2  \,,\\
\Sigma_{1,3}:~ \cM_{1,\{1,3\}}^\vee =  \tfrac{1}{2}\cM_{1} ~&\text{ and }~\cM_{3,\{1,3\}}^\vee = \cM_{3} \,.
\end{align*}

\smlhdg{Alternative packaging}

\noindent In some situations it is useful to repackage the information given by the Zariski chamber structure, instead writing a single formula that captures the behaviour of the decomposition throughout the effective cone.

Recall $\mc{R}(\surf)$ is the set of collections of rigid curves with negative intersection matrix. For a given $R \in \mc{R}(\surf)$, also recall that for any $C_A \in R$, one can define a unique effective dual curve $C_{A,R}^\vee$ with respect to $R$, which has $\mathrm{Supp}(C_{A,R}^\vee) \subseteq R$ and which satisfies $C_{A,R}^\vee \cdot C_B = -\delta_{AB}$ for all $C_B \in R$.

Each element $R \in \mc{R}(\surf)$ with $C_A \in R$ hence determines a divisor $C_{A,R}^\vee$, giving a set of possible duals for~$C_A$ written as $\mc{S}_A = \{ C_{A,R}^\vee \, | \, R \in \mc{R}(\surf), C_A\in R \}$. For example, in the $F_6$ case treated above, these are
\be
\mc{S}_1 = \{ \tfrac{1}{2}\cM_1 , \cM_1 + \cM_2 \} \,, \quad
\mc{S}_2 = \{ \cM_2 , \cM_1+2\cM_2 \} \,, \quad 
\mc{S}_3 = \{ \cM_3 \} \,.
\ee
These sets $\mc{S}_A$ provide an alternative way to determine a Zariski decomposition, as follows.

\begin{lma}\label{lma:alt_pack}
Let $D$ be an effective divisor with Zariski decomposition $D = P + N$, and let $C_A$ be a rigid curve. Then the following statements are true.
\begin{enumerate}
\item $C_A \in \mathrm{Supp}(N)$ if and only if there exists a divisor $C_{A,R}^\vee \in \mc{S}_A$ such that $C_{A,R}^\vee \cdot D < 0$.
\item If $C_A \in \mathrm{Supp}(N)$, then amongst $C_{A,R}^\vee \in \mc{S}_A$ the intersection $-C_{A,R}^\vee \cdot D$ is maximum when $R = \mathrm{Supp}(N)$.
\end{enumerate}
\end{lma}
\begin{proof}
These are straightforward to prove from the fact that $\mc{R}(\surf)$ is precisely the set of possible supports for the negative parts in Zariski decomposition.
\end{proof}

Defining $\mc{G}_A(D) = \{-C_{A,R}^\vee \cdot D \, | \, C_{A,R}^\vee \in \mc{S}_A \} \cup \{ 0 \}$, both results are contained in the statement that the coefficient of $C_A$ in $N$ is precisely the maximum of $\mc{G}_A$. Hence we have the following.
\begin{prp}\label{prp:alt_pack_form}
Let $D$ be an effective divisor. The negative part $N$ of its Zariski decomposition is given by
\be\label{eq:n_gen_ev}
N = \sum_{\cM_A \in \mc{I}(S)} \mathrm{max}\, (\mc{G}_A(D)) \, \cM_A \,.
\ee
\end{prp}
\begin{proof}
This follows immediately from Lemma~\eqref{lma:alt_pack}.
\end{proof}

The formula in Equation~\eqref{eq:n_gen_ev} has a natural interpretation in the context of detecting rigid curves by intersection, which was reviewed at the end of Section~\ref{sec:bas_defs}. For each rigid curve, the formula checks several candidate effective divisors to see which detects the maximal amount in $D$. The reason the number of candidates is small is because instead of checking every element in the cone of candidate divisors (see Section~\ref{sec:bas_defs}) it suffices to check the generators, as one can verify.

This perspective makes it clear that for a rigid curve $C_A$ the set of duals in $\mc{S}_A$ can also be understood as the generators of the cone in the N\'{e}ron-Severi group determined by the inequalities $\tilde{D} \cdot C_A \leq 0$ and $\tilde{D} \cdot C_i \geq 0$ for all $C_i \in \mc{I}(S)$ where $C_i \neq C_A$, excluding generators lying on a nef cone boundary. We also note in passing that these cones are a subset of the (closures of the) simple Weyl chambers as defined in Ref.~\cite{Rams16}.

\newpage
\section{Cohomology chambers}
\label{sec:bun_coh}

\subsection{Cohomology chambers and formulae}
\label{sec:coh_ch}

For $D$ an effective $\mbb{Z}$-divisor with Zariski decomposition $D = P + N$, Theorem~\eqref{thm:zar_coh_rel} asserts the preservation of cohomology $H^0\big(S, \mc{O}_S(D)\big) \cong H^0\big(S, \mc{O}_S(\floor{P})\big)$. Combining this with the general form of Zariski decomposition in Proposition~\eqref{prp:n_gen} gives the following explicit relation.

\begin{thm}\label{thm:flp_gen}
Let $\surf$ be a smooth projective surface, and let $D$ be an effective $\mbb{Z}$-divisor that lies within the Zariski chamber $\Sigma_{i_1,\ldots i_n}$, which is obtained by translating the codimension $n$ face of the nef cone that is orthogonal to the Mori cone generators $\cM_1,\ldots \cM_n$ along these generators. Then
\be
H^0\big(S, \mc{O}_S(D)\big) \cong H^0\bigg(S, \mc{O}_S\Big(D - \sum_{k=1}^n\, \ceil{-D\cdot \cM_{i_k,\{i_1,\ldots,i_n\}}^\vee} \,\cM_{i_k}\Big)\bigg) \,.
\ee
\end{thm}
\begin{proof}
Within a Zariski chamber, the form of the map $D \to P$ is fixed, with $N$ given by Equation~\eqref{eq:n_gen}. In the case of an integral divisor $D$, the round-up of $P = D - N$ leaves $D$ unchanged and so affects only the coefficients in $N$. Hence
\be
\floor{P} = D - \sum_{k=1}^n\, \ceil{-D\cdot \cM_{i_k,\{i_1,\ldots,i_n\}}^\vee} \,\cM_{i_k} \,.
\label{eq:flp_gen}
\ee
Combination with Theorem~\eqref{thm:zar_coh_rel} then gives the stated result.
\end{proof}

Since there is also the cohomology relation $H^0\big(S, \mc{O}_S(D)\big) \cong H^0\big(S, \mc{O}_S(\floor{P})\big)$ for the round-up $P$, one can also write an analogous result in this case.
\begin{crl}\label{crl:clp_gen}
In the situation of Theorem~\eqref{thm:flp_gen},
\be\label{eq:clp_gen}
H^0\big(S, \mc{O}_S(D)\big) \cong H^0\bigg(S, \mc{O}_S\Big(D - \sum_{k=1}^n\, \floor{-D\cdot \cM_{i_k,\{i_1,\ldots,i_n\}}^\vee} \,\cM_{i_k}\Big)\bigg) \,.
\ee
\end{crl}
\begin{proof}
This is analogous to the proof of Theorem~\eqref{thm:flp_gen}, using the cohomology relation of Corollary~\eqref{crl:zar_coh_rel_cl}.
\end{proof}

While the relations in Theorem~\ref{thm:flp_gen} and Corollary~\ref{crl:clp_gen} are valuable in themselves, unless something can be said about the cohomologies appearing on the right hand side, this is unlikely to be helpful in practice. On the classes of surfaces that we discuss below, something can indeed be said about these cohomologies, due to the existence of powerful vanishing theorems.

\smlhdg{Vanishing theorems and `pulling back' the index}

\noindent A vanishing theorem asserts the triviality of a number of the cohomologies for a subclass of line bundles, given certain properties of the variety. %
Perhaps the most well-known vanishing theorem is that of Kodaira.
\begin{nnthm}[Kodaira vanishing]
On a smooth irreducible complex projective variety $X$, for ample divisor $D$,
\be
H^i\big(X,\mc{O}_X(K_X + D)\big) = 0 ~ \text{for} ~ i > 0 \,.
\ee
\end{nnthm}

\noindent When a vanishing theorem ensures that all but one cohomology vanish, the remaining dimension can be computed from the index. For example, if Kodaira vanishing ensures that the higher cohomologies of a line bundle $\lb$ vanish, then
\be
\ind(\lb) = h^0(\surf,\lb) - h^1(\surf,\lb) + h^2(\surf,\lb) - \ldots = h^0(\surf,\lb) \,.
\ee
While individual cohomologies are generically difficult to compute, the index can be computed using only divisor intersection properties, due to the Hirzebruch-Riemann-Roch theorem. In the case of a surface $\surf$,

\be
\ind\big(\mc{O}_S(D)\big) = \ind\big( \mc{O}_S \big) + \frac{1}{2}( D \cdot D - D \cdot K_\surf ) \,,
\label{eq:hirz_riem_roch}
\ee
where $\mc{O}_S$ is the trivial bundle. Hence this gives a formula describing the sole non-trivial cohomology throughout the region of vanishing. Note in the surface case the formula is quadratic in the divisor $D$, or equivalently, quadratic in the integers specifying $D$ with respect to a basis.

\medskip

The set $\Sigma \cap \mathrm{NS}(\surf)$ of integral divisor classes in a Zariski chamber $\Sigma$ has images $\phifl\big(\Sigma\cap\mathrm{NS}(\surf)\big)$ and $\phicl\big(\Sigma\cap\mathrm{NS}(\surf)\big)$ under, respectively, the maps $\phifl \colon \divcls{D} \mapsto \divcls{\floor{P}}$ and $\phicl \colon \divcls{D} \mapsto \divcls{\ceil{P}}$ between integral divisor classes, defined in Equation~\eqref{eq:phi_map}. If a vanishing theorem applies across either of these images, then one can `pull back' the index to give a formula for cohomology throughout the Zariski chamber $\Sigma$.
\begin{prp}\label{prp:zar_to_coh}
Let $\surf$ be a smooth projective surface, and let $D$ be an effective $\mbb{Z}$-divisor that lies within the Zariski chamber $\Sigma_{i_1,\ldots i_n}$, which is obtained by translating the codimension $n$ face of the nef cone that is orthogonal to the Mori cone generators $\cM_1,\ldots \cM_n$ along these generators. If a vanishing theorem ensures triviality of the higher cohomologies for every line bundle in the image region $\phifl\big(\Sigma \cap \mathrm{NS}(\surf)\big)$, then
\be
h^0\big(S, \mc{O}_S(D)\big) = \ind\bigg(S, \mc{O}_S\Big(D - \sum_{k=1}^n\, \ceil{-D\cdot \cM_{i_k,\{i_1,\ldots,i_n\}}^\vee} \,\cM_{i_k}\Big)\bigg) \,.
\label{eq:gen_pullb}
\ee
If instead the vanishing theorem applies throughout the image region $\phicl\big(\Sigma \cap \mathrm{NS}(\surf)\big)$, then
\be\label{eq:gen_pullb_cl}
h^0\big(S, \mc{O}_S(D)\big) = \ind\bigg(S, \mc{O}_S\Big(D - \sum_{k=1}^n\, \floor{-D\cdot \cM_{i_k,\{i_1,\ldots,i_n\}}^\vee} \,\cM_{i_k}\Big)\bigg) \,.
\ee
In either situation, the Zariski chamber becomes a `cohomology chamber', in which the zeroth cohomologies are given throughout by a single formula.
\end{prp}
\begin{proof}
This is immediate given the cohomology preservation relations in Theorem~\eqref{thm:zar_coh_rel} and Corollary~\eqref{crl:zar_coh_rel_cl}.
\end{proof}
\noindent Note that while the image of a Zariski chamber under the map $\divcls{D} \to \divcls{P}$ lies on a boundary of the nef cone, due to the rounding operations involved in the maps $\phifl$ and $\phicl$ of integral divisor classes the images $\phifl\big(\Sigma\cap\mathrm{NS}(\surf)\big)$ and $\phicl\big(\Sigma\cap\mathrm{NS}(\surf)\big)$ will in general not lie entirely in the nef cone. Hence the vanishing theorems of interest are not simply those applying to the nef cone.

\begin{rmk}
In the case that every Zariski chamber is a cohomology chamber, the zeroth cohomology is described throughout the effective cone by using the expressions in Proposition~\eqref{prp:zar_to_coh} within each Zariski chamber. Though we do not know of cases where some Zariski chambers become cohomology chambers via the map $\phifl$ while others become cohomology chambers via the map $\phicl$, this is a possibility. When all Zariski chambers are cohomology chambers via the same map, then using the packaging in Proposition~\eqref{prp:alt_pack_form} one can alternatively write everywhere
\be
h^0\big(S, \mc{O}_S(D)\big) = \ind\bigg( D - \sum_{\cM_A \in \mc{I}(S)} \ceil{\mathrm{max}\, (\mc{G}_A(D))} \, \cM_A \bigg) \,,
\ee
where $D$ is any effective $\mbb{Z}$-divisor, and where $\mc{I}(\surf)$ is the set of negative curves on $\surf$ while $\mc{G}_A(D)$ is defined above Proposition~\eqref{prp:alt_pack_form}.
\end{rmk}

Note that since $\mathrm{ind}(D)$ is a quadratic polynomial in the divisor $D$ (or equivalently in the integers specifying $D$ with respect to a basis), the formula for the zeroth cohomology in a cohomology chamber is a polynomial in the divisor $\floor{P}$ or $\ceil{P}$. Since these involve rounding, the result is not a genuine polynomial in general. This is illustrated in the example in Section~\ref{sec:tor_sf_coh} below.

\smlhdg{Iteration and cohomology chambers}

\noindent While the integral divisors $\floor{P}$ and $\ceil{P}$ are not in general nef, if one iterates the process of Zariski decomposition and rounding, eventually this will reach an integral nef divisor. Naively then it seems that a vanishing theorem throughout the nef cone is sufficient to upgrade each Zariski chamber to a cohomology chamber.

However, two integral divisors from the same Zariski chamber may pass through distinct chambers on their journey to the nef cone, so that the combined map by which an index expression in the nef cone is `pulled back' would not be uniform throughout the original chamber, so that the Zariski chamber is not a cohomology chamber.

The following is an illustrative example. In Section~\ref{sec:zar_dec_subsec} we considered the divisor $D = 3D_1+D_2+5D_3$ on the Gorenstein Fano toric surface $F_8$, and determined its Zariski decomposition to be
\be
D = 3D_1 + D_2 + 5D_3 = \left(\frac{1}{2}D_1 + D_2 + \frac{1}{2}D_3\right) + \left(\frac{5}{2}D_1 + \frac{9}{2}D_3\right) \equiv P + N \,.
\ee
In Section~\ref{sec:zar_lb} we noted that $\floor{P}$ is not nef, so that a further Zariski decomposition is required. This decomposition is trivial,
\be
\floor{P} = 0 + D_2 \equiv P^{(1)} + N^{(1)} \,,
\ee
and the process terminates here, since $\floor{P^{(1)}} = \ceil{P^{(1)}} = 0$ is nef. Note one can check that $\ceil{P}$ is not nef, so in either case it requires multiple steps to reach the nef cone.

Now consider instead the divisor $D' = 3D_1+2D_2+5D_3$, which one can check has a Zariski decomposition
\be
D' = 3D_1 + 2D_2 + 5D_3 = \left(D_1 + 2D_2 + D_3\right) + \left(2D_1 +4D_3\right) \equiv P' + N' \,.
\ee
Since $\mathrm{Supp}(N) = \mathrm{Supp}(N')$, $D$ and $D'$ lie in the same Zariski chamber. However, in contrast to the case for $D$, the divisor $\floor{P'}$ is nef, so that the process terminates after a single step.

\smlhdg{Higher cohomologies}

\noindent When all Zariski chambers are also cohomology chambers, the zeroth cohomology is described throughout the entire effective cone by a set of regions and corresponding formulae, and is by definition zero outside. The higher cohomologies can then be obtained throughout the Picard group by Serre duality and the Hirzebruch-Riemann-Roch theorem
\be
\begin{gathered}
h^2\big(\surf, \mc{O}_\surf(D) \big) = h^0\big(\surf, \mc{O}_\surf(K_\surf-D) \big) \,, \\
h^0\big(\surf, \mc{O}_\surf(D) \big) - h^1\big(\surf, \mc{O}_\surf(D) \big) + h^2\big(\surf, \mc{O}_\surf(D) \big) = \ind\big( \mc{O}_\surf \big) + \frac{1}{2}( D \cdot D - D \cdot K_\surf ) \,.
\end{gathered}
\ee
In particular, we see that the chambers for the second cohomology are given by simply reflecting through the origin and translating by $K_\surf$ the Zariski chambers, while intersections of chambers in these two sets give chambers for the first cohomology.

\subsection{Toric surfaces}
\label{sec:tor_sf_coh}

\noindent On toric varieties, there is a powerful vanishing theorem due to Demazure. See for example Chapters~9.2 and~9.3 of Ref.~\cite{cox2011toric} for details and a proof.

\begin{nnthm}[Demazure vanishing for $\mbb{Q}$-divisors] Let $D$ be a nef $\mathbb Q$-divisor on a toric variety $X_{\Sigma}$ whose fan $\Sigma$ has convex support. Then
\begin{equation*}
H^q \big(X_\Sigma, \mathcal O_{X_\Sigma}(\floor{D}) \big) = 0,~\forall p>0 \,.
\end{equation*}
\end{nnthm}

Demazure's vanishing theorem is limited to toric varieties with convex support. However, this is not a restriction in the context of Zariski decomposition, because this condition holds when the toric variety is projective. To see this, note that a projective variety is compact. A toric variety is compact if and only if its fan is `complete' (see for example Theorem~3.1.19 of Ref.~\cite{cox2011toric}), which means its support is $\mbb{R}^n$ for some $n$. But this support is clearly convex. So compact toric varieties, and in particular projective toric varieties, are covered by Demazure vanishing.

This implies that on any projective toric surface, every Zariski chamber is also a cohomology chamber.
\begin{prp}
Let $\surf$ be a smooth projective toric surface, and $D$ an effective $\mbb{Z}$-divisor with Zariski decomposition $D = P + N$. Then
\be
h^0 \big(\surf, \mc O_\surf(D) \big) = \ind\big(\surf,\mc{O}_\surf(\floor{P})\big) \,.
\ee
Hence every Zariski chamber is upgraded to a cohomology chamber. Explicitly, if $D$ lies in the Zariski chamber $\Sigma_{i_1,\ldots i_n}$, then
\be
h^0\big(S, \mc{O}_S(D)\big) = \ind\bigg(S, \mc{O}_S\Big(D - \sum_{k=1}^n\, \ceil{-D\cdot \cM_{i_k,\{i_1,\ldots,i_n\}}^\vee} \,\cM_{i_k}\Big)\bigg) \,.
\ee
\end{prp}
\begin{proof}
The relation $h^0 \big(\surf, \mc O_\surf(D) \big) = \ind\big(\surf,\mc{O}_\surf(\ceil{P})\big)$ follows by combining the cohomology preserving property in Theorem~\eqref{thm:zar_coh_rel} with the Demazure vanishing theorem. The form of $\floor{P}$ is then as in Proposition~\ref{prp:zar_to_coh}.
\end{proof}

Moreover, on a projective toric surface the Zariski chamber decomposition is straightforward to implement, because the Mori cone, nef cone, and intersection form are all computed algorithmically from the toric data.

\smlhdg{Cohomology chambers on Gorenstein Fano toric surfaces}

\noindent A commonly used set of projective toric surfaces are the 16 Gorenstein Fano toric surfaces, whose fans are shown in \fref{fig:refl_poly} in Appendix~\ref{app:tor_surf_dat}. 

\begin{prp}
On the Gorenstein Fano toric surfaces $F_i$, the numbers $z(F_i)$ of cohomology chambers and, equally, the numbers of Zariski chambers are given by the following table,
\be
\begin{tabular}{C | C C C C C C C C C C C C C C C C}
i & 1 & 2 & 3 & 4 & 5 & 6 & 7 & 8 & 9 & 10 & 11 & 12 & 13 & 14 & 15 & 16 \\
 \hline
|\mc{I}(F_i)| & 0 & 0 & 1 & 1 & 3 & 3 & 6 & 4 & 5 & 4 & 6 & 7 & 7 & 8 & 8 & 9 \\
z(F_i) & 1 & 1 & 2 & 2 & 5 & 6 & 18 & 13 & 17 & 14 & 41 & 50 & 97 & 130 & 131 & 322
\end{tabular}
\ee
where we have included also the number $|\mc{I}(F_i)|$ of rigid curves.
\end{prp}
\begin{proof}
The intersection forms for the Gorenstein Fano toric surfaces are given in Appendix~\ref{app:tor_surf_dat}. From these one can determine the subsets of the Mori cone generators on which the intersection form is negative definite. These subsets count the Zariski chambers, together with the empty set which corresponds to the nef cone, and by the above discussion these are also cohomology chambers.
\end{proof}

\smlhdg{Example: cohomology chambers for the $F_6$ surface}

\noindent In Section~\ref{sec:zar_ch} we have determined the Zariski chambers for the example of the Gorenstein Fano toric surface $F_6$, and using the Demazure vanishing theorem these can be immediately upgraded to cohomology chambers.

From the intersection form, the subsets of Mori cone generators with negative definite intersection matrix are $\{ \{\cM_1\} \,, \{\cM_2\} \,, \{\cM_3\} \,, \{\cM_1 , \cM_2 \} \,, \{\cM_1 , \cM_3 \} \}$. Together with the nef cone, this gives six Zariski chambers, illustrated in \fref{fig:f6_regs}. In the upgrade to cohomology chambers, this gives the following formulae.

\begin{equation*}
\begin{tabular}{ L | L}
\Sigma &~ h^0\big(F_6,\mc{O}_{F_6}(D)\big) \\
\hline
\Sigma_\mathrm{nef} &~ {\rm ind}\big(F_6, \cO_{F_6}(D)\big) \\
\Sigma_1 &~ {\rm ind}\Big(F_6, \cO_{F_6}\big(D-\ceil{-\frac{1}{2}D\cdot\cM_1}\cM_1\big)\Big) \\
\Sigma_2 &~ {\rm ind}\Big(F_6, \cO_{F_6}\big(D-(-D\cdot\cM_2)\cM_2\big)\Big) \\
\Sigma_3 &~ {\rm ind}\Big(F_6, \cO_{F_6}\big(D-(-D\cdot\cM_3)\cM_3\big)\Big) \\
\Sigma_{1,2} &~ {\rm ind}\Big(F_6, \cO_{F_6}\big(D-(-D\cdot(\cM_1+\cM_2))\cM_1-(-D\cdot(\cM_1+2\cM_2))\cM_2\big)\Big) \\
\Sigma_{1,3} &~ {\rm ind}\Big(F_6, \cO_{F_6}\big(D-\ceil{-\frac{1}{2}D\cdot\cM_1}\cM_1-(-D\cdot\cM_3)\cM_3\big)\Big) \\
\end{tabular}
\end{equation*}
To compute the index with Hirzebruch-Riemann-Roch, one also needs that $-K_{F_6} = 2\cM_1+4\cM_2+3\cM_3$.

It is sometimes useful to express the cohomology formulae with respect to a basis. One obvious choice here is to write a general element $D$ of the N\'{e}ron-Severi group as a sum over the Mori cone generators,
\be
D = k_1 \moricn_1 + k_2 \moricn_2 + k_3 \moricn_3 \equiv (k_1, k_2 , k_3 ) \,.
\ee
The coefficients $-D \cdot \cM_{i,\{\ldots\}}^\vee$ in the general Zariski decomposition in Equation~\eqref{eq:n_gen} are functions of the $k_i$. For example, $-D \cdot \cM_{1,\{1,2\}}^\vee = - D \cdot (\cM_1 + \cM_2) = k_1 - k_3$, so that
in the Zariski chamber $\Sigma_{1,2}$ the map $D \to \floor{P}$ is
\be
\Sigma_{1,2}: ~ (k_1,k_2,k_3) ~\to~ (k_1 , k_2 , k_3) - (k_1-k_3,0,0) - (0,k_2-2k_3,0) = ( k_3, 2 k_3, k_3 ) \,,
\ee
and across all Zariski chambers the results for $\floor{P}$ are
\be
\begin{gathered}
    \Sigma_1:  ~ \big(\floor{\tfrac{1}{2}k_2}, k_2, k_3\big) \,, \quad
    \Sigma_2:  ~ \big( k_1, k_1+k_3, k_3\big) \,, \quad 
    \Sigma_3:  ~ \big( k_1, k_2, k_2\big) \,, \\
    \Sigma_{1,2}:  ~ \big( k_3, 2k_3, k_3 \big) \,, \quad
    \Sigma_{1,3}:  ~ \big( \floor{\tfrac{1}{2}k_2} , k_2, k_2\big) \,.
\end{gathered}
\ee
The formulae describing cohomology follow from these by using the expression for the index in this basis,
\be
\begin{aligned}
\ind \big( \mc{O}_{F_6}( k_1 \moricn_1 + k_2 \moricn_2 + k_3 \moricn_3 ) \big)
=  1 - k_1^2 + \tfrac{1}{2}k_2 + k_1 k_2 - \tfrac{1}{2}k_2^2 + \tfrac{1}{2}k_3 + k_2 k_3 - \tfrac{1}{2}k_3^2 \,,
\end{aligned}
\ee
so that the zeroth cohomology in each Zariski chamber is given by the following table.
\begin{equation*}
\begin{tabular}{ L | L }
\Sigma & ~ h^0\big(F_6, \mc{O}_{F_6}(k_1\cM_1+k_2\cM_2+k_3\cM_3)\big) \vspace{1pt}   \\
\hline
\Sigma_\mathrm{nef} & ~ 1 - k_1^2 + \tfrac{1}{2}k_2 + k_1 k_2 - \tfrac{1}{2}k_2^2 + \tfrac{1}{2}k_3 + k_2 k_3 - \tfrac{1}{2}k_3^2 \\
\Sigma_1 & ~ 1 + \tfrac{1}{2} k_2 - \tfrac{1}{2} k_2^2 + \tfrac{1}{2} k_3 + k_2 k_3 - \tfrac{1}{2} k_3^2 +  k_2 \floor{\tfrac{1}{2}k_2} -  \floor{\tfrac{1}{2}k_2}^2 
\\
\Sigma_2 & ~ 1 + \tfrac{1}{2} k_1 - \tfrac{1}{2} k_1^2 + k_3 + k_1 k_3 
\\
\Sigma_3 & ~ 1 - k_1^2 + k_2 + k_1 k_2 
\\
\Sigma_{1,2} & ~ 1 + \tfrac{3}{2} k_3 + \tfrac{1}{2} k_3^2 
\\
\Sigma_{1,3} & ~ 1 + k_2 +  k_2 \floor{\tfrac{1}{2} k_2} - \floor{\tfrac{1}{2} k_2}^2 
\end{tabular}
\end{equation*}

\subsection{Generalised del Pezzo (weak Fano) surfaces}
\label{sec:gen_dp}

\smlhdgnogap{Kawamata-Viehweg vanishing and Zariski decomposition}

\noindent On non-toric surfaces Demazure's vanishing theorem is unavailable. However there is the following generalisation of Kodaira vanishing (see for example Chapter~9.1.C of Ref.~\cite{lazarsfeld2004positivity2}).

\begin{nnthm}[Kawamata-Viehweg vanishing for $\mbb{Q}$-divisors] Let $X$ be a non-singular projective variety, and let $B$ be a $\mbb{Z}$-divisor. Assume that
\be
B \numeq D + \Delta
\ee
where $D$ is a nef and big $\mbb{Q}$-divisor, and $\Delta = \sum_i a_i D_i$ is a $\mbb{Q}$-divisor with fractional coefficients $0 \leq a_i < 1$ and with simple normal crossing support. Then
\begin{equation*}
H^q \big(X, \mathcal O_X(K_X + B) \big) = 0~\forall q>0 \,.
\end{equation*}
\end{nnthm}

We note the useful characterisation that on an irreducible projective variety $X$ of dimension $n$ a nef divisor $D$ is big if and only if $D^n>0$ (see Theorem 2.2.16 in Ref.~\cite{lazarsfeld2004positivity1}). Before applying the Kawamata-Viehweg vanishing theorem in the context of Zariski decomposition, we explain the definition of simple normal crossing support. A divisor $\sum_i a_i D_i$ on a variety of dimension $n$ has simple normal crossing support if each component $D_i$ is smooth (simpleness) and the reduced divisor $\sum_i D_i$ can be defined in the neighbourhood of any point by an equation 
\be
z_1 \cdot \ldots \cdot z_{k \leq n} = 0 \,,
\ee
where $z_i$ are independent local parameters (normal crossing support). For example, on a surface there are two independent local parameters, so if three components $D_i$ meet at a point, the divisor does not have simple normal crossing support.

\medskip

Note that the round-up $\ceil{D} = D + \Delta$ of a $\mbb{Q}$-divisor satisfies the requirements for $B$ in the theorem, provided that $\Delta = \ceil{D} - D$ has simple normal crossing support. Additionally, it is convenient to rewrite the theorem to state that higher cohomologies of $\mc{O}_X(B)$ vanish when $B$ is such that $B-K_X$ is nef and big. This gives the following corollary.

\begin{crl}\label{crl:kaw_vieh}
Let $\surf$ be a smooth projective surface and let $P$ be a $\mbb{Q}$-divisor. If $P-K_\surf$ is nef and big, and $\ceil{P}-P$ has simple normal crossing support, then
\begin{equation*}
H^q \big(\surf, \mathcal O_\surf(\ceil{P}) \big) = 0~\forall q>0 \,.
\end{equation*}
\end{crl}
\begin{proof}
This is immediate.
\end{proof}
\noindent Here we have suggestively written $P$ for the $\mbb{Q}$-divisor, as we are interested in applying this vanishing theorem to the positive part $P$ of a Zariski decomposition, as in Proposition~\eqref{prp:zar_to_coh}. While the positive part $P$ of a Zariski decomposition is by definition nef, it is not in general true that $P-K_\surf$ is nef and big, nor is it necessarily true that the fractional part of $P$ has simple normal crossing support. However, there is at least one obvious class of surfaces for which these conditions are satisfied for every Zariski decomposition of a $\mbb{Z}$-divisor. These are the generalised del Pezzo surfaces, as we now discuss.

\smlhdg{Application to generalised del Pezzo surfaces}
\newline
\noindent In order for Corollary~\eqref{crl:kaw_vieh} to apply to all $\mbb{Q}$-divisors $P$ throughout the nef cone, it is necessary that $P-K_\surf$ be nef and big for every nef $P$. Clearly, nefness of all $P-K_\surf$ for all $P$ requires that $-K_\surf$ be itself nef. Additionally, recalling that a nef divisor $D$ is big if and only if $D^2>0$, we check the self-intersection
\be
(P-K_\surf)^2 = P^2  + 2P \cdot (-K_\surf) + (-K_\surf)^2 \,.
\ee
Since $P$ is nef and effective, $P^2 \geq 0$, while since $P$ is nef and $-K_\surf$ is effective, $P \cdot (-K_\surf) \geq 0$. To guarantee that $P-K_\surf$ is always big, the final term must be positive, so that $-K_\surf$ must be big. A variety whose anti-canonical divisor is nef and big is called `weak Fano', or in two dimensions a `generalised del Pezzo' surface. 
All generalised del Pezzo surfaces except for the Hirzebruch surfaces $\mathbb F_0$ and $\mathbb F_2$ are blow-ups of the projective plane at $n\leq 8$ points in almost general position. For the main properties of generalised del Pezzo surfaces we refer the reader to textbook accounts such as Chapter~5.2 of Ref.~\cite{arzhantsev2015cox}, Chapter~8 of Ref.~\cite{dolgachev2012classical} and Ref.~\cite{Derenthal_2013}.

An important result for the present discussion is that on a generalised del Pezzo surface all negative curves are smooth and have self-intersection $-1$ or $-2$ (see e.g.~Lemma 2.7 in Ref.~\cite{elagin2020}). On del Pezzo surfaces the same statement holds, with the exception that there are no $-2$ curves. In particular, this result implies that on every (generalised) del Pezzo surface, if $D=P+N$ is the Zariski decomposition of an effective integral divisor $D$, the fractional divisor $\ceil{P}-P$ always has simple normal crossing support, as shown in the following proposition.

\begin{prp}\label{prp:norm_cros}
Let $\surf$ be a smooth projective surface, and let $D$ be an effective $\mbb{Z}$-divisor with Zariski decomposition $D = P + N$. If there are no curves $C$ on the surface $\surf$ with self-intersection $C^2 < -2$, then $\ceil{P}-P$ has normal crossing support.
\end{prp}
\begin{proof}
Clearly $\mathrm{Supp}(\ceil{P}-P) \subseteq \mathrm{Supp}(N)$. Hence if $N$ has normal crossing support then so does the $\ceil{P}-P$. A sufficient condition for $N$ to have normal crossing support is that no three curves in $\mathrm{Supp}(N)$ can intersect at a point. By definition, $\mathrm{Supp}(N)$ has negative definite intersection matrix, and hence so does any subset. The intersection matrix between three curves in the support of $N$ is hence a $3 \times 3$ negative definite matrix. In the assumption of the theorem, the diagonal entries are equal to $-1$ or $-2$. However such a matrix cannot have strictly positive elements in all off-diagonal entries, as is trivial to check. Hence any three curves in $\mathrm{Supp}(N)$ cannot all have pairwise intersections, and so certainly cannot all meet at a point. So $N$ has normal crossing support, and hence so does the $\ceil{P}-P$.
\end{proof}

Since all negative curves on a generalised del Pezzo surface are smooth, normal crossing support for  $\ceil{P}-P$ implies simple normal crossing support. Hence, the generalised del Pezzo surfaces are a class of surfaces on which Corollary~\eqref{crl:kaw_vieh} can be applied.

\begin{prp}
Let $\surf$ be a smooth generalised del Pezzo surface, and $D$ an effective $\mbb{Z}$-divisor with Zariski decomposition $D = P + N$. Then
\be
h^0 \big(\surf, \mc O_\surf(D) \big) = \ind\big(\surf,\mc{O}_\surf(\ceil{P})\big) \,.
\ee
Hence every Zariski chamber is upgraded to a cohomology chamber. Explicitly, if $D$ lies in a Zariski chamber $\Sigma_{i_1,\ldots i_n}$, then
\be
h^0\big(S, \mc{O}_S(D)\big) = \ind\bigg(S, \mc{O}_S\Big(D - \sum_{k=1}^n\, \floor{-D\cdot \cM_{i_k,\{i_1,\ldots,i_n\}}^\vee} \,\cM_{i_k}\Big)\bigg) \,.
\ee
\end{prp}
\begin{proof}
The relation $h^0 \big(\surf, \mc O_\surf(D) \big) = \ind\big(\surf,\mc{O}_\surf(\ceil{P})\big)$ follows by combining the cohomology preserving property in Theorem~\eqref{thm:zar_coh_rel} with Corollary~\eqref{crl:kaw_vieh}. The form of $\ceil{P}$ is then as in Proposition~\ref{prp:zar_to_coh}.
\end{proof}

\smlhdg{Classification of generalised del Pezzo surfaces}

\noindent Up to isomorphism, a generalised del Pezzo surface is either $\mbb{P}^1 \times \mbb{P}^1$, the Hirzebruch surface $\mbb{F}_2$, or a blow-up of $\mbb{P}^2$ at up to 8 points in almost general position. The ordinary del Pezzo surfaces, on which the anti-canonical divisor is not just nef and big but ample, are $\mbb{P}^1 \times \mbb{P}^1$ and the blow-ups of $\mbb{P}^2$ at points in general position.
A useful invariant of a generalised del Pezzo surface $\surf$ is the degree $d = (-K_\surf)^2$. On a generalised del Pezzo surface given by the blow-up of $\mbb{P}^2$ in $n$ points, the degree is $d=9-n$. In the remaining cases of $\mbb{P}^1 \times \mbb{P}^1$ and $\mbb{F}_2$ the degree is 8. Note $1 \leq d \leq 9$.

As already mentioned above, any curve on a generalised del Pezzo surface has self-intersection $C^2 \geq -2$, and any curve on an ordinary del Pezzo surface has self-intersection $C^2 \geq -1$. On a generalised del Pezzo surface, the number of curves with self-intersection $C^2=-2$ is at most 9, while the number of curves with self-intersection $C^2=-1$ is finite.

\medskip

Generalised del Pezzo surfaces are  classified in terms of their `type', as defined below.
\begin{nndefn}[Definition~3 in Ref.~\cite{Derenthal_2013}]
Two generalised del Pezzo surfaces have the same type if there is an isomorphism of their Picard groups preserving the intersection form that gives a bijection between their sets of classes of negative curves.
\end{nndefn}
\noindent This classification is particularly important for our purposes, since the decomposition of the Mori cone of a surface into Zariski chambers is determined by the Mori cone generators and the intersection form alone. While in general the negative curves do not fully specify the Mori cone, there is the following theorem.
\begin{nnthm}[Theorem~3.10 in Ref.~\cite{derenthal2007nef}]
On a generalised del Pezzo surface of degree $d \leq 7$, the effective cone is finitely generated by the set of $(-1)$- and $(-2)$-curves.
\end{nnthm}
\noindent While this theorem does not cover the cases with degrees 9 or 8, there is up to isomorphism precisely one generalised del Pezzo surface with degree 9, $\mbb{P}^2$, and three with degree 8, $\mbb{P}^1 \times \mbb{P}^1$, $\mbb{F}_2$, and $\mathrm{Bl}_1\mbb{P}^2$, which are all of distinct types. Hence the Mori cone and intersection form, and hence also the Zariski chambers, are fixed within a type.

Since the surfaces with degree $d=9$ or $d=8$ are toric and very simple, the classification of types of generalised del Pezzos can be restricted to $d \leq 7$. With this restriction, the Picard group $\mathrm{Pic}(\surf)$ and its intersection form depend only on the degree of $\surf$. What then differs among generalised del Pezzo surfaces of the same degree are the classes in $\mathrm{Pic}(\surf)$ which are effective.

To cut a long story short, the type of a generalised del Pezzo surfaces $S$ of degree $d\leq 7$ is specified by three elements: its degree $d$, the incidence graph $\Gamma$ of the (-2)-curves, which turns out to be always a disjoint union of Dynkin graphs of types $A$, $D$, $E$, and the number $m$ of (-1)-curves, hence the notation~$S_{d,\Gamma, m}$. 

For each degree $d\leq 7$, the graphs describing the possible configurations of (-2)-curves correspond to the Dynkin diagrams of all the subsystems of the root systems $R_d$ (up to automorphisms of $R_d$) given in the following table 
\be
\begin{tabular}{C | C C C C C C C}
d & 7 & 6 & 5 & 4 & 3 & 2 & 1 \\ \hline
R_d & \mathbf{A}_1 & \mathbf{A}_2+\mathbf{A}_1 & \mathbf{A}_4 & \mathbf{D}_5 & \mathbf{E}_6 & \mathbf{E}_7 & \mathbf{E}_8
\end{tabular}
\ee
with the exception of the subsystems $7\mathbf{A}_1$ of $R_2$ and $7\mathbf{A}_1$, $8\mathbf{A}_1$ and $\mathbf{D}_4 + 4\mathbf{A}_1$ of $R_1$, which only occur in characteristic 2 (see \cite{urabe1983,alexeev2006pezzo}). A subsystem consists of the set of $(-2)$-classes which are effective, and the simple roots of the subsystem are the irreducible elements, i.e.\ the $(-2)$-curves. 

The classes of the $(-1)$-curves are the elements $[D]\in \mathrm{Pic}(\surf)$ with $[D] \cdot [D] = -1$ and $[D] \cdot (-K_\surf) = 1$ which also satisfy $[D] \cdot [C] \geq 0$ for all $(-2)$-curves $C$. As an example, in each degree $d \leq 7$ there is a type corresponding to the empty subsystem. In this type there are no $(-2)$-curves, and this type contains precisely the ordinary del Pezzo surfaces of degree $d$. Here the constraint $[D] \cdot C \geq 0$ for all $(-2)$-curves $C$ is trivial so the classes of the $(-1)$-curves are determined by the conditions $[D] \cdot [D] = -1$ and $[D] \cdot (-K_\surf) = 1$.

There are 176 types of generalised del Pezzo surface. Within each type there can be multiple or infinitely many non-isomorphic surfaces. We note that 16 of these types contain a single toric surface, which are those in \fref{fig:refl_poly}, and these are the only toric examples. These are the Gorenstein Fano toric surfaces, discussed in Section~\ref{sec:tor_sf_coh}. For each degree $d$, there is one type containing the ordinary del Pezzo surfaces of degree $d$, except in the case $d=8$ where there are two types, each containing precisely one of the non-isomorphic ordinary del Pezzo surfaces of this degree. Note that the ordinary del Pezzo surfaces are non-toric only for $d<6$. The distribution of the types according to degree, and their breakdown into toric and non-toric cases, is
\be
\begin{tabular}{C | C C C C C C C C C}
d & 9 & 8 & 7 & 6 & 5 & 4 & 3 & 2 & 1 \\
 \hline
\#~\text{total} & 1 & 3 & 2 & 6 & 7 & 16 & 21 & 46 & 74 \\
\#~\text{toric} & 1 & 3 & 2 & 4 & 2 & 3 & 1 & 0 & 0 \\
\end{tabular} 
\ee

\smlhdg{Example: ordinary del Pezzo surfaces}

\noindent Among the generalised del Pezzo surfaces are the ordinary del Pezzo surfaces. As well as the simple case of $\mbb{P}^1 \times \mbb{P}^1$, these are the blow-ups of $\mbb{P}^2$ at $0 \leq n \leq 8$ points in general position, which we write as dP$_n$. These surfaces are non-toric only for $n>3$. The numbers of Zariski chambers $z(\mathrm{dP}_n)$ on dP$_n$ have been determined in Ref.~\cite{bauer2009counting}. By the above analysis, these are also the numbers of cohomology chambers. The numbers of chambers together with the numbers $|\mc{I}(\mathrm{dP}_n)|$ of negative curves (which must be $(-1)$-curves) are
\be
\begin{tabular}{C | C C C C C C C C C}
n & 0 & 1 & 2 & 3 & 4 & 5 & 6 & 7 & 8 \\
 \hline
 |\mc{I}(\mathrm{dP}_n)| & 0 & 1 & 3 & 6 & 10 & 16 & 27 & 56 & 240 \\
z(\mathrm{dP}_n) & 1 & 2 & 5 & 18 & 76 & 393 & 2764 & 33645 & 1501681
\end{tabular}
\ee
The description of the Zariski chambers is relatively simple. Recall that on an ordinary del Pezzo any curve has self-intersection $C^2 \geq -1$. Noting that a negative definite matrix with $-1$ on the diagonal must be zero off the diagonal, we see the supports of the negative part in Zariski decomposition are the sets of $(-1)$-curves having no mutual intersections. This implies that the duals appearing in the general Zariski decomposition in Equation~\eqref{eq:n_gen} are given by $C_{i,\mathrm{Supp}(N)}^\vee = -C_i$ for every support $\mathrm{Supp}(N) \ni C_i$. It is easy to check, for example recalling the discussion at the end of Section~\ref{sec:zar_ch}, that this implies for a divisor $D$ with Zariski decomposition $D = P + N$ that $C_i \in \mathrm{Supp}(N)$ if and only if $D \cdot C_i < 0$. Hence the boundaries between Zariski chambers are simply the hyperplanes orthogonal to rigid curves. The interiors of the Zariski chambers are hence the connected regions upon removing from the big cone this set of hyperplanes. These are just the simple Weyl chambers, as defined in Ref.~\cite{Rams16}.

Thus on ${\rm dP}_n$ the class of an effective divisor $D$ belongs to a Zariski chamber $\Sigma_{i_1,i_2\ldots,i_k}$ if and only if $D\cdot C_{j}<0$ for all $j\in \{i_1,i_2\ldots,i_s\}$ and $D\cdot C_j\geq 0$ for all $j\notin \{i_1,i_2\ldots,i_s\}$. Within this chamber, the zeroth cohomology of $\cO_{{\rm dP}_n}(D)$ is given by 
\begin{equation*}
h^0\left(\mathrm{dP}_n,\mc{O}_{\mathrm{dP}_n}(D)\right) = \mathrm{ind}\bigg(\mathrm{dP}_n, ~\mc{O}_{\mathrm{dP}_n}\Big(D ~~- \!\!\!\! \sum_{j\in \{i_1,i_2\ldots,i_s\}} \!\!\!\! (D\cdot(-C_j))~C_j\Big) \bigg) \,.
\end{equation*}
\noindent The formula can be alternatively written in the following form, which appeared in Refs.~\cite{Brodie:2019ozt, Brodie:2019pnz}:
\begin{equation*}
h^0\left(\mathrm{dP}_n,\mc{O}_{\mathrm{dP}_n}(D)\right) = \mathrm{ind}\Big(\mathrm{dP}_n, ~\mc{O}_{\mathrm{dP}_n}\big(D ~~ +\!\!\!\! \sum_{C_i\in \mc{I}(\mathrm{dP}_n)} \!\!\!\! \theta( - D \cdot C_i ) \, (D \cdot C_i)~C_i\big) \Big) \,.
\label{eq:dp2_h0_alt}
\end{equation*}

\smlhdg{Example: Degree 6 Type $\mbb{A}_1$}

\noindent On a generalised del Pezzo surface the Picard number $\rho$ is related to the degree $d$ by $\rho = 10-d$. For degrees $d \geq 7$, i.e.\ Picard numbers $\rho \leq 3$, all examples of generalised del Pezzo surfaces are toric. In degree $d = 6$, where the Picard number is $\rho = 4$, in addition to four toric types there are two non-toric types. These non-toric types correspond to root subsystems $\mathbf{A}_1$ and $\mathbf{A}_2$. All six types are shown in Table~4 of Ref.~\cite{Derenthal_2013}. 
We take the case of the subsystem $\mathbf{A}_1$ as a simple example of a generalised del Pezzo surface $\surf$ which is neither toric nor an ordinary del Pezzo surface.

The Picard group and intersection form depend only on the degree of the generalised del Pezzo surface. For degree $d=6$ the Picard lattice is spanned by $l_0$ and $l_i$ with $i=1,2,3$, which we write collectively as $l_A$. We write a general element $\sum_{A=0}^4 k_A l_A$ in this basis as a vector $(k_0, k_1, k_2, k_3)$.  In this basis the intersection form and the anti-canonical divisor class $-K_\surf$ are
\be
(l_A \cdot l_B) = 
\begin{pmatrix}
1 & 0 & 0 & 0 \\
0 & -1 & 0 & 0 \\
0 & 0 & -1 & 0 \\
0 & 0 & 0 & -1 \\
\end{pmatrix} \,,
\quad
-K_\surf=(3,1,1,1) \,.
\ee
There are six $(-1)$-classes, satisfying $D^2 = -1$ and $D \cdot (-K_\surf) = 1$, which are given by
\be
(0,1,0,0) ~,~ (0,0,1,0)  ~,~ (0,0,0,1)  ~,~ (1,-1,-1,0)  ~,~ (1,-1,0,-1)  ~,~ (1,0,-1,-1) \,,
\ee
and there are eight $(-2)$-classes, satisfying $D^2 = -2$ and $D \cdot (-K_\surf) = 0$, which are given by
\be
\pm(1,-1,-1,-1)  ~,~ \pm(0,1,-1,0)  ~,~ \pm (0,1,0,-1)  ~,~ \pm(0,0,1,-1) \,.
\ee

Taking the root subsystem to be $\mathbf{A}_1$ corresponds to a single $(-2)$-class $(1,-1,-1,-1)$ being effective, so that there is a single $(-2)$-curve $C^{(-2)}$. The effective $(-1)$-classes $C^{(-1)}$, i.e.\ the classes of $(-1)$-curves, are then those classes $C$ satisfying $C \cdot C^{(-2)} \geq 0$, explicitly,
$ \{ (0,1,0,0) ~,~ (0,0,1,0) ~,~ (0,0,0,1) \} \,.$
Together these give four rigid curves, which are precisely the generators $\cM_A$ of the Mori cone. The dual nef cone generators, which can be chosen to satisfy $\cM_A \cdot \cN_B = \delta_{AB}$, then follow, giving
\be
\begin{array}{ l l l l }
\cM_0 = (1,-1,-1,-1) & \cM_1 = (0,1,0,0) & \cM_2 = (0,0,1,0) & \cM_3 = (0,0,0,1) \,, \vspace{.1cm} \\
\cN_0 =  (1,0,0,0) & \cN_1 =  (1,-1,0,0) & \cN_2 =  (1,0,-1,0) & \cN_3 = (1,0,0,-1) \,.
\end{array}
\ee
The intersection form between the Mori cone generators is
\be
(\cM_A \cdot \cM_B) = 
\begin{pmatrix}
-2 & 1 & 1 & 1\\
1 & -1 & 0 & 0\\
1 & 0 & -1 & 0\\
1 & 0 & 0 & -1
\end{pmatrix} \,.
\ee

The above data determines the structure of the Zariski chambers. From the intersection matrix $(\cM_A \cdot \cM_B)$, there are eleven subsets of the rigid curves $\{\cM_A\}$ which have a negative definite intersection form, which are
\be
\begin{gathered}
\{\cM_0\} ~,~ \{\cM_1\} ~,~ \{\cM_2\} ~,~ \{\cM_3\} ~, \\
 \{\cM_0\,,\,\cM_1\} ~,~ \{\cM_0\,,\,\cM_2\} ~,~ \{\cM_0\,,\,\cM_3\} ~,~ \{\cM_1\,,\,\cM_2\} ~,~ \{\cM_1\,,\,\cM_3\} ~,~ \{\cM_2\,,\,\cM_3\} ~, \\
  \{\cM_1\,,\,\cM_2\,,\,\cM_3\} \,.
\end{gathered}
\ee
Together with the nef cone, this gives eleven Zariski chambers. The zeroth cohomology is then given throughout the effective cone by the following formulae.

\begin{equation*}
\begin{tabular}{ L | L}
\Sigma &~ h^0\big(\surf,\mc{O}_{\surf}(D)\big) \\
\hline
\Sigma_\mathrm{nef} &~ \ind\big(\surf,\mc{O}_\surf(D)\big)\\
\Sigma_0 &~ \ind\Big(\surf, \mc{O}_\surf\big(D-\ceil{-\frac{1}{2}D\cdot\cM_0}\cM_0\big)\Big) \\
\Sigma_{i\in\{1,2,3\}} &~ \ind\Big(\surf, \mc{O}_\surf\big(D-(-D\cdot\cM_i)\cM_i\big)\Big) \\
\Sigma_{0,i\in\{1,2,3\}} &~ \ind\Big(\surf, \mc{O}_\surf\big(D-\ceil{-\frac{1}{2}D\cdot\cM_0}\cM_0-(-D\cdot\cM_i)\cM_i\big)\Big) \\
\Sigma_{i\in\{1,2\},j>i} &~ \ind\Big(\surf, \mc{O}_\surf\big(D-(-D\cdot\cM_i)\cM_i-(-D\cdot\cM_j)\cM_j\big)\Big) \\
\Sigma_{1,2,3} &~ \ind\Big(\surf, \mc{O}_\surf\big(D-\sum_{i=1}^3(-D\cdot\cM_i)\cM_i \big)\Big)  \\
\end{tabular}
\end{equation*}

Writing $D=(k_0,k_1,k_2,k_3)$ in the basis $\{l_0, l_1, l_2, l_3\}$, the above formulae become:
\begin{equation*}
\begin{tabular}{ L | L}
\Sigma &~ h^0\big(\surf,\mc{O}_{\surf}(D)\big) \\
\hline
\Sigma_\mathrm{nef} &~ 1+\frac{1}{2} \left(3 k_ 0+k_ 0^2-k_ 1-k_ 1^2-k_ 2-k_ 2^2-k_ 3-k_3^2 \right)\\
\Sigma_0 &~ 1+\frac{1}{2} \left(3 k_ 0+k_ 0^2-k_ 1-k_ 1^2-k_ 2-k_ 2^2-k_ 3-k_3^2 \right) - \floor{\frac{1}{2} \left(k_ 0+k_1+k_ 2+k_ 3\right)}^2+ \\
&~~~~~~~~~~~~ \floor{\frac{1}{2}\left(k_ 0+k_ 1+k_ 2+k_ 3\right)} \left(3+k_ 0+k_ 1+k_ 2+k_ 3\right)\\
\Sigma_{i\in\{1,2,3\}} &~ \frac{1}{2} \left(3 k_ 0+k_ 0^2-k_ 1-k_ 1^2-k_ 2-k_ 2^2-k_ 3-k_3^2 +k_i^2+k_i \right)\\
\Sigma_{0,i\in\{1,2,3\}} &~ 1+\frac{1}{2} \left(3 k_ 0+k_ 0^2-k_ 1-k_ 1^2-k_ 2-k_ 2^2-k_ 3-k_3^2 +k_i^2+k_i  \right) - \floor{\frac{1}{2} \left(k_ 0+k_1+k_ 2+k_ 3\right)}^2+ \\
&~~~~~~~~~~~~ \floor{\frac{1}{2}\left(k_ 0+k_ 1+k_ 2+k_ 3\right)} \left(3+k_ 0+k_ 1+k_ 2+k_ 3-k_i\right)\\
\Sigma_{i\in\{1,2\},j>i} &~ \frac{1}{2} \left(3 k_ 0+k_ 0^2-k_ 1-k_ 1^2-k_ 2-k_ 2^2-k_ 3-k_3^2 +k_i^2+k_i +k_j^2+k_j \right)\\
\Sigma_{1,2,3} &~ 1+ \frac{1}{2} \left(3 k_ 0+k_ 0^2\right)\\
\end{tabular}
\end{equation*}

\subsection{K3 surfaces}

A complex K3 surface is a compact connected complex surface $\surf$ with trivial canonical bundle and with $H^1(\surf,\mc{O}_\surf) = 0$. These are the Calabi-Yau surfaces, excluding, by the latter condition, a product of tori. Among the smooth complex K3 surfaces, we restrict to the projective case, since this is the case in which Zariski decomposition can be applied. Below we will often say `K3 surface' where we mean `smooth projective complex K3 surface'.

On a K3 surface the Picard group and the N\'{e}ron-Severi group coincide, so we will not need to make a distinction. Moreover, the only negative curves on a K3 surface are $(-2)$-curves. See Ref.~\cite{huybrechts2016lectures} for more properties of K3 surfaces. 

\smlhdg{Vanishing in the big cone}

\noindent The Kawamata-Viehweg vanishing theorem can be applied to K3 surfaces, bearing in mind that in the present case the canonical bundle is trivial, which leads to the following specialisation of Corollary~\eqref{crl:kaw_vieh}.

\begin{crl} \label{crl:kaw_vieh_k3} Let $\surf$ be a smooth projective complex K3 surface and let $P$ be a $\mbb{Q}$-divisor. If $P$ is nef and big, and the fractional part $\ceil{P}-P$ has simple normal crossing support, then
\begin{equation*}\label{crl:kaw_vieh_k3}
H^q \big(\surf, \mathcal O_\surf(\ceil{P}) \big) = 0~\forall q>0 \,.
\end{equation*}
\end{crl}

The positive part $P$ of an effective divisor is nef by definition, but in general it is not big. Hence on a K3 surface, Kawamata-Viehweg vanishing applies only to the subset of possible positive parts that are in the big cone. This is in contrast to the case of generalised del Pezzo surfaces treated in Section~\ref{sec:gen_dp} above, where the vanishing theorem applied throughout the nef cone.
On the other hand, since on a K3 surface the only negative curves are smooth $(-2)$-curves, Proposition~\ref{prp:norm_cros} implies that the fractional part $\ceil{P}-P$ always has simple normal crossing support. As such, the Kawamata-Viehweg vanishing theorem ensures the vanishing of the higher cohomologies of $\mc{O}_\surf(\ceil{P})$ for those positive parts $P$ that are in the big cone (the interior of the Mori cone). In particular, this excludes the cases where $\mc{O}_\surf(\ceil{P})=\mc{O}_\surf$. 

The question remains which effective $\mbb{Z}$-divisors have positive parts in the interior of the Mori cone and which have positive parts on the boundary. This is answered by the following lemma.

\begin{lma}
Let $\surf$ be a smooth projective surface on which linear equivalence and numerical equivalence coincide, and let $D$ be an effective $\mbb{Q}$-divisor. The positive part $P$ in the Zariski decomposition of $D$ is in the interior of the Mori cone if and only if $D$ is.
\label{lma:preimage}
\end{lma}
\begin{proof}
Note both $P$ and $D$ are in the Mori cone, either in the interior or on the boundary. We prove the statement by showing that if $D$ is on the boundary then so is $P$, and that if $D$ is in the interior then so is $P$.

First suppose $D$ is on the boundary. Then there exists a nef cone generator $\mc{N}_0$ such that $D \cdot \mc{N}_0 = 0$. Since $P = D - N$, this means $P \cdot \mc{N}_0 = - N \cdot\mc{N}_0$. But since $P$ and $N$ are both in the Mori cone, $P \cdot \mc{N}_0 \geq 0$ and $N \cdot \mc{N}_0 \geq 0$. This is consistent only if $P \cdot \mc{N}_0 = 0$, so that $P$ is on the boundary of the Mori cone.

Next suppose $D$ is in the interior. We have $P^2 = P\cdot(D-N)=P\cdot D$, since $P\cdot N=0$ by definition. But since $D$ is in the interior of the Mori cone, it follows that $P\cdot D>0$, unless $P$ is numerically equivalent and hence linearly equivalent to $0$. However, $P$ is linearly equivalent to $0$ only when $D$ lies on the boundary of the Mori cone, which cannot happen since $D$ is big. Therefore $P^2>0$, which implies $P$ is big, i.e.\ in the interior of the Mori cone. 
\end{proof}

\noindent This immediately gives the following proposition.

\begin{prp}
Let $\surf$ be a smooth projective complex K3 surface, and $D$ an effective $\mbb{Z}$-divisor not on the boundary of the Mori cone with Zariski decomposition $D = P + N$. Then
\be
h^0 \big(\surf, \mc O_\surf(D) \big) = \ind\big(\surf,\mc{O}_\surf(\ceil{P})\big) \,.
\ee
Hence every Zariski chamber, excluding its intersection with the boundary of the Mori cone, is upgraded to a cohomology chamber. Explicitly, if $D$ lies in the Zariski chamber $\Sigma_{i_1,\ldots i_n}$, then
\be
h^0\big(S, \mc{O}_S(D)\big) = \ind\bigg(S, \mc{O}_S\Big(D - \sum_{k=1}^n\, \floor{-D\cdot \cM_{i_k,\{i_1,\ldots,i_n\}}^\vee} \,\cM_{i_k}\Big)\bigg) \,.
\ee
\end{prp}
\begin{proof}
Since $P$ is big from Lemma~\ref{lma:preimage}, combining the cohomology preserving property in Theorem~\eqref{thm:zar_coh_rel} with Corollary~\ref{crl:kaw_vieh_k3} implies that $h^0 \big(\surf, \mc O_\surf(D) \big) = \ind\big(\surf,\mc{O}_\surf(\ceil{P})\big)$. The form of $\ceil{P}$ is then as in Proposition~\ref{prp:zar_to_coh}.
\end{proof}

For integral divisors lying on the boundary of the Mori cone, the zeroth cohomology is in general not determined by the current framework of combining Zariski decomposition with vanishing theorems, and these require a separate discussion, which we will not attempt here. However, in the special case of integral divisors on the boundary whose positive part is trivial, there is the following simple result.

\begin{prp}
Let $\surf$ be a smooth projective surface and $D$ an effective $\mbb{Z}$-divisor on the boundary of the Mori cone. If the intersection form is negative definite on the support of $D$, then 
$ h^0\big(S, \mc{O}_\surf(D)\big) = h^0\big(S,\mc{O}_\surf\big) \,.$
\end{prp}
\begin{proof}
This is immediate, since in this case the negative part of $D$ is $D$ itself, so its positive part is trivial.
\end{proof}

\smlhdg{Example: quartic hypersurface in $\IP^3$ with Picard number $3$}

\noindent In Ref.~\cite{bauer2010K3} it was shown that there exist K3 surfaces $\surf$ constructed as smooth quartic surfaces in $\IP^3$ with Picard number $3$ and three $(-2)$-curves $L_1$, $L_2$ and $C$ (two lines and a conic). These are the generators of the Mori cone and we write $L_1 = {\cM}_1$, $L_2 = {\cM}_2$, and $C = {\cM}_3$. The intersection form is
\begin{equation*}
(\cM_i \cdot \cM_j) =
\left( \begin{array}{rrr}
\!\!-2 & 0 & 2\\
0 & -2 &2 \\
2 & 2& -2
\end{array}\right) \,,
\end{equation*}
and in this basis the dual nef cone is generated by ${\cN}_1=(0,1,1)$, ${\cN}_2=(1,0,1)$ and ${\cN}_3=(1,1,1)$. There are four subsets of $\{\cM_1,\cM_2,\cM_3\}$ on which the intersection form is negative definite, namely, $\{\cM_1\}$, $\{\cM_2\}$, $\{\cM_3\}$ and $\{\cM_1,\cM_2\}$, giving rise to four Zariski chambers. Together with the nef cone $\Sigma_{\rm nef}$, this makes a total of five Zariski chambers.

Hence, for any divisor $D$ not on the boundary of the Mori cone, the zeroth cohomology is given by the formulae in the table below. We also include as an additional line the case of divisors on faces of the Mori cone which project to the origin under the map $D\rightarrow P$:
\begin{equation}
\begin{tabular}{L | L }
\Sigma &~ h^0(S, \cO_\surf(D)) \\
\hline
\Sigma_\mathrm{nef}\cap{\rm Big}(S)
&~ {\rm ind}\big(S, \cO_\surf(D)\big) \\[4pt]
\Sigma_1  \cap{\rm Big}(S) 
&~ {\rm ind}\,\Big(S, \cO_\surf\big(D-\ceil{-\frac{1}{2}D\cdot\cM_1}\cM_1\big)\Big) \\[4pt]
\Sigma_2\cap{\rm Big}(S) 
&~ {\rm ind}\Big(S, \cO_\surf\big(D-\ceil{-\frac{1}{2}D\cdot\cM_2}\cM_2\big)\Big) \\[4pt]
\Sigma_3\cap{\rm Big}(S)
&~ {\rm ind}\Big(S, \cO_\surf\big(D-\ceil{-\frac{1}{2}D\cdot\cM_3}\cM_3\big)\Big) \\[4pt]
\Sigma_{1,2}\cap{\rm Big}(S) &~ {\rm ind}\Big(S, \cO_\surf\big(D-\ceil{-\frac{1}{2}D\cdot\cM_1}\cM_1-\ceil{-\frac{1}{2}D\cdot\cM_2}\cM_2\big)\Big) \\[4pt]
\langle \cM_1, \cM_2 \rangle_{\mbb{R}_{\geq0}},  \langle \cM_3 \rangle_{\mbb{R}_{\geq0}}
&~ h^0\big(S, \mc{O}_\surf \big) = 1
\end{tabular}
\label{eq:firstK3_CohFormula}
\end{equation}
The zeroth cohomology is undetermined on the remaining parts of the Mori cone boundary, which are $\langle \cM_2, \cM_3 \rangle_{\mbb{R}_{>0}}$ and $\langle \cM_1, \cM_3 \rangle_{\mbb{R}_{>0}}$. Note that in this example the Zariski chambers are simple Weyl chambers and as such, in the region $\Sigma_{1,2}$, the duals $\cM_{1,\{1,2\}}^\vee$ and $\cM_{2,\{1,2\}}^\vee$ are simply given by $\cM_{i,\{1,2\}}^\vee= \frac{1}{|\cM_i\cdot\cM_i|}\cM_i$.

\begin{figure}[H]
\begin{center}
\includegraphics[width=7cm]{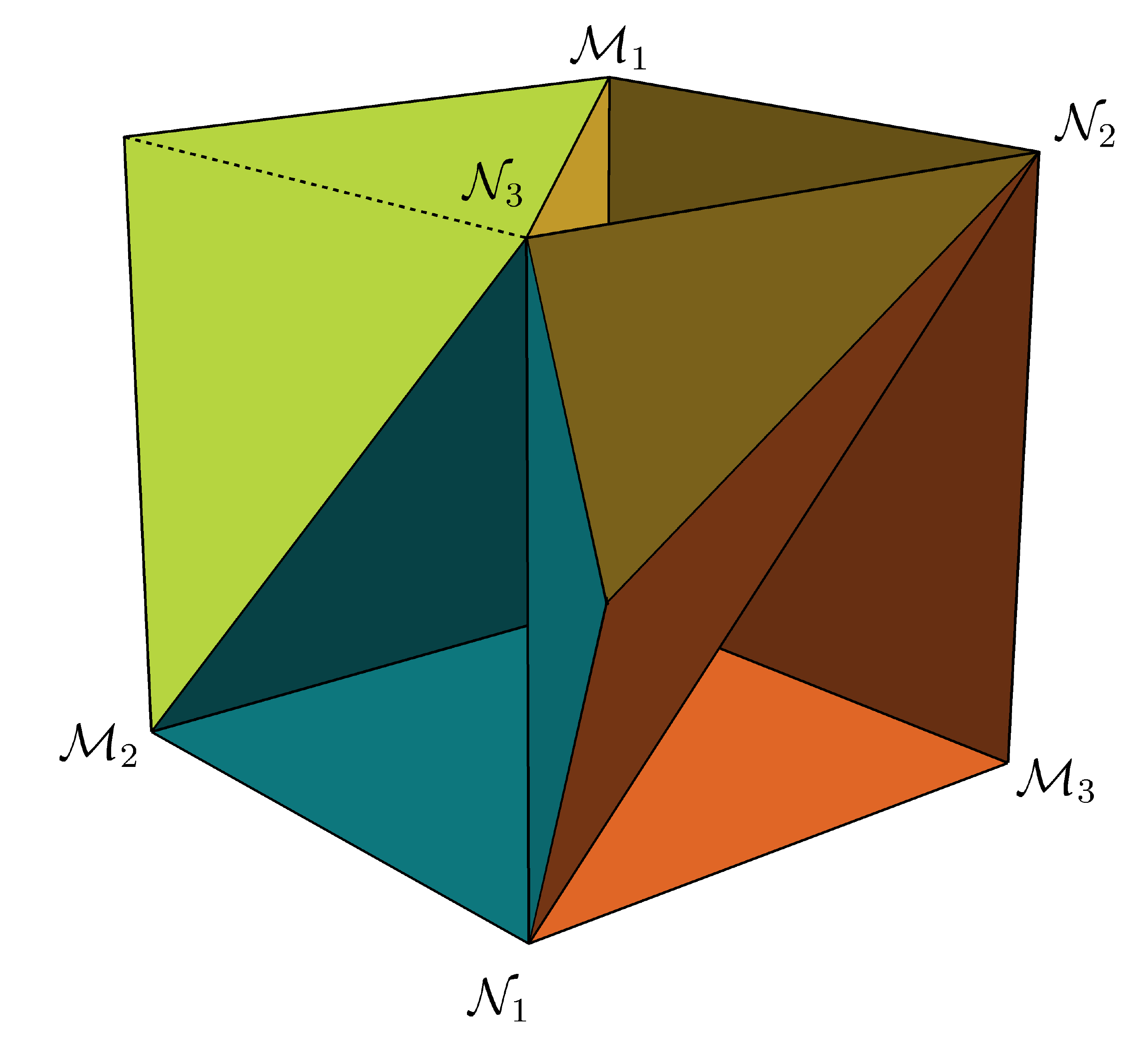}
\capt{5.8in}{fig:k3_3dex_regs}{The delineation of the effective cone into Zariski chambers on the K3 surface described by a quartic hypersurface in $\mbb{P}^3$ with Picard number 3. We have labelled the rays of the Mori cone generators $\moricn_i$ and the nef cone generators $\nefcn_j$.}
\end{center}
\end{figure}

It is sometimes useful to recast the above cohomology formulae in a basis. For numerical classes in the effective cone we write $D=k_1 L_1+k_2L_2+k_3C$, with $k_1,k_2,k_3\geq 0$. In this basis, the index formula in Equation~\eqref{eq:hirz_riem_roch} with $K_S=\cO_S$ and ${\rm ind}\big(S, \cO_\surf\big)=2$ becomes
\begin{equation}
{\rm ind}\big(S, \cO_\surf(D)\big) = 2+\tfrac{1}{2}D\cdot D = 2-(k_1-k_3)^2-(k_2-k_3)^2+k_3^2 \,.
\end{equation}
\noindent Using this expression, the zeroth cohomology formulae in Equation~\eqref{eq:firstK3_CohFormula} become
\begin{equation}
\begin{tabular}{L | L }
\Sigma &~ h^0(S, \cO_\surf(D)) \\
\hline
\Sigma_\mathrm{nef}\cap{\rm Big}(S)
&~2-(k_1-k_3)^2-(k_2-k_3)^2+k_3^2\\[4pt]
\Sigma_1\cap{\rm Big}(S) 
&~ 2 + (2 k_3 - k_2)k_2 \\[4pt]
\Sigma_2\cap{\rm Big}(S) 
&~2 + (2 k_3 - k_1)k_1\\[4pt]
\Sigma_3\cap{\rm Big}(S)
&~ 2+2k_1k_2\\[4pt]
\Sigma_{1,2}\cap{\rm Big}(S) &~ 2+k_3^2 \\[4pt]
\langle \cM_1, \cM_2 \rangle_{\mbb{R}_{\geq0}},  \langle \cM_3 \rangle_{\mbb{R}_{\geq0}} &  ~1
\end{tabular}
\end{equation}
\vspace{21pt}

\smlhdg{Example: Weierstrass model}

\noindent We now discuss a K3 surface $S$ realised as a Weierstrass fibration of an elliptic curve over $\mathbb P^1$. The example we take can be realised as a hypersurface in a three-dimensional toric variety, whose fan is given by a triangulation of the surface of the polytope shown in \fref{fig:K3polytope}.
\begin{figure}[h]
\begin{center}
\raisebox{.1in}{\includegraphics[width=4.8cm]{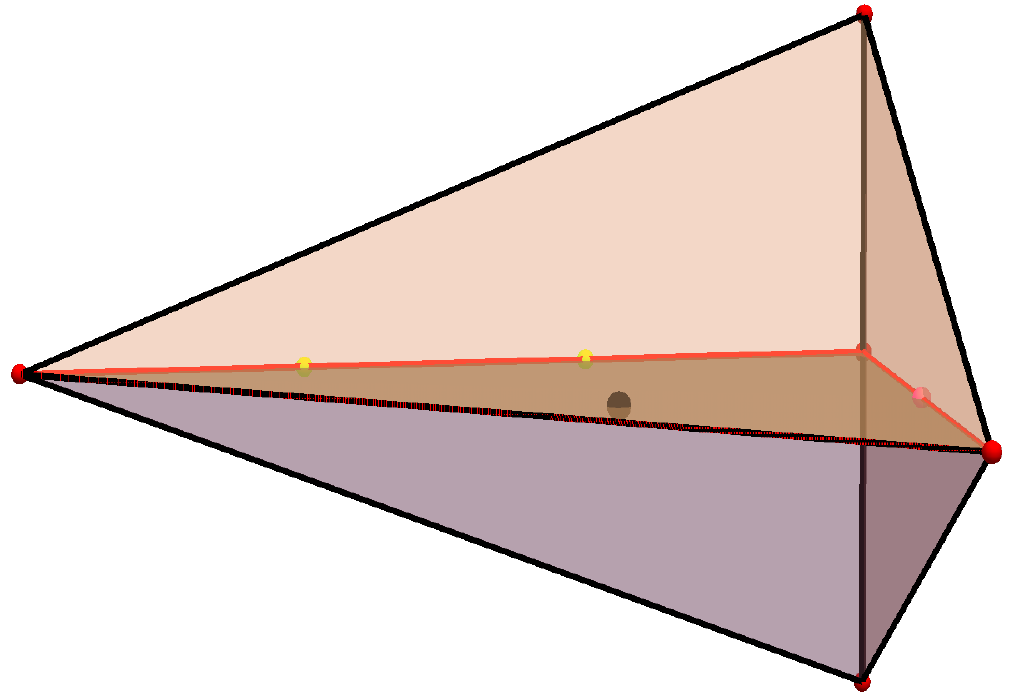}}
\capt{5.8in}{fig:K3polytope}{The polytope giving the ambient toric variety. The ambient space for the Weierstrass elliptic curve, $\mathbb{P}_{231}$, corresponds to the `slice', while the vertical direction corresponds to the $\mathbb P^1$ base. The vertices of the polytope are $\{(-1, 1, 1), (0, -2, 1), (1, 1, 1), (0, 1, -1)\}$.}
\end{center}
\end{figure}

The fibration  $\pi\colon \surf \to \mbb{P}^1$ has a single zero-section $\sigma$. The Mori cone is generated by $\cM_1 = \pi^*(H)$ and $\cM_2 = \sigma$, where $\pi^*(H)$ is the pullback of the hyperplane class (point) on the $\mbb{P}_1$ base. In this basis, the intersection form is
\begin{equation*}
(\cM_i \cdot \cM_j) = 
\left( \begin{array}{rr}
\!\!0 & 1 \\
1 & -2
\end{array}\right) \,.
\end{equation*}

In this basis the dual nef cone is generated by $\cN_1 = (1,0)$ and $\cN_2 = (2,1)$. There is only one subset of $\{\cM_1,\cM_2\}$ on which the intersection form is negative definite, which is $\{\cM_2\}$. As such, apart from the nef cone, there is only one Zariski chamber, $\Sigma_{2}$, obtained by extending the face $\langle\cN_2\rangle_{\mathbb R_{\geq 0}}$ along $\cM_2$. The other face of the nef cone, $\langle\cN_1\rangle_{\mathbb R_{\geq 0}}$ is on the boundary of the effective cone, and is not covered by our present cohomology discussion. We then obtain the following formula for the zeroth cohomology of effective line bundles.

\begin{equation*}
\begin{tabular}{L | L }
\Sigma &~ h^0(S, \cO_\surf(D)) \\
\hline
\Sigma_\mathrm{nef} \setminus \langle \cM_1 \rangle_{\mbb{R}_{\geq0}} &~{\rm ind}\big(S, \cO_\surf(D)\big) \\
\Sigma_2 \setminus \langle \cM_2 \rangle_{\mbb{R}_{\geq0}} &~{\rm ind}\Big(S, \cO_\surf\big(D-\ceil{-\frac{1}{2}D\cdot\cM_2}\cM_2\big)\Big) \\
 \langle \cM_2 \rangle_{\mbb{R}_{\geq0}} & ~h^0\big(S,\cO_S) = 1 \\ 
\end{tabular}
\end{equation*}

\noindent In fact in this present simple case of a Weierstrass K3 surface, it is straightforward to find formulae describing cohomology by using the Leray spectral sequence to lift those on the base $\mbb{P}^1$, analogously to the discussion in Section~\ref{sec:lift_coh} below for three-folds. In particular, we can then find the formula for the zeroth cohomology on the remaining region, $\langle \cM_1 \rangle_{\mbb{R}_{\geq0}}$, which we include here for completeness to complement the above table.
\begin{equation*}
\begin{tabular}{L | L }
\Sigma &~ h^0(S, \cO_\surf(D)) \\
\hline
\langle \cM_1 \rangle_{\mbb{R}_{\geq0}} &~ \mathrm{ind}\big(S,\cO_\surf (D+\cM_2)\big) \\ 
\end{tabular}
\hspace{1.3cm}
\end{equation*}

As in the previous example, we can recast these formulae in a basis. Writing $D=k_1\cM_1+k_2\cM_2$, the index is
${\rm ind}\big(S, \cO_\surf(D)\big) = 2+\tfrac{1}{2}D\cdot D = 2+(k_1-k_2)k_2 \,.$
\noindent The formulae above then become the following.
\begin{equation}
\begin{tabular}{L | L }
\Sigma &~ h^0(S, \cO_\surf(D)) \\
\hline
\Sigma_\mathrm{nef} \setminus \langle \cM_1 \rangle_{\mbb{R}_{\geq0}} &~ 2+(k_1-k_2)k_2~ \\
\Sigma_2 \setminus \langle \cM_2 \rangle_{\mbb{R}_{\geq0}} &~2+\floor{\tfrac{1}{2}k_1}\ceil{\tfrac{1}{2}k_1} \\
\langle \cM_1 \rangle_{\mbb{R}_{\geq0}} &~ 1+k_1 \\
\langle \cM_2 \rangle_{\mbb{R}_{\geq0}} & ~1
\end{tabular}
\end{equation}

\begin{figure}[t]
\begin{center}
	\includegraphics[scale=.5]{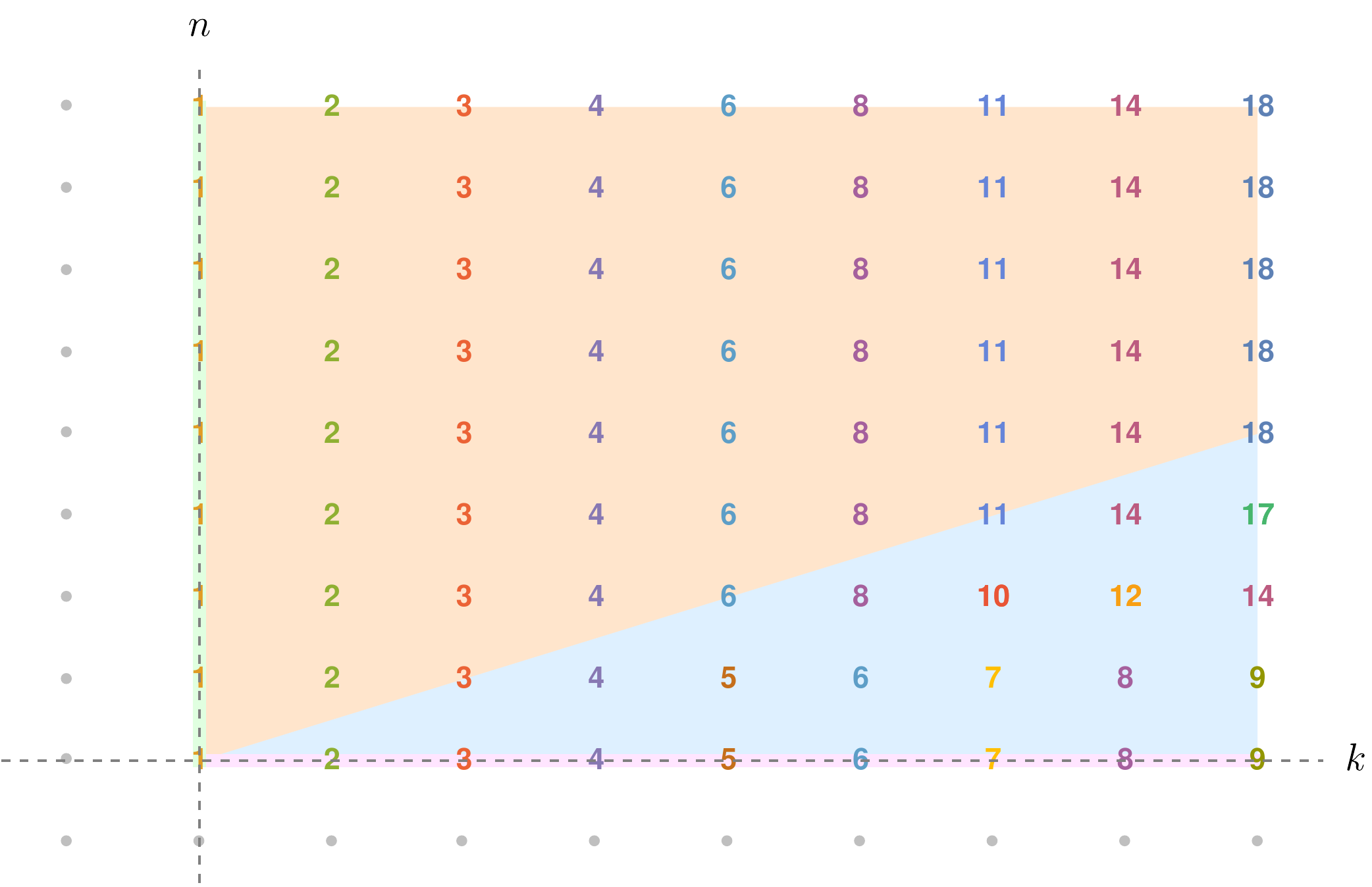} \\
	\raisebox{.1in}{\includegraphics[scale=.55]{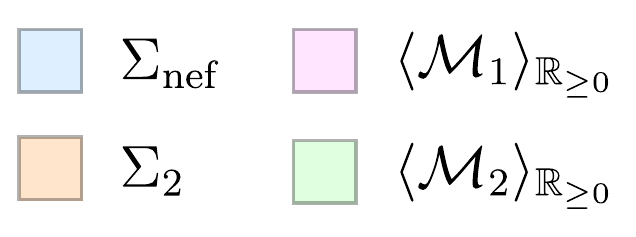}}\\
	\capt{5.8in}{fig:k3_2dex_regions}{Zeroth line bundle cohomologies $\mc{O}_{\surf}\big(n\fibsec+k\fibprj^*(H)\big)$ on the K3 surface $\surf$ given by a generic Weierstrass model with a single section. We show the cohomologies and the regions where different formulae apply.}
\end{center}
\end{figure}

It is not a surprising fact that the formula along the $\langle \cM_1 \rangle_{\mbb{R}_{\geq0}}$ is linear, rather than quadratic in $k_1$. This comes in agreement with the holomorphic Morse inequalities (see Remark 2.2.20 in Ref.~\cite{lazarsfeld2004positivity1}), according to which on a projective variety $X$ of dimension $n$, if $D$ is a nef divisor, then for every integer $q\in[0,n]$ one has $h^q(X, \cO_X(m D))\leq O(m^{n-1})$.

\newpage
\section{Cohomology chambers on elliptic Calabi-Yau three-folds}

In the previous section we obtained formulae for line bundle cohomology on surfaces. One immediate application is to lift these formulae to higher-dimensional manifolds which use these surfaces as building blocks. 
An obvious construction of this kind is to consider fibrations over the surfaces studied above. In these constructions, the lift of cohomologies can be computed straightforwardly through the Leray spectral sequence. A class of fibrations which are both simple and have many applications in string theory are elliptically fibered Calabi-Yau three-folds, and we study these in the present section. We will consider the simplest setting, in which the generic fibration is smooth, since in this case the lift by the Leray spectral sequence is straightforward\footnote{When the elliptic fibration is singular, one can often resolve the singularities to give a smooth Calabi-Yau. These singularities are very important in the context of F-theory, where they determine gauge and matter fields, as well as couplings. However in this case it is more involved to lift cohomologies with the Leray spectral sequence.}.

\subsection{Elliptically fibered Calabi-Yau three-folds}
\label{sec:ell_tf}

\noindent In this section we provide a brief summary on the construction and properties of elliptically fibered Calabi-Yau three-folds. This is provided as a reminder and can be skipped by a reader familiar with this material.

\smlhdg{Weierstrass models}

\noindent An elliptically fibered manifold $\cy{}$ consists of a fibration $\fibprj : \cy{} \to \base{}$ of an elliptic curve over a base manifold~$\base{}$, with a section $\fibsec : \base{} \xhookrightarrow{} \cy{}$ that embeds the base into the total space. 

We focus on elliptic fibrations which can be constructed with a Weierstrass model. Note that every elliptic fibration is birationally equivalent to a Weierstrass one. In a Weierstrass model, the elliptic curve $E(b)$ over a point $b \in \base{}$ is described as a hypersurface in an ambient space. A useful choice for the ambient space, for reasons that will become clear in a moment, is the weighted complex projective space $\mbb{P}_{231}[x:y:z]$. The elliptic curve is defined by a degree six polynomial, which can by coordinate redefinitions be written in the form
\be
y^2 = x^3 + f(b) x z^4 + g(b)z^6 \,.
\label{eq:weierstrass}
\ee
Here $f(b)$ and $g(b)$ are parameters that define the elliptic curve $E(b)$.

A fibration of the elliptic curve is inherited if the space $\mbb{P}_{231}$ is fibered over the base. In order to fiber $\mbb{P}_{231}$ over the base, the homogeneous coordinates $x,y,z$ are taken as sections of powers of a line bundle $\lbb$ on the base: $x \in \Gamma(\lbb^2)$, $y \in \Gamma(\lbb^3)$, $z \in \Gamma(\mc{O}_{\base{2}})$. From the Weierstrass equation this means the parameters $f$ and $g$ must vary over the base as sections $f \in \Gamma(\lbb^4)$, $g \in \Gamma(\lbb^6)$.

Upon fibering, different choices of the ambient space for the elliptic curve are not equivalent, but rather determine the existence of sections of the fibration. In particular, the choice of $\mbb{P}_{231}$ ensures the existence of a single section, given by $z=0$, as one can verify. Note that $\mbb{P}_{231}$ is the Gorenstein Fano toric surface $F_{10}$ whose ray diagram is shown in \fref{fig:refl_poly}. Making a different choice among these gives a different structure of sections. See for example Table~2 of Ref.~\cite{Braun:2013nqa}, or the earlier Ref.~\cite{grassi2012weierstrass}.

One can check using adjunction that for the resulting manifold to be Calabi-Yau, the defining line bundle $\lbb$ must be chosen to be the anti-canonical bundle of the base, $\lbb = K_{\base{}}^{-1}$. This leads to the requirement that the anti-canonical bundle $K_{\base{}}^{-1}$ must have sections, which constrains the possible choices for the base spaces. This condition is true on for example toric surfaces and generalised del Pezzo surfaces. This is also true for K3 surfaces, but in this case the fibration is trivial, and the three-fold is a product $\mathrm{K3} \times \mathrm{T}^2$.

\smlhdg{Smoothness}

\noindent The elliptic curve $E(b)$ described in Weierstrass form in Equation~\eqref{eq:weierstrass} can be singular, depending on the values of $f(b)$ and $g(b)$. In particular, one can check that the elliptic curve $E(b)$ is singular if the discriminant $\Delta(b)$ vanishes, where
\be
\Delta(b) = 4f(b)^3 + 27g(b)^2 \,.
\ee
In the fibration of the elliptic curve, the discriminant varies over the base as a section $\Delta \in \Gamma\big(\lbb^{12}\big)$, and the elliptic curve is singular over the zero locus of this section. In the case of an elliptic fibration giving rise to a Calabi-Yau three-fold, $\Delta \in \Gamma\big(K_{\base{}}^{-12}\big)$.

Importantly however, the elliptic fibration is often smooth despite the singular elliptic fibers. In particular, only severe singularities of the fiber give rise to singularities of the fibration, with the severity of the fiber singularity essentially determined by the vanishing orders of $f$, $g$, and $\Delta$. The singularities in the fibration can occur over loci in the base of various codimension.

In the case of singularities over codimension one loci in the base, the types of singularity in the fibration are summarised in Table~\ref{tab:kod_class}, which is reproduced from Ref.~\cite{Morrison:1996pp}. Singularities over loci of higher codimension are more complicated. More details can be found for example in Ref.~\cite{weig2018tasi} and references contained therein.

\begin{center}
\begin{table}[h]
\begin{center}
\begin{tabular} { | c | c | c | c | c | }
\hline
ord($f$) & ord($g$) & ord($\Delta$) & fiber-type & singularity-type \\
\hline \hline
$\geq0$ & $\geq0$ & 0 & smooth & none \\
\hline
0 & 0 & n & $I_n$ & $A_{n-1}$ \\
\hline
$\geq1$ & 1 & 2 & $II$ & none \\
\hline
1 & $\geq2$ & 3 & $III$ & $A_1$ \\
\hline
$\geq2$ & 2 & 4 & $IV$ & $A_2$ \\
\hline
2 & $\geq3$ & n+6 & $I_n^*$ & $D_{n+4}$ \\
\hline
$\geq2$ & 3 & n+6 & $I_n^*$ & $D_{n+4}$ \\
\hline
$\geq3$ & 4 & 8 & $IV^*$ & $E_6$ \\
\hline
3 & $\geq5$ & 9 & $III^*$ & $E_7$ \\
\hline
$\geq4$ & 5 & 10 & $II^*$ & $E_8$ \\
\hline
\end{tabular}
\end{center}
\caption{\itshape\small Classification of singularities from fiber degeneracy over codimension one loci in the base.}
\label{tab:kod_class}
\end{table}
\end{center}

In the case of a Weierstrass model with $\lbb = -K_{\base{2}}$, chosen to give rise to a Calabi-Yau three-fold $\cy{3}$, a necessary condition for a smooth fibration to exist is that the base $\base{2}$ does not contain curves $C$ of self-intersection $C^2 < -2$. To see this, first note that the self-intersection of a curve $C$ on a surface $\base{2}$ is related to the genus $g_C$ of the curve by $(K_{\base{2}}+C) \cdot C = 2g_C-2$, so that
\be
-K_{\base{2}} \cdot C = C \cdot C - 2g_C+2 \,.
\ee
Since $g_C \geq 0$, the right-hand side is negative if $C \cdot C < -2$. But this negative intersection indicates the presence of $C$ in every element of the complete linear system $\cls{-K_{\base{2}}}$, by the argument in Section~\ref{sec:bas_defs} above. The amount of $C$ detected is even greater for multiples $\cls{-mK_{\base{2}}}$, and specifically any element takes the form
\be
\cls{-mK_{\base{2}}} \ni \left\lceil \frac{m(-K_{\base{2}} \cdot C)}{C \cdot C} \right\rceil C + \ldots
\ee
One can check that if $C \cdot C < -2$, then in the cases $m=4$ and $m=6$, the coefficient on the right is at least two. That is, any sections $f \in \Gamma(K_{\base{2}}^{-4})$ and $g \in \Gamma(K_{\base{2}}^{-6})$ must have vanishing orders of at least two over the curve $C$. Glancing at Table~\ref{tab:kod_class}, we see that in this case the three-fold is singular over this locus.

This condition constrains the set of base spaces which can give rise to smooth fibrations. This condition holds on the generalised del Pezzo surfaces. Among the toric surfaces, the only cases satisfying this condition are the 16 cases which are also generalised del Pezzo - on any other toric surfaces the Weierstrass model will not be generically smooth.

\smlhdg{Properties of the three-fold}

\noindent When the Weierstrass model is smooth, the properties of the elliptically fibered manifold $\cy{}$ follow straightforwardly from those of the base $\base{}$, as we now discuss.

Associated to the projection map $\fibprj: \cy{} \to \base{}$ is a pullback map $\fibprj^*$, which lifts bundles on the base $\base{}$ to bundles on the total space $\cy{}$, and lifts divisors to divisors. This gives an injection $\fibprj^*: \mathrm{div}(\base{}) \to \mathrm{div}(\cy{})$. In addition to the pullback divisors, there is also the section\footnote{We abuse notation and write $\fibsec$ both for the map that embeds the base into the total space and for the image.} $\fibsec$. In particular, a basis of the Picard group on $\cy{}$ is given by $\{\fibsec \,, \fibprj^*( \divb_i) \}$, where $\{\divb_1,\divb_2,\ldots\}$ is a basis of the Picard group on the base. We write $\divtf_0 \equiv \fibsec$ and $\divtf_i \equiv \fibprj^*(\divb_i)$. A line bundle $\lb$ on $X$ can hence be specified by an integer $n$ along the section, and a pullback of a line bundle $\lbb$ on the base,
\be
\lb = \mc{O}_{\base{2}}(n\fibsec) \otimes \fibprj^* \lbb  \,.
\label{eq:gen_lb}
\ee
The cone $\mathrm{Eff}(X)$ of effective divisors on $X$ is trivially related to the Mori cone $\cM(B)$ on the base,
\be
\mathrm{Eff}(X) = \langle \pi^*(\cM(B)) \,, \fibsec \rangle \,.
\ee
This is clear for example from the Leray spectral sequence below. 

As well as the pullback there is the inclusion $\fibsec: \base{} \xhookrightarrow{} \cy{}$. In the case of a three-fold $\cy{3}$ over a complex surface $\base{2}$, this sends divisors on the base to curves in the total space. These give rise to curve classes $\fibsec(D_i)$, which together with the fiber class $F$ give a basis of curves $\{\crvtf_0 \,, \crvtf_i \} \equiv \{ \fibcls \,, \fibsec(D_i) \}$. If the anti-canonical bundle of the base is nef, as in our cases, then like the effective cone the Mori cone $\cM(\cy{3})$ of the three-fold is trivially related to that of the base (see for example the argument in Ref.~\cite{Donagi_1999}),
\be
\cM(X) = \langle \fibsec(\cM(B)) \,, F \rangle \,.
\ee

The intesections between the above curves and divisors are
\be
\begin{tabular}{ C | C C }
				& \sigma								& \fibprj^*(D)			\\  \hline
\sigma			& \sigma(K_{\base{2}})					& \sigma(D)			\\
\fibprj^*(D')		& \sigma(D')							&	 (D \cdot D')F		\\
\end{tabular}
\hspace{1.5cm}
\begin{tabular}{ C | C C }
				& \fibcls								& \sigma(D)			\\ \hline
\sigma			& 1									& K_{\base{2}} \cdot D	\\
\fibprj^*(D')		& 0									& D \cdot D'			\\
\end{tabular}
\ee
The nef cone is the dual of the Mori cone with respect to these intersections. The triple intersection numbers are
\be
d_{000} = K_{\base{2}} \cdot K_{\base{2}} \,, \quad d_{00i} = K_{\base{2}} \cdot D_i \,, \quad d_{0ij} = D_i \cdot D_j \,, \quad d_{ijk} = 0 \,,
\ee
where $d_{0ij} = \divtf_0 \cdot \divtf_i \cdot \divtf_j = \fibsec \cdot \fibsec \cdot \fibprj^*(D_i)$, etc. The second Chern class and the Euler number are \cite{Friedman:1997yq}
\be
\begin{gathered}
c_2(\cy{3}) = \fibprj^*\big(c_2(\base{2})\big) + 11\fibprj^*\big(c_1(\base{2})^2\big) + 12\fibsec\big(c_1(\base{2})\big) \,, \\
\chi(\cy{3}) = -60 \, K_{\base{2}} \cdot K_{\base{2}} \,.
\end{gathered}
\ee

\subsection{Lifting base cohomologies}
\label{sec:lift_coh}

\smlhdgnogap{The Leray spectral sequence}

\noindent On a fibration $\pi: \cy{} \to \base{}$, the cohomology of a bundle $V$ on the total space $\cy{}$ can be computed in terms of cohomologies on the base $\base{}$. In particular, the relevant objects on the base are the cohomologies of the higher direct images $R^q\pi_*V$ of the bundle, where $q = 0 , 1 , \ldots$, which are in general sheaves. Note that $R^0\pi_*V \equiv \pi_*V$ is simply the pushforward. While we will not compute them explicitly, we note that $R^q\pi_*V$ is equal to the sheaf associated to the presheaf
\be
U \to H^q\big(\pi(U),V|_{\pi^{-1}(U)}\big)
\label{eq:hidirim_sheafcoh}
\ee
where $U$ are the open sets on the base, and we refer the reader to Chapter~III.8 of Ref.~\cite{Hartshorne1977} for details.

The relation between the cohomology of $V$ and the cohomologies of the higher direct images is provided by the Leray spectral sequence. The definition of the Leray spectral sequence begins with the definition of the second page $E_2^{p,q}$ as
\be
E_2^{p,q} = H^p ( \base{}, R^q\pi_*V ) \,.
\ee
Within the $r$th page $E_r^{p,q}$ where $r \geq 2$, there are maps $\dd_r : E_r^{p,q} \to E_r^{p+r,q-r+1}$ such that $\dd_r^2=0$. The $(r+1)$th page is defined by the cohomology of these maps on the $r$th page,
\be
E_{r+1}^{p,q} = \frac{\mathrm{ker}(\dd_r : E_r^{p,q} \to E_r^{p+r,q-r+1})}{\mathrm{im}(\dd_r : E_r^{p-r,q+r-1} \to E_r^{p,q})} \,.
\label{eq:ler_spec_rel}
\ee
When this iterative process stabilises, so that $E_r^{p,q} = E_n^{p,q}$~$\forall r \geq n$ for some $n$, we write $E_n^{p,q} \equiv E_\infty^{p,q}$, and the cohomologies of the bundle $V$ on the total space $X$ are given by
\be
H^i ( \cy{} , V ) = \bigoplus_{p+q=i} E_\infty^{p,q} \,.
\label{eq:gen_ler_res}
\ee

\smlhdg{Higher direct images}

\noindent We recall from Section~\ref{sec:ell_tf} above that on a generically smooth Weierstrass Calabi-Yau three-fold $\cy{3}$, a line bundle $\lb$ can be written uniquely as
\be
\lb = \mc{O}_{\cy{3}}(n\fibsec) \otimes \fibprj^* \lbb \,,
\ee
where $\lbb$ is a line bundle on the base $\base{2}$, and $\fibsec$ is the section of the fibration. For such a tensor product of bundles, the direct and higher direct images can be simplified by use of the projection formula (see for example Chapter~III.8 in Ref.~\cite{Hartshorne1977}),
\be
R^i \fibprj_* \lb =  R^i\fibprj_* \left( \mc{O}_{\cy{3}}(n\fibsec) \right) \otimes \lbb  \,.
\label{eq:prj_frm}
\ee
Hence, all higher direct images are determined by knowledge of those of $\mc{O}_{\cy{3}}(n\fibsec)$. Recall also that the higher direct images vanish trivially unless $i=0,1$, so we require only $\fibprj_*\mc{O}_{\cy{3}}(n\fibsec)$ and $R^1 \fibprj_*\mc{O}_{\cy{3}}(n\fibsec)$, for all $n$. These have been worked out elsewhere (see for example Appendix~C of Ref.~\cite{Donagi:2004ia}), and we collect the results below\footnote{That there are special cases for $\fibprj_*\mc{O}_{\cy{3}}(\sigma)$ and $R^1\fibprj_*\mc{O}_{\cy{3}}(-\sigma)$ can be traced to the fact there are no meromorphic functions on the torus with a single pole. We also note that the above formulae reflect `relative duality' (see for example Chapter~III.12 of Ref.~\cite{barth}), which in the Calabi-Yau case states that $R^1\fibprj_*V = (\fibprj_* V^*)^* \otimes K_{\base{2}}$ for any vector bundle $V$.}.
\be
\begin{aligned}
\fibprj_*\mc{O}_{\cy{3}}(n \fibsec) &= 
\begin{cases}
0 																					& \mathrm{for}~n<0 \,, \\
\mc{O}_{\base{2}} 																		& \mathrm{for}~n=0,1 \,, \\
\mc{O}_{\base{2}} \oplus K_{\base{2}}^2 \oplus K_{\base{2}}^3 \oplus \ldots \oplus K_{\base{2}}^n	& \mathrm{for}~n\geq2 \,, \\
\end{cases}
 \\
R^1\fibprj_*\mc{O}_{\cy{3}}(n \fibsec) &= 
\begin{cases}
0 																					& \mathrm{for}~n>0 \,, \\
K_{\base{2}} 																			& \mathrm{for}~n=-1,0 \,, \\
K_{\base{2}} \oplus K_{\base{2}}^{-1} \oplus K_{\base{2}}^{-2} \oplus \ldots \oplus K_{\base{2}}^{1+n}& \mathrm{for}~n\leq-2 \,. \\
\end{cases}
\end{aligned}
\label{eq:hidirim_list}
\ee
Here $0$ is the rank zero bundle, which gives $0$ upon taking the tensor product with any other bundle, and for which all cohomologies are trivial.

For the below, it will be important that $R^i \fibprj_* L$ is always a sum of line bundles. As the $i$th cohomology of a sum of line bundles is the sum of the $i$th cohomologies of the line bundles, it is straightforward to determine $H^i(\base{2},R^j \fibprj_* \lb )$ and hence the second page $E_2^{p,q}$ given knowledge of line bundle cohomology on the base.

\smlhdg{Degeneration of the Leray spectral sequence}

\noindent In the present context, the Leray spectral sequence simplifies dramatically.

\begin{prp}
On a generically smooth Weierstrass Calabi-Yau three-fold $\cy{3}$ over a smooth complex projective base $\base{2}$, the Leray spectral sequence for a line bundle $\lb$ on $\cy{3}$ degenerates at the second page, i.e.\ $E_2^{p,q} = E_\infty^{p,q}$, so that
\be
H^i ( \cy{3} , \lb ) = \bigoplus_{p+q=i} H^p ( \base{2}, R^q\pi_*\lb ) \,.
\ee
\end{prp}
\begin{proof}
First note from the dimension of the base the trivial vanishings $H^p\big(\base{2},R^q\fibprj_*\lb\big) = 0$ unless $p=0,1,2$. Additionally, glancing at the expression in Equation~\eqref{eq:hidirim_sheafcoh}, from the dimension of the fiber $R^q\fibprj_*\lb = 0$ unless $q=0,1$. Inserting these into Equation~\eqref{eq:ler_spec_rel}, one finds that every element in the third page is trivially related to the second page, $E_3^{p,q} = E_2^{p,q}$, except for two cases: $E_3^{0,1}$ and $E_3^{2,0}$. The relations for these remaining terms are
\be
E_{3}^{0,1} = \mathrm{ker}(\dd_2 : E_2^{0,1} \to E_2^{2,0}) ~~\mathrm{and}~~ E_{3}^{2,0} = \frac{E_2^{2,0}}{\mathrm{im}(\dd_2 : E_2^{0,1} \to E_2^{2,0})} \,.
\ee
It is easy to check that if either one of $E_2^{0,1}$ and $E_2^{2,0}$ is zero, then these final relations too are trivial, i.e.\ $E_3^{0,1}=E_2^{0,1}$ and $E_3^{2,0}=E_2^{2,0}$.

The line bundle $\lb$ can be written $\lb = \mc{O}_{\cy{3}}(n\fibsec) \otimes \fibprj^* \lbb$ where $\lbb$ is a line bundle on the base. Consider the cases $n \neq 0$. From Equations~\eqref{eq:prj_frm} and \eqref{eq:hidirim_list}, we see immediately that either $\pi_*\lb = 0$ or $R^1\pi_*\lb = 0$. So in these cases either $E_2^{0,1}=0$ or $E_2^{2,0}=0$. The only remaining case to check is $n=0$. Here we have
\be
\begin{aligned}
E_2^{0,1}|_{n=0} &~=~ H^0\big(\base{2},R^1\fibprj_*(\mc{O}_{\base{2}}\otimes\mc{L})\big)					&&=~~ H^0(\base{2},\lbb \hspace{5pt} \otimes K_{\base{2}}) \,, \\
E_2^{2,0}|_{n=0} &~=~ H^2\big(\base{2},\fibprj_*(\mc{O}_{\base{2}}\otimes\mc{L})\big)	= H^2(\base{2},\lbb ) 	&&=~~ H^0(\base{2},\dual{\lbb} \otimes K_{\base{2}} ) \,,
\end{aligned}
\ee
where in the final equality we used Serre duality. The first term is non-zero only when $\lbb \otimes K_{\base{2}}$ is in the effective cone, while the second is non-zero only when $\lbb^* \otimes K_{\base{2}}$ is in the effective cone. For both to be non-zero, the Mori cone $\cM$ must overlap its reflection $-\cM$ through the origin after being shifted by $-2K_{\base{2}}$. But $-K_{\base{2}}$ is effective for the Weierstrass model to exist. If the Mori cone is strongly convex, an effective shift will separate $\cM$ from $-\cM$ without overlap. This holds on a projective surface, which proves the result.
\end{proof}

\noindent More explicitly, the cohomology of $L$ on the three-fold $X_3$ is given by the following relations
\be
\begin{aligned}
H^0(\cy{3},\lb) &= H^0(\base{2},\fibprj_* \lb) \,, \\
H^1(\cy{3},\lb) &= H^1(\base{2},\fibprj_* \lb) \oplus H^0(\base{2},R^1\fibprj_* \lb) \,, \\
H^2(\cy{3},\lb) &= H^2(\base{2},\fibprj_* \lb) \oplus H^1(\base{2},R^1\fibprj_* \lb) \,, \\
H^3(\cy{3},\lb) &= H^2(\base{2},R^1\fibprj_* \lb) \,.
\end{aligned}
\label{eq:final_coh_rels}
\ee

\smlhdg{Explicit expressions}

\noindent Given the relations~\eqref{eq:final_coh_rels}, it only remains to plug in the expressions from Equation~\eqref{eq:hidirim_list} for the higher direct images. Note it is sufficient to determine only the zeroth and first cohomologies, since the second and third are trivially related to these by Serre duality. Again writing $\lb = \mc{O}_{\cy{3}}(n\fibsec) \otimes \fibprj^* \lbb$, one has
\be
\begin{gathered}
h^0(\cy{3},\lb) = h^0(\base{2},\fibprj_*L) ~~ \mathrm{and} ~~ 
h^1(\cy{3},\lb) = h^0(\base{2},R^1\fibprj_*L) + h^1(\base{2},\fibprj_*L) \,,
\\ \mathrm{where}
\\
\begin{tabular}{L C L L}
h^0(\base{2},\fibprj_*L) &=& 
\begin{cases}
0 																		\\
h^0(\base{2},\lbb)															\\
h^0(\base{2},\lbb) + \sum_{i=2}^n h^0(\base{2},\lbb \otimes K_{\base{2}}^i)			\\
\end{cases}
	&
	\begin{tabular}{L}
	\mathrm{for}~n<0 \\
	\mathrm{for}~n=0,1 \\
	\mathrm{for}~n\geq2
	\end{tabular}
\vspace{10pt} \\
h^0(\base{2},R^1\fibprj_*L) &=&
\begin{cases}
0 																							\\
h^0(\base{2},\lbb\otimes K_{\base{2}})															\\
h^0(\base{2},\lbb\otimes K_{\base{2}}) + \sum_{i=2}^{-n} h^0(\base{2},\lbb \otimes K_{\base{2}}^{1-i})	\\
\end{cases}
	&
	\begin{tabular}{L}
	\mathrm{for}~n>0 \\
	\mathrm{for}~n=-1,0 \\
	\mathrm{for}~n\leq-2
	\end{tabular}
\vspace{10pt} \\
h^1(\base{2},\fibprj_*L) &=& 
\begin{cases}
0 																		\\
h^1(\base{2},\lbb)															\\
h^1(\base{2},\lbb) + \sum_{i=2}^n h^1(\base{2},\lbb \otimes K_{\base{2}}^i)			\\
\end{cases}
	&
	\begin{tabular}{L}
	\mathrm{for}~n<0 \\
	\mathrm{for}~n=0,1 \\
	\mathrm{for}~n\geq2
	\end{tabular}
\end{tabular}
\end{gathered}
\label{eq:terms_in_dir_sums}
\ee

When line bundle cohomologies on the base are described by simple formulae along the lines of Section~\ref{sec:bun_coh}, the above give expressions for line bundle cohomology on the three-fold. These expressions provide a simple and fast method to determine any line bundle cohomology. We note that, while these expressions can be guessed from raw data, for example by equation fitting or machine learning methods, as pursued in Refs.~\cite{Klaewer:2018sfl,Brodie:2019dfx,Larfors:2019sie}, the present approach has the advantages of giving a proof, and the knowledge there are no missed edge-cases.

\subsection{Example}
\label{sec:thfld_ex}

On a given base $\base{2}$, we expect the expressions in Equation~\eqref{eq:terms_in_dir_sums} to simplify substantially, to give compact regions and formulae describing all line bundle cohomology on the three-fold, analogous to the surface case. Here we consider the simplest example, of a Weierstrass three-fold with base a projective plane, $\base{2} = \mbb{P}^2$.

\smlhdg{Properties of the three-fold}

\noindent We recall the discussion in Section~\ref{sec:ell_tf}, of the properties of the Weierstrass three-fold for a given base. In the present case, the base space $\base{2} = \mbb{P}^2$ has Picard number 1, with the Picard lattice spanned by the hyperplane class $H$. We can write a line bundle $\lbb$ on $\base{2}$ as $\lbb = \mc{O}_{\base{2}}(k H)$, where $k \in \mbb{Z}$. The anti-canonical bundle is $-K_{\base{2}} = \mc{O}_{\base{2}}(3H)$. The intersections are determined by $H \cdot H = 1$, and the Mori cone is generated by $H$.

The properties of the elliptic three-fold are then as follows. The divisor and curve bases are respectively $\{\divtf_0 \,, \divtf_1 \}  \equiv \{\fibsec \,, \fibprj^*( H ) \}$ and $\{\crvtf_0 \,, \crvtf_1 \} \equiv \{ \fibcls \,, \fibsec(H) \}$. The intersections on the three-fold are
\be
\begin{tabular}{ C | C C }
				& \sigma							& \fibprj^*(H)				\\  \hline
\sigma			& \sigma(-3H)						& \sigma(H)				\\
\fibprj^*(H)		& \sigma(H)						& F						\\
\end{tabular}
\hspace{1.5cm}
\begin{tabular}{ C | C C }
				& \fibcls							& \sigma(H)				\\ \hline
\sigma			& 1								& -3						\\
\fibprj^*(H)		& 0								& 1						\\
\end{tabular}
\ee
The triple intersection numbers follow in the above divisor basis as, up to permutations,
\be
d_{000} = 9 \,, ~~ d_{001} = -3 \,, ~~ d_{011} = 1 \,, ~~ d_{111} = 0 \,.
\ee
The generators of the Mori cone $\moricn(\cy{3})$ and of the dual nef cone $\nefcn(\cy{3})$, are
\be
\{\moricn_i(\cy{3})\} = \{ \fibcls \,, \fibsec(H) \} \,, \quad \{\nefcn_j(\cy{3})\} = \{ \fibsec + 3 \fibprj^*(H) \,, \fibprj^*(H) \} \,.
\ee
Finally, the second Chern class is $c_2 = 102\fibcls + 36\,\fibsec(H)$. With the triple intersection numbers and the second Chern class, we can also write the index of a line bundle. On a Calabi-Yau three-fold, the index of a line bundle is given by
\be
\ind\big(\cy{3}, \mc{O}_{\cy{3}}(\divtf) \big) = \tfrac{1}{6}\left( \divtf^3 + \tfrac{1}{2} \divtf \cdot c_2 \right) \,,
\ee
so in our case we have
\be
\ind\big(\cy{3}, \mc{O}_{\cy{3}}(\cy{3},n\fibsec + k\fibprj^*(H)) \big) =  \tfrac{3}{2}n^3 - \tfrac{3}{2} n^2k + \tfrac{1}{2} nk^2 - \tfrac{1}{2} n +  3 k \,.
\label{eq:ex_3f_ind}
\ee

In the following discussion, we write $\lb = \mc{O}_{\cy{3}}\big(n\fibsec + k\fibprj^*(H)\big)$ for a general line bundle $\lb$ on the three-fold~$\cy{3}$. Additionally, we will use the shorthand notation $\mc{O}_{\cy{3}}(n,k)$ for such a line bundle.

\smlhdg{Zeroth cohomology}

\noindent On the projective plane $\mbb{P}^2$, the zeroth cohomology of a line bundle is given by the Bott formula as the binomial coefficient $h^0\big(\mc{O}_{\mbb{P}^2}(kH)\big) = \binom{k+2}{2}$. Hence from Equation~\eqref{eq:terms_in_dir_sums}
\be
h^0(\cy{3},\lb) = 
\begin{cases}
\theta(k)\binom{k+2}{2}											& \mathrm{for}~n=0,1 \,, \\
\theta(k)\binom{k+2}{2} + \sum_{i=2}^n \theta(k-3i)\binom{k-3i+2}{2}		& \mathrm{for}~n\geq2 \,, \\
\end{cases}
\ee
and $h^0(\cy{3},\lb)=0$ otherwise. Here $\theta(\,\cdot\,)$ is a step function, equal to one for $x \geq 0$ and zero otherwise.

We plot the numerical values in \fref{fig:p2tf_h0_fig}. From the expressions or the figure it is clear the effective cone is simply the positive quadrant, $n,k \geq 0$. This cone further naturally splits into two regions (at least for $n \geq 2$): when $k \geq 3n$, all of the step functions in the sum are satisfied, while when $k < 3n$, the sum is cut off by the step functions. In the latter region, notably there will no dependence on $n$.

\begin{figure}[h]
\begin{center}
	\includegraphics[scale=.55]{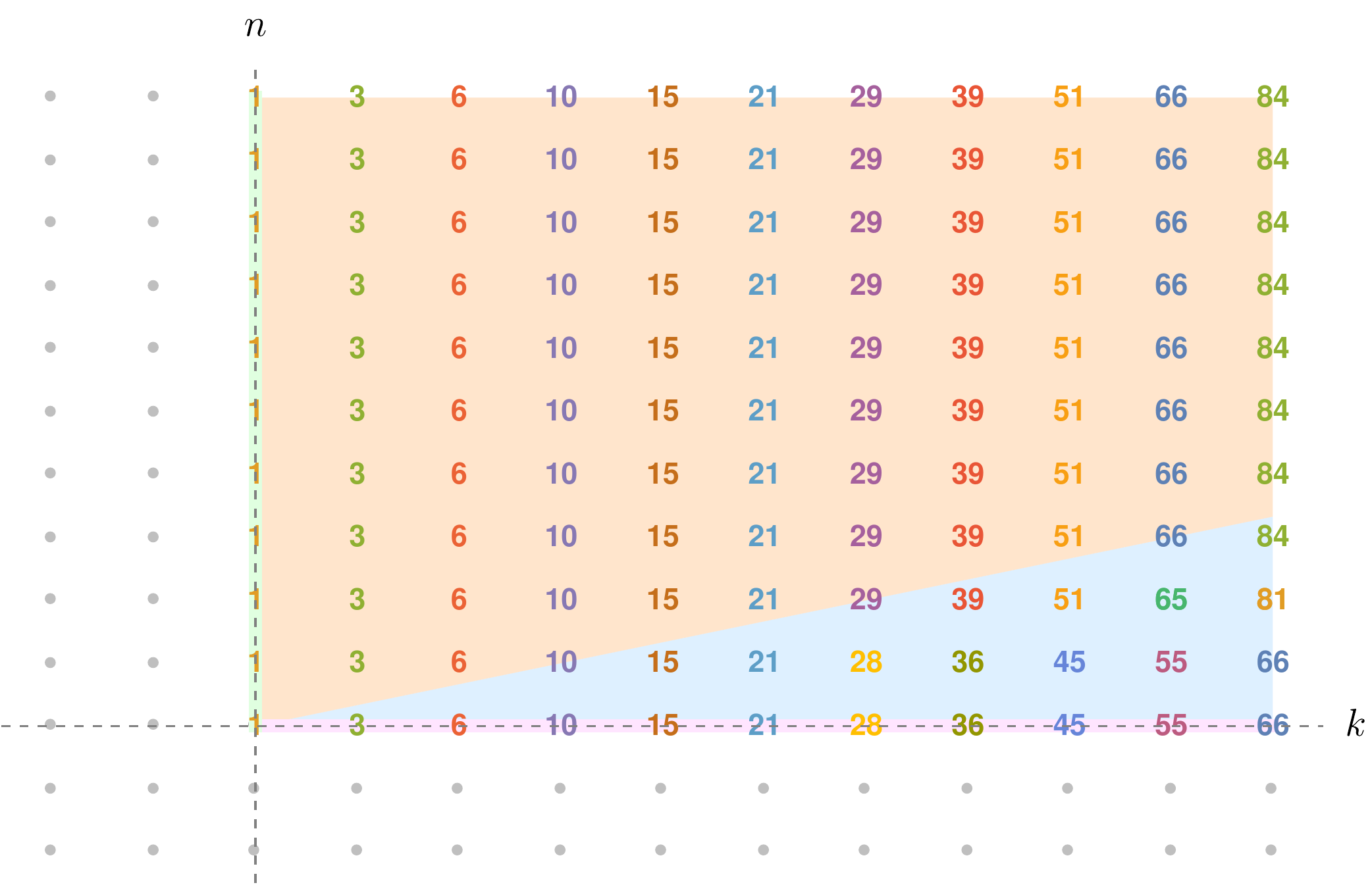} \\
	\includegraphics[scale=.55]{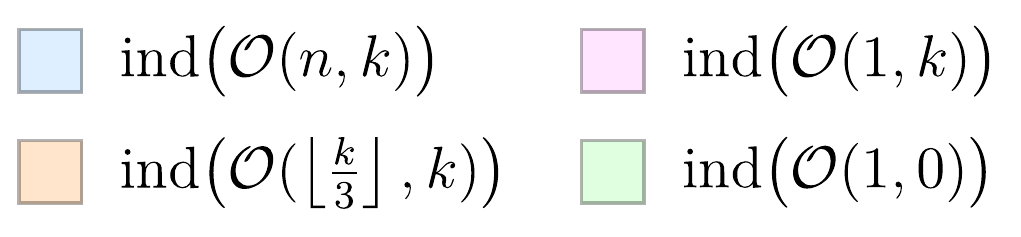}\\[4pt]
	\capt{5.8in}{fig:p2tf_h0_fig}{Zeroth line bundle cohomologies $h^0 \big( \mc{O}(n,k) \big)$ on the Weierstrass Calabi-Yau with base $\mbb{P}^2$. We show the cohomologies and the regions where different formulae apply. Here $\mc{O}(n,k)$ is shorthand for $\mc{O}_{\cy{3}}\big(n\fibsec+k\fibprj^*(H)\big)$.}
	\end{center}
\end{figure}

In the $k \geq 3n$ region of the effective cone, the above sum that appears for $n \geq 2$ has unit coefficients, giving
\be
\sum_{i=2}^n \binom{k-3i+2}{2} = \frac{1}{2}(n-1)\left(3(n+1)(n-k)+k^2+2\right) \,.
\label{eq:binom_sum_h0}
\ee
Note that this expression happens to be zero when $n=1$ (but not when $n=0$), so we may include it in the $n=1$ case as well. The full expression for $h^0(\cy{3},\lb)$ in the $n > 0 \,, k \geq 3n$ region is then
\be
\binom{k+2}{2} + \sum_{i=2}^n \binom{k-3i+2}{2} = \frac{1}{2}(3n^3-3n^2k+nk^2-n+6k)  \,.
\label{eq:comp_to_ind}
\ee
There remains only the $n=0$ boundary. Here the cohomologies are given by $\binom{k+2}{2} = \frac{1}{2}(1+k)(2+k)$.

When $k < 3n$, some step functions are not satisfied. In particular, the upper limit on the sum becomes $\floor{\frac{k}{3}}$. Since the only appearance of $n$ in the sum is in the limit, we can simply make the replacement $n \to \floor{\frac{k}{3}}$ in Equation~\eqref{eq:binom_sum_h0}. From the lower limit, we should only include the sum when $\floor{\frac{k}{3}} \geq 2$. However one can check the sum is anyway zero when $n=1$, so the replacement is still correct when $\floor{\frac{k}{3}} = 1$, i.e.\ $k=3,4,5$. Additionally, one can check that $\sum_{i=2}^{\floor{\frac{k}{3}}} \binom{k-3i+2}{2} = 0$ when $k=1,2$ (but not when $k=0$). Hence the replacement $n \to \floor{\frac{k}{3}}$ in Equation~\eqref{eq:binom_sum_h0} gives the correct expression for $h^0$ in the region $k < 3n$, $n > 0$, but not on the boundary $k=0$, $n\geq0$. On the boundary $h^0(\cy{3},\lb)$ is given simply by the non-sum term $\binom{k+2}{2}|_{k=0} = 1$.

Glancing at the properties of the three-fold above, we see that the region $n\geq0$, $k \geq 3n$ precisely corresponds to the nef cone. Hence by Kodaira vanishing, we know that in the interior of this cone the zeroth cohomology must be given by the index. We see that this is indeed borne out by our results, from comparison of Equation~\eqref{eq:comp_to_ind} with the expression for the index in Equation~\eqref{eq:ex_3f_ind}.

We can then compactly write all the expressions in terms of the index as follows. We depict the regions where the formulae apply in \fref{fig:p2tf_h0_fig}.
\be
h^0(\cy{3},\lb) = 
\begin{cases}
\ind\big(\cy{3}, \mc{O}_{\cy{3}}(n,k)\big)											& \mathrm{for}~n > 0 \,, k\geq3n \,, \\
\ind\big(\cy{3}, \mc{O}_{\cy{3}}(\floor{\frac{k}{3}},k)\big)								& \mathrm{for}~k > 0 \,, k<3n \,, \\
\ind\big(\cy{3}, \mc{O}_{\cy{3}}(1,k)\big)											& \mathrm{for}~n=0 \,, k > 0 \,, \\	
\ind\big(\cy{3}, \mc{O}_{\cy{3}}(1,0)\big)											& \mathrm{for}~n \geq 0 \,, k = 0 \,, \\					
\end{cases}
\label{eq:p2tf_h0_final_form}
\ee
and $h^0(\cy{3},\lb)=0$ otherwise.

\smlhdg{First cohomology}

\noindent On the projective plane, all first cohomologies vanish. Hence the general expression for first cohomologies on the three-fold in Equation~\eqref{eq:final_coh_rels} simplifies, leaving only one term in the direct sum,
\be
H^1(\cy{3},\lb) = H^0(\base{2},R^1\fibprj_* \lb) \,.
\ee
Using the Bott formula for cohomologies on the projective plane, the expression in Equation~\eqref{eq:terms_in_dir_sums} becomes
\be
h^1(\cy{3},\lb) = 
\begin{cases}
\theta(k-3)\binom{k-3+2}{2}													& \mathrm{for}~n=-1,0 \,, \\
\theta(k-3)\binom{k-3+2}{2} + \sum_{i=2}^{-n} \theta(k+3(i-1))\binom{k+3(i-1)+2}{2}	& \mathrm{for}~n\leq-2 \,, \\
\end{cases}
\ee
and $h^0(\cy{3},\lb)=0$ otherwise.

We plot the numerical values in \fref{fig:p2tf_h1_fig}. The step functions in the sum are $\theta(k+3) \,, \theta(k+6) \,, \ldots \,,$ $ \theta(k+3(-n-1))$, which gives in addition to the line $n=0$ the other boundary of the non-zero region as $k= 3(n+1)$. This cone then naturally splits into two regions (at least for $n \leq -2$): when $k \geq -3$, all step functions in the sum are satisfied, while when $k < -3$ the lower limit is raised by the step functions.

\begin{figure}[ht]
\begin{center}
	\includegraphics[scale=.5]{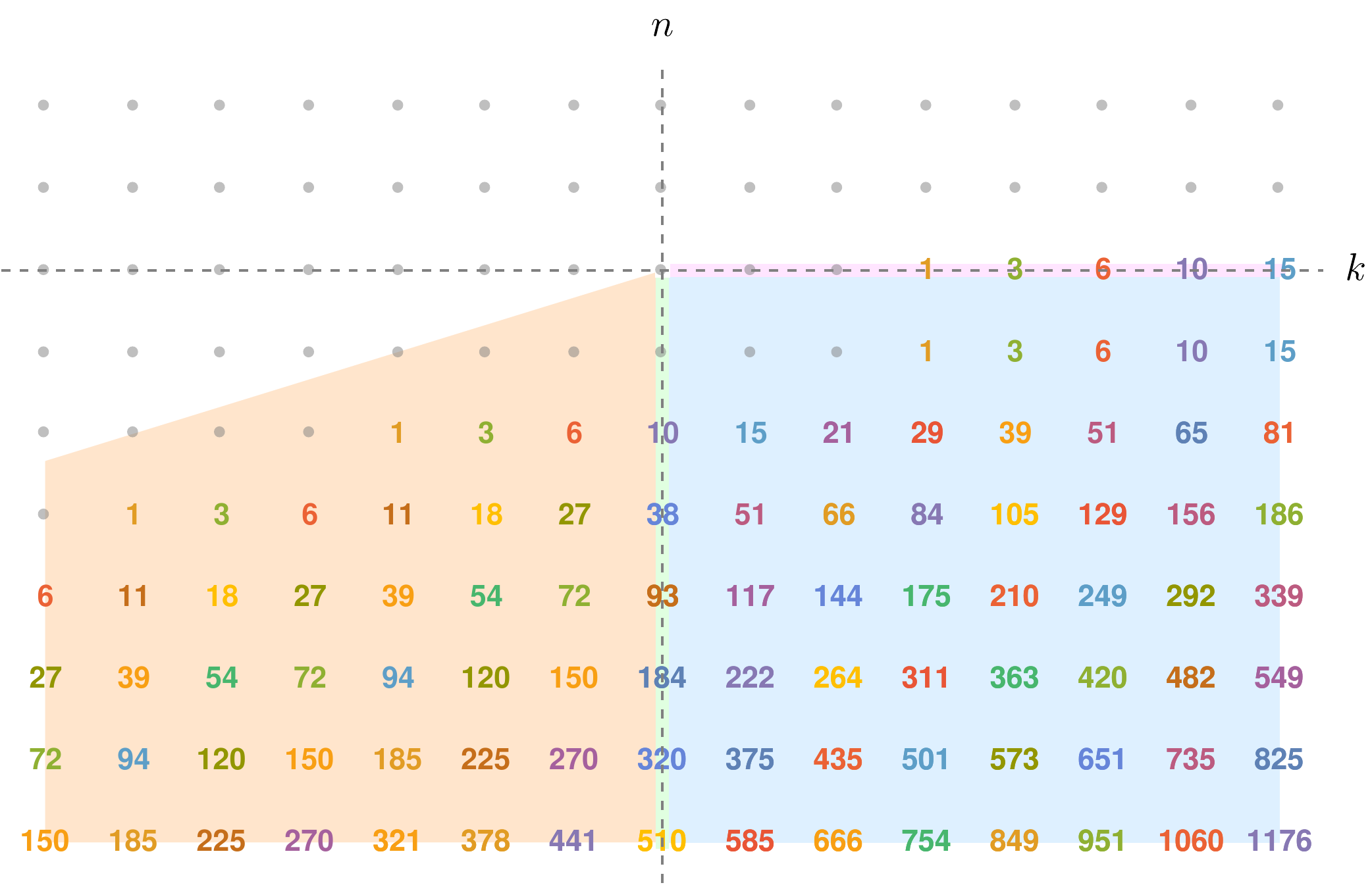} \\[4pt]
	\includegraphics[scale=.55]{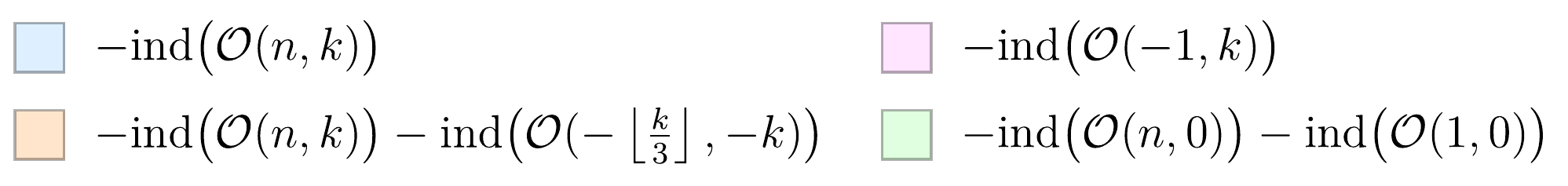}\\[4pt]
	\capt{5.8in}{fig:p2tf_h1_fig}{First line bundle cohomologies $h^1 \big( \mc{O}(n,k) \big)$ on the Weierstrass Calabi-Yau with base $\mbb{P}^2$. We show the cohomologies and the regions where different formulae apply. Here $\mc{O}(n,k)$ is shorthand for $\mc{O}_{\cy{3}}\big(n\fibsec+k\fibprj^*(H)\big)$. 
	}
	\end{center}
\end{figure}

In the $k\geq-3$ region, the sum that appears for $n \leq -2$ has unit coefficients, giving
\be
\sum_{i=2}^{-n} \binom{k+3(i-1)+2}{2} = -\frac{1}{2}(n+1)\left(3(n-1)(n-k)+k^2+2\right)\,.
\label{eq:binom_sum_h1}
\ee
Note this is related to the sum in Equation~\eqref{eq:binom_sum_h0} by a sign and exchanging $(n-1)$ and $(n+1)$. Hence this expression is zero when $n=-1$ (but not when $n=0$), so can be included in the $n=-1$ case. Further, though the non-sum term $\binom{k-3+2}{2}$ is present only for $k \geq 3$, since it is zero when $k=1$ or $k=2$ (but not when $k=-3,-2,-1,0$), the natural region is $k \geq 1$. Hence in the $n < 0$, $k > 0$ region $h^1(\cy{3},\lb)$ is given by
\be
\binom{k-3+2}{2} + \sum_{i=2}^{-n} \binom{k+3(i-1)+2}{2} = \frac{1}{2}(-3n^3+3n^2k-nk^2+n-6k) \,.
\ee
This is simply the negative of the index. It is easy to see this must be so, since by Serre duality and knowledge of the non-zero regions for $h^0$ and $h^1$, only $h^1$ is non-zero in this (open) bottom-right quadrant. Finally, on the boundary region $n=0$, $k > 0$, there is simply the non-sum term, $\binom{k-3+2}{2}$, which we note is $-\ind\big(\mc{O}(-1,k)\big)$. This excludes the origin, on which one has $h^1\big(\mc{O}(0,0)\big)=0$.

In the $k < -3$ region, the lower limit in the sum is $\ceil{\frac{-k+3}{3}}=-\floor{\frac{k}{3}}+1$, and this is also correct for $k=-3,-2,-1$ (but not $k=0$). This gives as the expression for $h^1(\cy{3},\lb)$ in the $k < 0$, $k \geq 3(n+1)$ region,
\be
\sum_{i=-\floor{\frac{k}{3}}+1}^{-n} \binom{k+3(i-1)+2}{2} = -\frac{1}{2}\left(n-\floor{\tfrac{k}{3}}\right)\left(-1+(k-3n)^2+3(k-3n)(n-\floor{\tfrac{k}{3}})+3(n-\floor{\tfrac{k}{3}})^2 \right)\,.
\ee
However from the first factor this expression is anyway zero when $k \in \{3(n+1)-1 \,, 3(n+1)-2 \,, 3(n+1)-3 \}$. So a more natural boundary is $k \geq 3n$. We also note that this expression is manifestly invariant under the shift $(k,n) \to (k + 3, n +1)$, which is obvious in the data in \fref{fig:p2tf_h1_fig}. Since $h^0(\cy{3},\lb)=h^2(\cy{3},\lb)=0$ in this region, the expresssion can also be written as
\be
\begin{aligned}
h^1(\cy{3},\lb) &= -\ind\big(\cy{3},\mc{O}_{\cy{3}}(n,k)\big)-h^3(\cy{3},\lb) \\
& = -\ind(\cy{3}, \mc{O}_{\cy{3}}(n,k)\big)-\ind\big(\cy{3},\mc{O}_{\cy{3}}(-\floor{\tfrac{k}{3}},-k)\big) \,,
\end{aligned}
\ee
where in the second equality we used Serre duality and the expression for $h^0(\cy{3},\lb)$ from above. Finally, on the boundary $n \leq -1$, $k=0$, the expression is just given by Equation~\eqref{eq:binom_sum_h1} in the case $k=0$. We can also note that by Serre duality and the expression for $h^0(\cy{3},\lb)$ above this is $$h^1(\cy{3},\lb)= -\ind\big(\cy{3}, \mc{O}_{\cy{3}}(n,0)\big)-\ind\big(\cy{3},\mc{O}_{\cy{3}}(1,0)\big)$$.

Below we collect the above results using the index expressions, and depict the regions in \fref{fig:p2tf_h1_fig}.
\be
h^1(\cy{3},\lb) = 
\begin{cases}
-\ind\big(\cy{3}, \mc{O}_{\cy{3}}(n,k)\big)											& \mathrm{for}~n < 0 \,, k > 0 \,, \\	
-\ind\big(\cy{3}, \mc{O}_{\cy{3}}(n,k)\big)-\ind\big(\cy{3},\mc{O}_{\cy{3}}(-\floor{\tfrac{k}{3}},-k)\big)		& \mathrm{for}~k \geq 3n \,, k < 0 \,, \\	
-\ind\big(\cy{3}, \mc{O}_{\cy{3}}(-1,k)\big)											& \mathrm{for}~n = 0 \,, k > 0 \,, \\	
-\ind\big(\cy{3}, \mc{O}_{\cy{3}}(n,0)\big)-\ind\big(\cy{3},\mc{O_{\cy{3}}}(1,0)\big)						& \mathrm{for}~n < 0 \,, k = 0 \,, \\	
~0																& \mathrm{for}~n = 0 \,, k = 0 \,, \\	
\end{cases}
\label{eq:p2tf_h1_final_form}
\ee
and $h^1(\cy{3},\lb)=0$ otherwise.

\smlhdg{Insight into general structure}

\noindent Above we have used the Leray spectral sequence to determine the regions and formulae describing all line bundle cohomology on a simple Weierstrass Calabi-Yau three-fold. However, it is clear that this route does not reflect the natural structure of the three-fold cohomology: the naive regions from the lift had to be redrawn, and several special cases had to be taken into account. Hence, while this method is viable for any given smooth Weierstrass model, it does not appear to provide insight into the structure of line bundle cohomology on Calabi-Yau three-folds more generally.

\section{Conclusions}
The central message of the present work is that there are large classes of complex projective manifolds for which the computation of line bundle cohomology is not only tractable, but can be captured throughout the Picard group by closed-form expressions. In the case of surfaces, the idea we have emphasised above is that the combination of Zariski decomposition with vanishing theorems is sufficient on many examples to determine these cohomology formulae.

The role of Zariski decomposition (followed by a round-down operation) is to map effective line bundles to effective line bundles with the same zeroth cohomology. The effective cone naturally splits into chambers inside which the Zariski decomposition retains the same form. For certain classes of surfaces, which include toric surfaces, del Pezzo and generalised del Pezzo surfaces as well as K3 surfaces, there exist theorems strong enough to guarantee the vanishing of all higher cohomologies of any line bundle resulting from the application of a single Zariski decomposition and round-down. In such cases, the Zariski chambers are also cohomology chambers. This gives a complete description of the zeroth cohomology function (except, in the case of K3 surfaces, on the boundary of the Mori cone), and combined with the index formula and Serre duality, this also gives closed-form expressions for the higher cohomologies. The input data required for these formulae is the same as what is needed for building up Zariski chambers: the intersection form, the Mori cone, and the nef cone. 

In the second part of the paper our analysis of simple elliptically fibered Calabi-Yau three-folds leads to closed-form expressions for all line bundle cohomology groups in terms of the cohomology on the two-dimensional base, providing an efficient means to compute line bundle cohomology. On the other hand, the cohomology regions on the three-fold are related to the underlying cohomology chambers on the surface in a complicated way, as illustrated by our analysis of an elliptically fibered three-fold over $\mathbb P^2$. 

\smlhdg{Outlook}

\noindent We would like to end by outlining a few directions of future research. For Physics, the case of Calabi-Yau three-folds has the greatest potential for applications. Our analysis of elliptically fibered Calabi-Yau three-folds has immediate applications in the context of heterotic/F-theory duality, along the lines of research initiated in Refs.~\cite{Braun:2017feb, Braun:2018ovc, Anderson:2019axt}. However, more needs to be understood about the general structure of cohomology formulae on three-folds. If detailed enough, such an understanding will unlock new techniques for working out the topology of the extra dimensions and their gauge degrees of freedom by starting from the raw data of experimental Physics (bottom-up model building). This is certainly an important question for String Phenomenology. 

There are also interesting questions related to the two-dimensional case, the exploration of cohomology formulae for surfaces being far from complete. It would be important to understand what happens in the case of surfaces where the available vanishing theorems are not strong enough to guarantee that Zariski chambers are also zeroth cohomology chambers. The reasonable expectation is that nothing can be said in general, however, it would be interesting to look at other classes of surfaces. 

There is also the question about the higher cohomology groups, in particular $H^1(S,\mathcal L)$. While a formula for $h^1(S,\mathcal L)$ can be derived in terms of the index and the zeroth cohomology, it would be interesting to derive independently the origin of the emerging formulae.

\section*{Acknowledgements}
We are grateful to James Halverson, Andre Lukas, Fabian R\"{u}hle, and Yinan Wang for helpful discussions. 
%

\newpage

\appendix

\section{Toric surfaces and elliptic three-folds}
\label{app:ell_info}
Toric manifolds are used in various examples throughout the text. Moreover, they appear in many string theory applications. For these reasons, in this appendix we provide a brief summary of toric technology, as well as an outline of how to construct a Weierstrass model as a hypersurface in a toric variety,  complementing the discussion in Section~\ref{sec:ell_tf}. For the sake of brevity we will not be completely precise. See for example Ref.~\cite{cox2011toric} for a general introduction to the subject of toric varieties.

\subsection{Toric surfaces}
\label{app:tor_surf}

\smlhdgnogap{Toric varieties}

\noindent Toric varieties are generalisations of projective spaces. Recall that the projective plane, for example, can be specified by beginning with $\mbb{C}^3[x_1,x_2,x_3]$ and imposing an identification $(x_1,x_2,x_3) \sim (\lambda x_1,\lambda x_2,\lambda x_3) ~ \forall \lambda \in \mbb{C}^*$, and demanding that the three coordinates cannot simultaneously vanish, i.e.\ removing the point $(x_1,x_2,x_3)=(0,0,0)$. A toric variety can be specified by beginning with $\mbb{C}^n$ and imposing more complicated identifications under scalings and conditions on allowed simultaneous vanishing coordinates.

The information defining a toric variety can be encoded in a fan, which is a collection of cones in $\mbb{R}^n$ whose generators correspond to integral points. The complex dimension of the corresponding toric variety equals the real dimension $n$ of the fan. The one-dimensional cones are rays. For the example of $\mbb{P}^2$, the fan is shown in the first diagram in \fref{fig:refl_poly}. There are three one-dimensional cones and three two-dimensional cones.

The fan determines the scaling equivalences and allowed coordinate vanishings as follows. Firstly, the generators $\vec{v}_i$ of the rays are in one-to-one correspondence with the complex coordinates $x_i$. A scaling identification between these coordinates corresponds to a linear dependency relation between the generators. Specifically, if
\be
c_i \vec{v}_i = \vec{0} \,,
\ee
then there is an identification $(x_1 , \ldots , x_m) \sim (\lambda^{c_1} x_1 , \ldots, \lambda^{c_m} x_m) ~ \forall \lambda \in \mbb{C}^*$. In the example of $\mbb{P}^2$, one sees from \fref{fig:refl_poly} that the three generators simply add to zero. One typically writes the weights $c_i$ in a table, called the `weight system'.

The allowed simultaneous vanishings of the coordinates are determined by the full structure of the fan: if the generators corresponding to a number of coordinates share a common cone, then the coordinates are allowed to simultaneously vanish, otherwise they are not. In the example of $\mbb{P}^2$, any two generators share a common cone, but all three do not. One typically writes the allowed vanishings in the `Stanley-Reisner ideal'.

\smlhdg{Divisors: intersections and linear equivalences}

\noindent To each complex coordinate $x_i$ one can associate a toric divisor $D_i$ corresponding to the vanishing locus of $x_i$. Hence there are natural divisors corresponding to the rays in the fan. One particularly simple aspect of toric varieties is that the anti-canonical divisor class is given by the sum of the toric divisor classes, $-K = \sum_i \divcls{D_i}$.

The toric divisors are not all linearly inequivalent. In particular, for any vector $\vec{u}$ in the space of the fan, taking the dot product with all generators $\vec{v}_i$ in the toric diagram gives a linear equivalence relation,
\be
0 \sim (\vec{v}_i \cdot \vec{u}) D_i \,.
\ee
It is clear that there are as many independent such relations as the dimension of the fan. Hence on a toric variety, the dimension $h^{1,1}$ of the Picard group equals the number of rays minus the dimension of the fan.

The question of whether distinct toric divisors mutually intersect is the same as the question of whether the corresponding coordinates can simultaneously vanish. This is hence determined by whether the corresponding rays share a cone, as discussed above. Self-intersections are then determined by rewriting a toric divisor in terms of others using a linear equivalence relation, and using mutual intersections.

\medskip

The fan of a toric surface is two-dimensional. In this case, the mutual intersections can be found straightforwardly: if two toric rays are neighbours, then the corresponding toric divisors have a mutual intersection of $1$, otherwise it is~$0$. Self-intersections are determined as usual via the linear equivalence relations.

\smlhdg{Surface case: Mori cone and irreducible negative self-intersection divisors}

\noindent We are particularly interested in the case of toric surfaces, on which divisors and curves can be identified. In the main text we are interested in the Mori cone of these surfaces, and the irreducible, negative self-intersection divisors.

On a compact toric surface, it is straightforward to determine the Mori cone. The Mori cone generators are a subset of the toric divisor classes (see for example Theorem~6.3.20 of Ref.~\cite{cox2011toric}). Since a Mori cone generator has non-positive self-intersection, these must be a subset of those with non-positive self-intersection. This makes it straightforward to determine the generators from the ray diagram: take the toric divisor classes with non-positive self-intersection, and throw away any that can be expressed as a non-negative sum of the others. The nef cone follows as the dual of the Mori cone.

In Zariski decomposition, one is interested in the set of irreducible, negative self-intersection divisors. On a compact toric surface, these are precisely the toric divisor classes with negative self-intersection. To see this, note again that every irreducible, negative self-intersection divisor class is a generator of the Mori cone, and that the Mori cone is generated by toric divisor classes.

\subsection{Toric description of Weierstrass models}

\noindent Since in the simple Weierstrass models utilised in the main text the elliptic fiber is described by a polynomial in the coordinates of the weighted projective plane, which is toric, when the base is also toric it is easy to give a description of the three-fold as a hypersurface in a toric ambient space. In this section we provide a summary of this construction.

The ambient space is a fiber bundle of the weighted projective plane $\mbb{P}_{(2,3,1)}$ over the base space, with the fibering determined by the choice of coordinates for the toric rays. The fan of $\mbb{P}_{(2,3,1)}$ has rays with coordinates $(-1,0)$, $(0,-1)$, $(2,3)$, and in the fan of the four-fold we lift these simply as
\be
\vec{x} = (0,0,-1,0) \,, \quad \vec{y} = (0,0,0,-1) \,, \quad \vec{z} = (0,0,2,3) \,.
\label{eq:fiber_rays}
\ee
In a trivial product $\base{2} \times \mbb{P}_{(2,3,1)}$, we would lift the rays of the base-space analogously, putting zeroes in the final two entries. The non-trivial fibering is achieved by using the non-zero entries
\be
\vec{b} = (b_1,b_2,2,3) \,,
\ee
for any ray $(b_1,b_2)$ of the base-space fan. The projection of the four-fold to the base corresponds to sending the last two coordinates to zero.

By taking the dot product of the above ray vectors with the vectors $(0,0,1,0)$ and $(0,0,0,1)$, we get the divisor equivalences
\be
D_x = 2 \sum_b D_b + 2 D_z \quad \mathrm{and} \quad D_y = 3 \sum_b D_b + 3 D_z \,.
\ee
where the sums run over all rays of the base fan. These equivalences correspond to the statement that the coordinates $x$ and $y$ are sections of $K_{\base{2}}^{-2}$ and $K_{\base{2}}^{-3}$ respectively. The coordinates of these rays determine the weight system, which tabulates the equivalences of coordinates under the various projective scalings. The rows correspond to linear combinations of the rays that sum to zero.
\be
\hspace{2cm}
\begin{tabular}{ C C C C }
b_i 		& x 				& y 				& z		\\ \hline
0		& 2				& 3				& 1		\\
s^1_i	& 2\sum_i s^1_i	& 3\sum_i s^1_i	& 0		\\
\vdots	& \vdots			& \vdots			& \vdots 	\\
\end{tabular}
\hspace{1cm}
\begin{tabular}{ C  }
\sum			\\ \hline
6				\\
6\sum_i s^1_i		\\
\vdots 			\\
\end{tabular}
\ee
Here $b_i$ are the coordinates on the base, and for each $a$, $(s^a_i)$ is a row in the weight system of the base. We have also included on the right the sum of each row, which is the scaling of sections of the anti-canonical bundle of the toric four-fold.

In addition to the rays, one must specify the triangulation of the resulting polytope, i.e.\ the top-dimensional cones of the fan. To describe a fiber bundle, the triangulation should be taken to be simply the product triangulation, i.e.\ the same triangulation as for the fan of the product space $\base{2} \times \mbb{P}_{(2,3,1)}$.
\vspace{12pt}

The elliptic Calabi-Yau is then described as a hypersurface inside this toric ambient space by a Weierstrass equation. One can check that the Weierstrass polynomial is a section of the anti-canonical bundle of the ambient space, so that the three-fold is indeed Calabi-Yau. This corresponds to each monomial having the scaling of the sum of all the columns in the weight system above.

\smlhdg{Example}

\noindent The simplest toric base space is $\mbb{P}^2$. Writing $u$, $v$, and $w$ for the homogeneous coordinates, the rays of the toric four-fold are
\be
\vec{u} = (-1,0,2,3) \,, \quad \vec{v} = (0,-1,2,3) \,, \quad \vec{w} = (1,1,2,3) \,,
\ee
as well as the three rays in Equation~\eqref{eq:fiber_rays}. The weight system is
\be
\begin{tabular}{ C C C C C C }
u 			& v 				& w 				& x 				& y 				& z		\\ \hline
0			& 0				& 0				& 2				& 3				& 1		\\
1			& 1				& 1				& 6				& 9				& 0		\\
\end{tabular}
\ee
The top-dimensional cones have generators given by the product cones. Schematically,
\be
\begin{gathered}
\{ \langle c_1, c_2 \rangle \,|\, c_1 \in \mathrm{Cones}_2(\mbb{P}^2) \,\mathrm{and}\, c_2 \in \mathrm{Cones}_2(\mbb{P}_{(2,3,1)}) \} \,, \\
\mathrm{where}~
\mathrm{Cones}_2(\mbb{P}^2) =  \{ \langle \vec{u} \,, \vec{v} \rangle \,, \langle \vec{u} \,, \vec{w} \rangle \,, \langle \vec{v} \,, \vec{w} \rangle \}
~ \mathrm{and}~ 
\mathrm{Cones}_2(\mbb{P}_{(2,3,1)}) = \{ \langle \vec{x} \,, \vec{y} \rangle \,, \langle \vec{x} \,, \vec{z} \rangle \,, \langle \vec{y} \,, \vec{z} \rangle \} \,.
\end{gathered}
\ee
The polynomials $f$ and $g$ in the Weierstrass equation are in this case homogeneous polynomials of degree 12 and 18 respectively in $u$, $v$, and $w$.


\newpage
\section{The sixteen reflexive polytopes}
\label{app:tor_surf_dat}

\noindent A set of toric surfaces with many applications in string theory are the 16 Gorenstein Fano toric surfaces. We also frequently take these as examples in the main text. In this appendix we collect for each of these surfaces the properties required to determine the Zariski chambers. These properties can be straightforwardly determined using the methods in described in Appendix~\ref{app:tor_surf}.

The ray diagrams of the 16 Gorenstein Fano toric surfaces are given by the 16 reflexive polytopes in \fref{fig:refl_poly}. The rank $\rho(\surf)$ of the Picard group in each case is as follows.
\be
\begin{array}{c | c c c c c c c}
\rho(S) & 1 & 2 & 3 & 4 & 5 & 6 & 7
\\ \hline
S&~F_1~&~F_2\,,F_3\,,F_4~&~F_5\,,F_6~&~F_7\,,F_8\,,F_9\,,F_{10}~&~F_{11}\,,F_{12}~&~F_{13}\,,F_{14}\,,F_{15}~&~F_{16}~
 \end{array}
\ee
Several of these spaces are isomorphic to Hirzebruch or ordinary del Pezzo surfaces. Specifically
\be
F_1 = \mbb{P}^2 \,, \quad F_2 = \mbb{P}^1\times\mbb{P}^1 \,, \quad F_3 = \mathrm{dP}_1 = \mbb{F}_1 \,, \quad F_4 = \mbb{F}_2 \,, \quad F_5 = \mathrm{dP}_2 \,, \quad F_7 = \mathrm{dP}_3 \,,
\ee
and, further, all of the others can be seen as blow-ups of Hirzebruch surfaces. Below we skip the spaces corresponding to the first two reflexive polytopes, which are isomorphic to $\mbb{P}^2$ and $\mbb{P}^1 \times \mbb{P}^1$ and so are trivial.

\begin{figure}[H]
\begin{center}
\raisebox{0in}{\includegraphics[width=10cm]{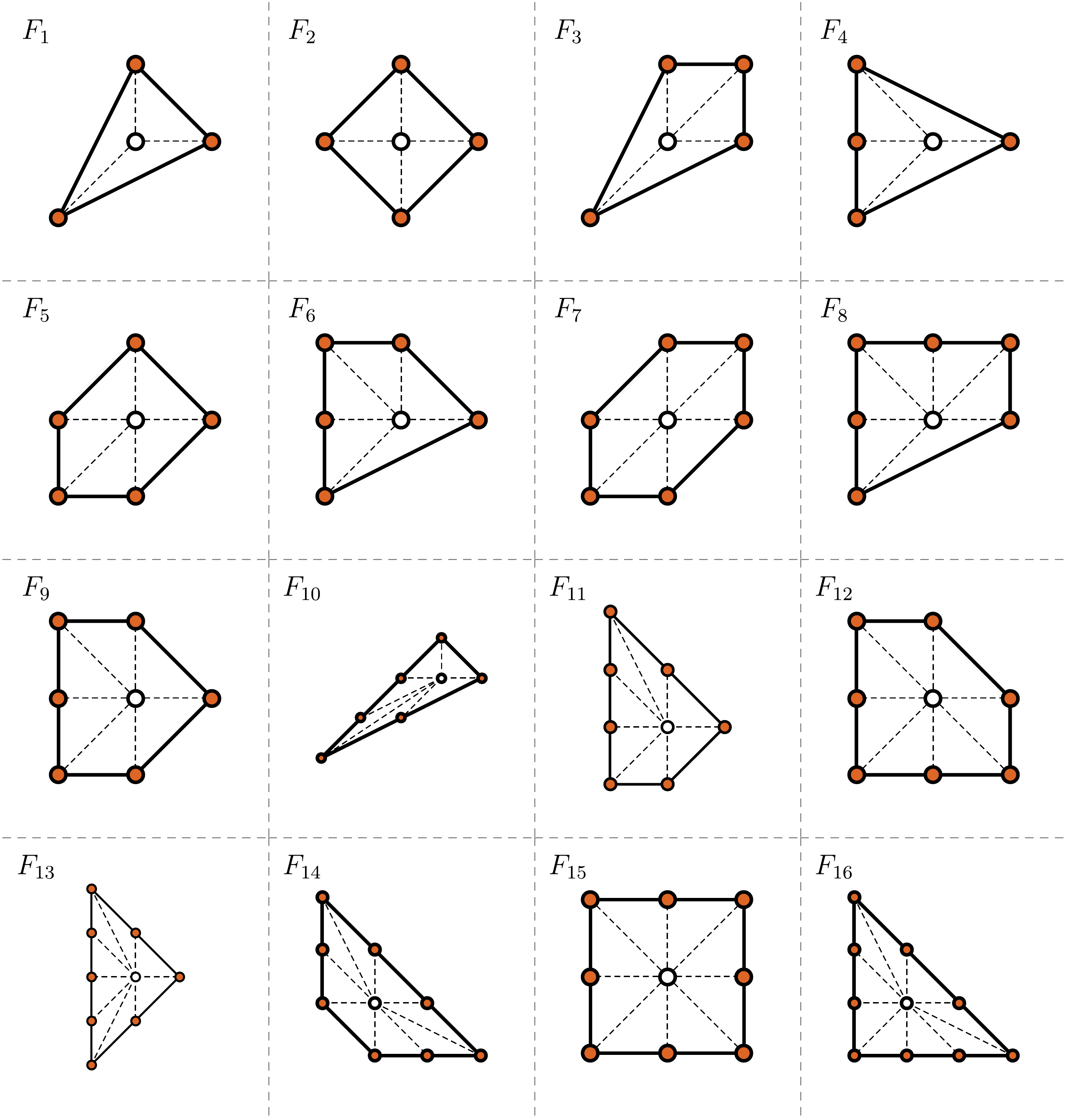}} 
\capt{5.8in}{fig:refl_poly}{The 16 equivalence classes of reflexive lattice polygons in $\mathbb R^2$}
\vspace{-12pt}
\end{center}
\end{figure}

\smlhdg{Data for $F_{3}$}

\noindent The ray diagram for $F_{3}$ has 4 rays, hence its Picard number is $2$. We can use 2 of the toric divisors as a basis for the Picard lattice, and the remaining toric divisors can be expressed in this basis. The choice we use is
\begin{longequation}
\begin{gathered}
\vspace{.3cm}
\begin{tabular} { C  C  C  C } 
\vec{v}_{1}=(0, 1)&\vec{v}_{2}=(1, 1)&\vec{v}_{3}=(1, 0)&\vec{v}_{4}=(-1, -1)\\ 
D_{1}=(1, 0)&D_{2}=(0, 1)&D_{3}=(1, 0)&D_{4}=(1, 1)\\ 
\end{tabular}\\ 
\end{gathered}
\end{longequation}
In this basis the intersection form and the anti-canonical divisor are given by
\begin{equation*}
G:=(D_i\cdot D_j)=
\begin{pmatrix}
0 & ~~1 \\
1 & -1 \\
\end{pmatrix}
\\,\quad-K= \sum_i D_i= (3 , 2)\,.
\end{equation*}
The Mori cone generators $\moricn_i$ and the nef cone generators $\nefcn_j$, which satisfy $\moricn_i \cdot \nefcn_j = \delta_{ij}$,  are given by
\begin{equation*}
\{\moricn_i\} = \{(1, 0)\,, ~(0, 1)\}\,, \quad
\{\nefcn_j\} = \{(1, 1)\,, ~(1, 0)\}\,.
\end{equation*}
There is a single irreducible rigid divisor, namely $\moricn_2$.

\smlhdg{Data for $F_{4}$}

\noindent The ray diagram for $F_{4}$ has 4 rays, hence its Picard number is $2$. We can use 2 of the toric divisors as a basis for the Picard lattice, and the remaining toric divisors can be expressed in this basis. The choice we use is
\begin{longequation}
\begin{gathered}
\vspace{.3cm}
\begin{tabular} { C  C  C  C } 
\vec{v}_{1}=(-1, 0)&\vec{v}_{2}=(-1, 1)&\vec{v}_{3}=(1, 0)&\vec{v}_{4}=(-1, -1)\\ 
D_{1}=(1, 0)&D_{2}=(0, 1)&D_{3}=(1, 2)&D_{4}=(0, 1)\\ 
\end{tabular}\\ 
\end{gathered}
\end{longequation}
In this basis the intersection form and the anti-canonical divisor are given by
\begin{equation*}
G:=(D_i\cdot D_j)=
\begin{pmatrix}
-2 & 1 \\
~~1 & 0 \\
\end{pmatrix}
\\,\quad-K= \sum_i D_i= (2 , 4)\,.
\end{equation*}
The Mori cone generators $\moricn_i$ and the nef cone generators $\nefcn_j$, which satisfy $\moricn_i \cdot \nefcn_j = \delta_{ij}$,  are given by
\begin{equation*}
\{\moricn_i\} = \{(1, 0)\,, ~(0, 1)\}\,, \quad 
\{\nefcn_j\} = \{(0, 1)\,, ~(1, 2)\}\,.
\end{equation*}
There is a single irreducible rigid divisor, namely $\moricn_1$.

\smlhdg{Data for $F_{5}$}

\noindent The ray diagram for $F_{5}$ has 5 rays, hence its Picard number is $3$. We can use 3 of the toric divisors as a basis for the Picard lattice, and the remaining toric divisors can be expressed in this basis. The choice we use is
\begin{longequation}
\begin{gathered}
\vspace{.3cm}
\begin{tabular} { C  C  C  C  C } 
\vec{v}_{1}=(-1, 0)&\vec{v}_{2}=(0, 1)&\vec{v}_{3}=(1, 0)&\vec{v}_{4}=(0, -1)&\vec{v}_{5}=(-1, -1)\\ 
D_{1}=(1, 0, 0)&D_{2}=(0, 1, 0)&D_{3}=(0, 0, 1)&D_{4}=(1, 1, -1)&D_{5}=(-1, 0, 1)\\ 
\end{tabular}
\end{gathered}
\end{longequation}
In this basis the intersection form and the anti-canonical divisor are given by
\begin{equation*}
G:=(D_i\cdot D_j)=
\begin{pmatrix}
-1 & 1 & 0 \\
~~1 & 0 & 1 \\
~~0 & 1 & 0 \\
\end{pmatrix}
\\,\quad-K= \sum_i D_i= (1 , 2 , 1)\,.
\end{equation*}
The Mori cone generators $\moricn_i$ and the nef cone generators $\nefcn_j$, which satisfy $\moricn_i \cdot \nefcn_j = \delta_{ij}$,  are given by
\begin{equation*}
\{\moricn_i\} = \{(1, 0, 0)\,, ~(1, 1, -1)\,, ~(-1, 0, 1)\}\,, \quad
\{\nefcn_j\} = \{(0, 1, 0)\,, ~(0, 0, 1)\,, ~(1, 1, 0)\}\,.
\end{equation*}
The irreducible rigid divisors $\indiv_k$ are just the Mori cone generators.

\smlhdg{Data for $F_{6}$}

\noindent The ray diagram for $F_{6}$ has 5 rays, hence its Picard number is $3$. We can use 3 of the toric divisors as a basis for the Picard lattice, and the remaining toric divisors can be expressed in this basis. The choice we use is
\begin{longequation}
\begin{gathered}
\vspace{.3cm}
\begin{tabular} { C  C  C  C  C } 
\vec{v}_{1}=(-1, 0)&\vec{v}_{2}=(-1, 1)&\vec{v}_{3}=(0, 1)&\vec{v}_{4}=(1, 0)&\vec{v}_{5}=(-1, -1)\\ 
D_{1}=(1, 0, 0)&D_{2}=(0, 1, 0)&D_{3}=(0, 0, 1)&D_{4}=(1, 2, 1)&D_{5}=(0, 1, 1)\\ 
\end{tabular}
\end{gathered}
\end{longequation}
In this basis the intersection form and the anti-canonical divisor are given by
\begin{equation*}
G:=(D_i\cdot D_j)=
\begin{pmatrix}
-2 & ~~1 &~~ 0~ \\
~~1 & -1 & ~~1~ \\
~~0 & ~~1 & -1~ \\
\end{pmatrix}
\\,\quad-K= \sum_i D_i= (2 , 4 , 3)\,.
\end{equation*}
The Mori cone generators $\moricn_i$ and the nef cone generators $\nefcn_j$, which satisfy $\moricn_i \cdot \nefcn_j = \delta_{ij}$,  are given by
\begin{equation*}
\{\moricn_i\} = \{(1, 0, 0)\,, ~(0, 1, 0)\,, ~(0, 0, 1)\}\,, \quad
\{\nefcn_j\} = \{(0, 1, 1)\,, ~(1, 2, 2)\,, ~(1, 2, 1)\}\,.
\end{equation*}
The irreducible rigid divisors $\indiv_k$ are just the Mori cone generators.

\smlhdg{Data for $F_{7}$}

\noindent The ray diagram for $F_{7}$ has 6 rays, hence its Picard number is $4$. We can use 4 of the toric divisors as a basis for the Picard lattice, and the remaining toric divisors can be expressed in this basis. The choice we use is
\begin{longequation}
\begin{gathered}
\vspace{.3cm}
\begin{tabular} { C  C  C  C  C  C } 
\vec{v}_{1}=(-1, 0)&\vec{v}_{2}=(0, 1)&\vec{v}_{3}=(1, 1)&\vec{v}_{4}=(1, 0)&\vec{v}_{5}=(0, -1)&\vec{v}_{6}=(-1, -1)\\ 
D_{1}=(1, 0, 0, 0)&D_{2}=(0, 1, 0, 0)&D_{3}=(0, 0, 1, 0)&D_{4}=(0, 0, 0, 1)&D_{5}=(1, 1, 0, -1)&D_{6}=(-1, 0, 1, 1)\\ 
\end{tabular}
\end{gathered}
\end{longequation}
In this basis the intersection form and the anti-canonical divisor are given by
\begin{equation*}
G:=(D_i\cdot D_j)=
\begin{pmatrix}
-1 & 1 & 0 & 0 \\
1 & -1 & 1 & 0 \\
0 & 1 & -1 & 1 \\
0 & 0 & 1 & -1 \\
\end{pmatrix}
\\,\quad-K= \sum_i D_i= (1 , 2 , 2 , 1)\,.
\end{equation*}
The Mori cone generators $\moricn_i$ and the nef cone generators $\nefcn_j$ are given by
\begin{equation*}
\begin{gathered}
\{\moricn_i\} = \{(1, 0, 0, 0)\,, ~(0, 1, 0, 0)\,, ~(0, 0, 1, 0)\,, ~(0, 0, 0, 1)\,,  ~(1, 1, 0, -1)\,, ~(-1, 0, 1, 1)\}\,,
\\
\{\nefcn_j\} = \{(1, 1, 0, 0)\,, ~(0, 0, 1, 1)\,, ~(1, 1, 1, 0)\,, ~(0, 1, 1, 0)\,, ~(0, 1, 1, 1)\}\,.
\end{gathered}
\end{equation*}
The irreducible rigid divisors $\indiv_k$ are just the Mori cone generators.

\smlhdg{Data for $F_{8}$}

\noindent The ray diagram for $F_{8}$ has 6 rays, hence its Picard number is $4$. We can use 4 of the toric divisors as a basis for the Picard lattice, and the remaining toric divisors can be expressed in this basis. The choice we use is
\begin{longequation}
\begin{gathered}
\vspace{.3cm}
\begin{tabular} { C  C  C  C  C  C } 
\vec{v}_{1}=(-1, 0)&\vec{v}_{2}=(-1, 1)&\vec{v}_{3}=(0, 1)&\vec{v}_{4}=(1, 1)&\vec{v}_{5}=(1, 0)&\vec{v}_{6}=(-1, -1)\\ 
D_{1}=(1, 0, 0, 0)&D_{2}=(0, 1, 0, 0)&D_{3}=(0, 0, 1, 0)&D_{4}=(0, 0, 0, 1)&D_{5}=(1, 2, 1, 0)&D_{6}=(0, 1, 1, 1)\\ 
\end{tabular}
\end{gathered}
\end{longequation}
In this basis the intersection form and the anti-canonical divisor are given by
\begin{equation*}
G:=(D_i\cdot D_j)=
\begin{pmatrix}
-2 & 1 & 0 & 0 \\
1 & -1 & 1 & 0 \\
0 & 1 & -2 & 1 \\
0 & 0 & 1 & -1 \\
\end{pmatrix}
\\,\quad-K= \sum_i D_i= (2 , 4 , 3 , 2)\,.
\end{equation*}
The Mori cone generators $\moricn_i$ and the nef cone generators $\nefcn_j$, which satisfy $\moricn_i \cdot \nefcn_j = \delta_{ij}$,  are given by
\begin{equation*}
\begin{gathered}
\{\moricn_i\} = \{(1, 0, 0, 0)\,, ~(0, 1, 0, 0)\,, ~(0, 0, 1, 0)\,, ~(0, 0, 0, 1)\}\,, \\
\{\nefcn_j\} = \{(0, 1, 1, 1)\,, ~(1, 2, 2, 2)\,, ~(1, 2, 1, 1)\,, ~(1, 2, 1, 0)\}\,.
\end{gathered}
\end{equation*}
The irreducible rigid divisors $\indiv_k$ are just the Mori cone generators.

\smlhdg{Data for $F_{9}$}

\noindent The ray diagram for $F_{9}$ has 6 rays, hence its Picard number is $4$. We can use 4 of the toric divisors as a basis for the Picard lattice, and the remaining toric divisors can be expressed in this basis. The choice we use is
\begin{longequation}
\begin{gathered}
\vspace{.3cm}
\begin{tabular} { C  C  C  C  C  C } 
\vec{v}_{1}=(-1, 0)&\vec{v}_{2}=(-1, 1)&\vec{v}_{3}=(0, 1)&\vec{v}_{4}=(1, 0)&\vec{v}_{5}=(0, -1)&\vec{v}_{6}=(-1, -1)\\ 
D_{1}=(1, 0, 0, 0)&D_{2}=(0, 1, 0, 0)&D_{3}=(0, 0, 1, 0)&D_{4}=(0, 0, 0, 1)&D_{5}=(1, 2, 1, -1)&D_{6}=(-1, -1, 0, 1)\\ 
\end{tabular}
\end{gathered}
\vspace{-12pt}
\end{longequation}
In this basis the intersection form and the anti-canonical divisor are given by
\begin{equation*}
G:=(D_i\cdot D_j)=
\begin{pmatrix}
-2 & 1 & 0 & 0 \\
1 & -1 & 1 & 0 \\
0 & 1 & -1 & 1 \\
0 & 0 & 1 & 0 \\
\end{pmatrix}
\\,\quad-K= \sum_i D_i= (1 , 2 , 2 , 1)\,.
\end{equation*}
The Mori cone generators $\moricn_i$ and the nef cone generators $\nefcn_j$ are given by
\vspace{-5pt}
\begin{equation*}
\begin{gathered}
\{\moricn_i\} = \{(1, 0, 0, 0)\,, ~(0, 1, 0, 0)\,, ~(0, 0, 1, 0)\,, ~(1, 2, 1, -1)\,, ~(-1, -1, 0, 1)\}\,, \\
\{\nefcn_j\} = \{(1, 2, 1, 0)\,, ~(0, 0, 0, 1)\,, ~(1, 2, 2, 0)\,, ~(0, 0, 1, 1)\,, ~(0, 1, 1, 0)\}\,.
\end{gathered}
\vspace{-5pt}
\end{equation*}
The irreducible rigid divisors $\indiv_k$ are just the Mori cone generators.

\smlhdg{Data for $F_{10}$}

\noindent The ray diagram for $F_{10}$ has 6 rays, hence its Picard number is $4$. We can use 4 of the toric divisors as a basis for the Picard lattice, and the remaining toric divisors can be expressed in this basis. The choice we use is
\begin{longequation}
\begin{gathered}
\vspace{.3cm}
\begin{tabular} { C  C  C  C  C  C } 
\vec{v}_{1}=(-1, 0)&\vec{v}_{2}=(0, 1)&\vec{v}_{3}=(1, 0)&\vec{v}_{4}=(-1, -1)&\vec{v}_{5}=(-3, -2)&\vec{v}_{6}=(-2, -1)\\ 
D_{1}=(1, 0, 0, 0)&D_{2}=(0, 1, 0, 0)&D_{3}=(0, 0, 1, 0)&D_{4}=(0, 0, 0, 1)&D_{5}=(1, 2, -1, -1)&D_{6}=(-2, -3, 2, 1)\\ 
\end{tabular}
\end{gathered}
\vspace{-12pt}
\end{longequation}
In this basis the intersection form and the anti-canonical divisor are given by
\begin{equation*}
G:=(D_i\cdot D_j)=
\begin{pmatrix}
-2 & 1 & 0 & 0 \\
1 & 0 & 1 & 0 \\
0 & 1 & 1 & 1 \\
0 & 0 & 1 & -2 \\
\end{pmatrix}
\\,\quad-K= \sum_i D_i= (0 , 0 , 2 , 1)\,.
\end{equation*}
The Mori cone generators $\moricn_i$ and the nef cone generators $\nefcn_j$, which satisfy $\moricn_i \cdot \nefcn_j = \delta_{ij}$,  are given by
\vspace{-5pt}
\begin{equation*}
\begin{gathered}
\{\moricn_i\} = \{(1, 0, 0, 0)\,, ~(0, 0, 0, 1)\,, ~(1, 2, -1, -1)\,, ~(-2, -3, 2, 1)\}\,, \\
\{\nefcn_j\} = \{(0, 1, 0, 0)\,, ~(0, 0, 1, 0)\,, ~(0, 0, 2, 1)\,, ~(1, 2, 0, 0)\}\,.
\end{gathered}
\vspace{-5pt}
\end{equation*}
The irreducible rigid divisors $\indiv_k$ are just the Mori cone generators.

\smlhdg{Data for $F_{11}$}

\noindent The ray diagram for $F_{11}$ has 7 rays, hence its Picard number is $5$. We can use 5 of the toric divisors as a basis for the Picard lattice, and the remaining toric divisors can be expressed in this basis. The choice we use is
\begin{longequation}
\begin{gathered}
\vspace{.3cm}
\begin{tabular} { C  C  C  C } 
\vec{v}_{1}=(-1, 0)&\vec{v}_{2}=(-1, 1)&\vec{v}_{3}=(-1, 2)&\vec{v}_{4}=(0, 1)\\ 
D_{1}=(1, 0, 0, 0, 0)&D_{2}=(0, 1, 0, 0, 0)&D_{3}=(0, 0, 1, 0, 0)&D_{4}=(0, 0, 0, 1, 0)\\ 
\end{tabular}\\ 
\vspace{.3cm}
\begin{tabular} { C  C  C } 
\vec{v}_{5}=(1, 0)&\vec{v}_{6}=(0, -1)&\vec{v}_{7}=(-1, -1)\\ 
D_{5}=(0, 0, 0, 0, 1)&D_{6}=(1, 2, 3, 1, -1)&D_{7}=(-1, -1, -1, 0, 1)\\ 
\end{tabular}\\ 
\end{gathered}
\end{longequation}
In this basis the intersection form and the anti-canonical divisor are given by
\begin{equation*}
G:=(D_i\cdot D_j)=
\begin{pmatrix}
-2 & 1 & 0 & 0 & 0 \\
1 & -2 & 1 & 0 & 0 \\
0 & 1 & -1 & 1 & 0 \\
0 & 0 & 1 & -2 & 1 \\
0 & 0 & 0 & 1 & 0 \\
\end{pmatrix}
\\,\quad-K= \sum_i D_i= (1 , 2 , 3 , 2 , 1)\,.
\end{equation*}
The Mori cone generators $\moricn_i$ and the nef cone generators $\nefcn_j$ are given by
\begin{longequation}
\begin{gathered}
\{\moricn_i\} = \{(1, 0, 0, 0, 0)\,, ~(0, 1, 0, 0, 0)\,, ~(0, 0, 1, 0, 0)\,, ~(0, 0, 0, 1, 0)\,, ~(1, 2, 3, 1, -1)\,, ~(-1, -1, -1, 0, 1)\}\,, \\
\begin{aligned}
\{\nefcn_j\} = \;&\{(1, 2, 3, 1, 0)\,, ~(0, 0, 0, 0, 1)\,, ~(2, 4, 6, 3, 0)\,, ~(0, 0, 0, 1, 2)\,, ~ (1, 2, 4, 2, 0)\,, ~(0, 0, 1, 1, 1)\,, \\ &~(0, 1, 2, 1, 0)\,, ~(0, 2, 4, 2, 0)\}\,.
\end{aligned}
\end{gathered}
\end{longequation}
The irreducible rigid divisors $\indiv_k$ are just the Mori cone generators.

\smlhdg{Data for $F_{12}$}

\noindent The ray diagram for $F_{12}$ has 7 rays, hence its Picard number is $5$. We can use 5 of the toric divisors as a basis for the Picard lattice, and the remaining toric divisors can be expressed in this basis. The choice we use is
\begin{longequation}
\begin{gathered}
\vspace{.3cm}
\begin{tabular} { C  C  C  C } 
\vec{v}_{1}=(-1, 0)&\vec{v}_{2}=(-1, 1)&\vec{v}_{3}=(0, 1)&\vec{v}_{4}=(1, 0)\\ 
D_{1}=(1, 0, 0, 0, 0)&D_{2}=(0, 1, 0, 0, 0)&D_{3}=(0, 0, 1, 0, 0)&D_{4}=(0, 0, 0, 1, 0)\\ 
\end{tabular}\\ 
\vspace{.3cm}
\begin{tabular} { C  C  C } 
\vec{v}_{5}=(1, -1)&\vec{v}_{6}=(0, -1)&\vec{v}_{7}=(-1, -1)\\ 
D_{5}=(0, 0, 0, 0, 1)&D_{6}=(1, 2, 1, -1, -2)&D_{7}=(-1, -1, 0, 1, 1)\\ 
\end{tabular}\\ 
\end{gathered}
\end{longequation}
In this basis the intersection form and the anti-canonical divisor are given by
\begin{equation*}
G:=(D_i\cdot D_j)=
\begin{pmatrix}
-2 & 1 & 0 & 0 & 0 \\
1 & -1 & 1 & 0 & 0 \\
0 & 1 & -1 & 1 & 0 \\
0 & 0 & 1 & -1 & 1 \\
0 & 0 & 0 & 1 & -1 \\
\end{pmatrix}
\\,\quad-K= \sum_i D_i= (1 , 2 , 2 , 1 , 0)\,.
\end{equation*}
The Mori cone generators $\moricn_i$ and the nef cone generators $\nefcn_j$ are given by
\begin{equation*}
\begin{gathered}
\begin{aligned}
\{\moricn_i\} &= \{(1, 0, 0, 0, 0)\,, ~(0, 1, 0, 0, 0)\,, ~(0, 0, 1, 0, 0)\,, ~(0, 0, 0, 1, 0)\,, ~(0, 0, 0, 0, 1)\,, ~ \\&~~~~ (1, 2, 1, -1, -2)\,, ~(-1, -1, 0, 1, 1)\}\,,
\end{aligned} \\
\begin{aligned}
\{\nefcn_j\} &= \{(0, 0, 0, 1, 1)\,, ~(1, 2, 1, 1, 0)\,, ~(1, 2, 1, 0, 0)\,, ~(0, 0, 2, 2, 0)\,, ~(0, 0, 1, 1, 0)\,, ~ \\&~~~~ (1, 2, 2, 0, 0)\,, ~(0, 0, 1, 1, 1)\,, ~(0, 1, 1, 0, 0)\,, ~(0, 1, 1, 1, 0)\}\,.
\end{aligned}
\end{gathered}
\end{equation*}
The irreducible rigid divisors $\indiv_k$ are just the Mori cone generators.

\smlhdg{Data for $F_{13}$}

\noindent The ray diagram for $F_{13}$ has 8 rays, hence its Picard number is $6$. We can use 6 of the toric divisors as a basis for the Picard lattice, and the remaining toric divisors can be expressed in this basis. The choice we use is
\begin{longequation}
\begin{gathered}
\vspace{.3cm}
\begin{tabular} { C  C  C  C } 
\vec{v}_{1}=(-1, 0)&\vec{v}_{2}=(-1, 1)&\vec{v}_{3}=(-1, 2)&\vec{v}_{4}=(0, 1)\\ 
D_{1}=(1, 0, 0, 0, 0, 0)&D_{2}=(0, 1, 0, 0, 0, 0)&D_{3}=(0, 0, 1, 0, 0, 0)&D_{4}=(0, 0, 0, 1, 0, 0)\\ 
\end{tabular}\\ 
\vspace{.3cm}
\begin{tabular} { C  C  C  C } 
\vec{v}_{5}=(1, 0)&\vec{v}_{6}=(0, -1)&\vec{v}_{7}=(-1, -2)&\vec{v}_{8}=(-1, -1)\\ 
D_{5}=(0, 0, 0, 0, 1, 0)&D_{6}=(0, 0, 0, 0, 0, 1)&D_{7}=(1, 2, 3, 1, -1, -1)&D_{8}=(-2, -3, -4, -1, 2, 1)\\ 
\end{tabular}\\ 
\end{gathered}
\end{longequation}
In this basis the intersection form and the anti-canonical divisor are given by
\begin{equation*}
G:=(D_i\cdot D_j)=
\begin{pmatrix}
-2 & 1 & 0 & 0 & 0 & 0 \\
1 & -2 & 1 & 0 & 0 & 0 \\
0 & 1 & -1 & 1 & 0 & 0 \\
0 & 0 & 1 & -2 & 1 & 0 \\
0 & 0 & 0 & 1 & 0 & 1 \\
0 & 0 & 0 & 0 & 1 & -2 \\
\end{pmatrix}
\\,\quad-K= \sum_i D_i= (0 , 0 , 0 , 1 , 2 , 1)\,.
\end{equation*}
The Mori cone generators $\moricn_i$ and the nef cone generators $\nefcn_j$ are given by
\begin{equation*}
\begin{gathered}
\begin{aligned}
\{\moricn_i\} &= \{(1, 0, 0, 0, 0, 0)\,, ~(0, 1, 0, 0, 0, 0)\,, ~(0, 0, 1, 0, 0, 0)\,, ~(0, 0, 0, 1, 0, 0)\,, ~ \\&~~~~ (0, 0, 0, 0, 0, 1)\,, ~(1, 2, 3, 1, -1, -1)\,, ~(-2, -3, -4, -1, 2, 1)\}\,,
\end{aligned}
\\
\begin{aligned}
\{\nefcn_j\} &= \{(1, 2, 3, 1, 0, 0)\,, ~(0, 0, 0, 0, 2, 1)\,, ~(0, 0, 0, 0, 1, 0)\,, ~(2, 4, 6, 3, 0, 0)\,, ~(0, 0, 0, 2, 4, 2)\,, ~(0, 0, 0, 1, 2, 0)\,, ~ \\&~~~~ (1, 2, 4, 2, 0, 0)\,, ~(0, 0, 2, 2, 2, 1)\,, ~(0, 0, 1, 1, 1, 0)\,, ~(0, 1, 2, 1, 0, 0)\,, ~(0, 2, 4, 2, 0, 0)\}\,.
\end{aligned}
\end{gathered}
\end{equation*}
The irreducible rigid divisors $\indiv_k$ are just the Mori cone generators.

\vspace{21pt}
\smlhdg{Data for $F_{14}$}

\noindent The ray diagram for $F_{14}$ has 8 rays, hence its Picard number is $6$. We can use 6 of the toric divisors as a basis for the Picard lattice, and the remaining toric divisors can be expressed in this basis. The choice we use is
\begin{longequation}
\begin{gathered}
\vspace{.3cm}
\begin{tabular} { C  C  C  C } 
\vec{v}_{1}=(-1, 0)&\vec{v}_{2}=(-1, 1)&\vec{v}_{3}=(-1, 2)&\vec{v}_{4}=(0, 1)\\ 
D_{1}=(1, 0, 0, 0, 0, 0)&D_{2}=(0, 1, 0, 0, 0, 0)&D_{3}=(0, 0, 1, 0, 0, 0)&D_{4}=(0, 0, 0, 1, 0, 0)\\ 
\end{tabular}\\ 
\vspace{.3cm}
\begin{tabular} { C  C  C  C } 
\vec{v}_{5}=(1, 0)&\vec{v}_{6}=(2, -1)&\vec{v}_{7}=(1, -1)&\vec{v}_{8}=(0, -1)\\ 
D_{5}=(0, 0, 0, 0, 1, 0)&D_{6}=(0, 0, 0, 0, 0, 1)&D_{7}=(1, 1, 1, 0, -1, -2)&D_{8}=(-1, 0, 1, 1, 1, 1)\\ 
\end{tabular}\\ 
\end{gathered}
\end{longequation}
In this basis the intersection form and the anti-canonical divisor are given by
\begin{equation*}
G:=(D_i\cdot D_j)=
\begin{pmatrix}
-1 & 1 & 0 & 0 & 0 & 0 \\
1 & -2 & 1 & 0 & 0 & 0 \\
0 & 1 & -1 & 1 & 0 & 0 \\
0 & 0 & 1 & -2 & 1 & 0 \\
0 & 0 & 0 & 1 & -2 & 1 \\
0 & 0 & 0 & 0 & 1 & -1 \\
\end{pmatrix}
\\,\quad-K= \sum_i D_i= (1 , 2 , 3 , 2 , 1 , 0)\,.
\end{equation*}
The Mori cone generators $\moricn_i$ and the nef cone generators $\nefcn_j$ are given by
\begin{equation*}
\begin{gathered}
\begin{aligned}
\{\moricn_i\} &= \{(1, 0, 0, 0, 0, 0)\,, ~(0, 1, 0, 0, 0, 0)\,, ~(0, 0, 1, 0, 0, 0)\,, ~(0, 0, 0, 1, 0, 0)\,, ~ \\&~~~~ (0, 0, 0, 0, 1, 0)\,, ~(0, 0, 0, 0, 0, 1)\,, ~(1, 1, 1, 0, -1, -2)\,, ~(-1, 0, 1, 1, 1, 1)\}\,,
\end{aligned}
\\
\begin{aligned}
\{\nefcn_j\} &= \{(1, 1, 1, 0, 0, 0)\,, ~(2, 2, 2, 0, 0, 0)\,, ~(1, 1, 1, 1, 1, 1)\,, ~(3, 3, 3, 2, 1, 0)\,, ~(2, 2, 2, 1, 0, 0)\,, ~ \\&~~~~ (0, 0, 1, 1, 1, 1)\,, ~(1, 1, 3, 2, 1, 0)\,, ~(1, 1, 2, 1, 0, 0)\,, ~(0, 1, 2, 1, 0, 0)\,, ~(0, 1, 2, 1, 1, 1)\,, ~ \\&~~~~ (0, 2, 4, 2, 1, 0)\,, ~(0, 1, 2, 2, 2, 2)\,, ~(0, 3, 6, 4, 2, 0)\,, ~(0, 2, 4, 2, 0, 0)\,, ~(0, 1, 3, 2, 1, 0)\}\,.
\end{aligned}
\end{gathered}
\end{equation*}
The irreducible rigid divisors $\indiv_k$ are just the Mori cone generators.

\smlhdg{Data for $F_{15}$}

\noindent The ray diagram for $F_{15}$ has 8 rays, hence its Picard number is $6$. We can use 6 of the toric divisors as a basis for the Picard lattice, and the remaining toric divisors can be expressed in this basis. The choice we use is
\begin{longequation}
\begin{gathered}
\vspace{.3cm}
\begin{tabular} { C  C  C  C } 
\vec{v}_{1}=(-1, 0)&\vec{v}_{2}=(-1, 1)&\vec{v}_{3}=(0, 1)&\vec{v}_{4}=(1, 1)\\ 
D_{1}=(1, 0, 0, 0, 0, 0)&D_{2}=(0, 1, 0, 0, 0, 0)&D_{3}=(0, 0, 1, 0, 0, 0)&D_{4}=(0, 0, 0, 1, 0, 0)\\ 
\end{tabular}\\ 
\vspace{.3cm}
\begin{tabular} { C  C  C  C } 
\vec{v}_{5}=(1, 0)&\vec{v}_{6}=(1, -1)&\vec{v}_{7}=(0, -1)&\vec{v}_{8}=(-1, -1)\\ 
D_{5}=(0, 0, 0, 0, 1, 0)&D_{6}=(0, 0, 0, 0, 0, 1)&D_{7}=(1, 2, 1, 0, -1, -2)&D_{8}=(-1, -1, 0, 1, 1, 1)\\ 
\end{tabular}\\ 
\end{gathered}
\end{longequation}
In this basis the intersection form and the anti-canonical divisor are given by
\begin{equation*}
G:=(D_i\cdot D_j)=
\begin{pmatrix}
-2 & 1 & 0 & 0 & 0 & 0 \\
1 & -1 & 1 & 0 & 0 & 0 \\
0 & 1 & -2 & 1 & 0 & 0 \\
0 & 0 & 1 & -1 & 1 & 0 \\
0 & 0 & 0 & 1 & -2 & 1 \\
0 & 0 & 0 & 0 & 1 & -1 \\
\end{pmatrix}
\\,\quad-K= \sum_i D_i= (1 , 2 , 2 , 2 , 1 , 0)\,.
\end{equation*}
The Mori cone generators $\moricn_i$ and the nef cone generators $\nefcn_j$ are given by
\begin{equation*}
\begin{gathered}
\begin{aligned}
\{\moricn_i\} &= \{(1, 0, 0, 0, 0, 0)\,, ~(0, 1, 0, 0, 0, 0)\,, ~(0, 0, 1, 0, 0, 0)\,, ~(0, 0, 0, 1, 0, 0)\,, ~ \\&~~~~ (0, 0, 0, 0, 1, 0)\,, ~(0, 0, 0, 0, 0, 1)\,, ~(1, 2, 1, 0, -1, -2)\,, ~(-1, -1, 0, 1, 1, 1)\}\,,
\end{aligned}
\\
\begin{aligned}
\{\nefcn_j\} &= \{(1, 2, 1, 0, 0, 0)\,, ~(2, 4, 2, 0, 0, 0)\,, ~(0, 0, 0, 1, 1, 1)\,, ~(1, 2, 1, 2, 1, 0)\,, ~(1, 2, 1, 1, 0, 0)\,, ~ \\&~~~~ (0, 0, 2, 4, 2, 0)\,, ~(0, 0, 1, 2, 1, 0)\,, ~(1, 2, 2, 2, 0, 0)\,, ~(0, 0, 1, 2, 1, 1)\,, ~(0, 0, 1, 2, 2, 2)\,, ~ \\&~~~~ (0, 1, 1, 1, 0, 0)\,, ~(0, 1, 1, 1, 1, 1)\,, ~(0, 2, 2, 2, 1, 0)\,, ~(0, 1, 1, 2, 1, 0)\}\,.
\end{aligned}
\end{gathered}
\end{equation*}
The irreducible rigid divisors $\indiv_k$ are just the Mori cone generators.

\smlhdg{Data for $F_{16}$}

\noindent The ray diagram for $F_{16}$ has 9 rays, hence its Picard number is $7$. We can use 7 of the toric divisors as a basis for the Picard lattice, and the remaining toric divisors can be expressed in this basis. The choice we use is
\begin{longequation}
\begin{gathered}
\vspace{.3cm}
\begin{tabular} { C  C  C  C  C } 
\vec{v}_{1}=(-1, 0)&\vec{v}_{2}=(-1, 1)&\vec{v}_{3}=(-1, 2)&\vec{v}_{4}=(0, 1)&\vec{v}_{5}=(1, 0)\\ 
D_{1}=(1, 0, 0, 0, 0, 0, 0)&D_{2}=(0, 1, 0, 0, 0, 0, 0)&D_{3}=(0, 0, 1, 0, 0, 0, 0)&D_{4}=(0, 0, 0, 1, 0, 0, 0)&D_{5}=(0, 0, 0, 0, 1, 0, 0)\\ 
\end{tabular}\\ 
\vspace{.3cm}
\begin{tabular} { C  C  C  C  C } 
\vec{v}_{6}=(2, -1)&\vec{v}_{7}=(1, -1)&\vec{v}_{8}=(0, -1)&\vec{v}_{9}=(-1, -1)\\ 
D_{6}=(0, 0, 0, 0, 0, 1, 0)&D_{7}=(0, 0, 0, 0, 0, 0, 1)&D_{8}=(1, 2, 3, 1, -1, -3, -2)&D_{9}=(-1, -1, -1, 0, 1, 2, 1)\\ 
\end{tabular}
\end{gathered}
\end{longequation}
In this basis the intersection form and the anti-canonical divisor are given by
\begin{equation*}
G:=(D_i\cdot D_j)=
\begin{pmatrix}
-2 & 1 & 0 & 0 & 0 & 0 & 0 \\
1 & -2 & 1 & 0 & 0 & 0 & 0 \\
0 & 1 & -1 & 1 & 0 & 0 & 0 \\
0 & 0 & 1 & -2 & 1 & 0 & 0 \\
0 & 0 & 0 & 1 & -2 & 1 & 0 \\
0 & 0 & 0 & 0 & 1 & -1 & 1 \\
0 & 0 & 0 & 0 & 0 & 1 & -2 \\
\end{pmatrix}
\\,\quad-K= \sum_i D_i= (1 , 2 , 3 , 2 , 1 , 0 , 0)\,.
\end{equation*}
The Mori cone generators $\moricn_i$ and the nef cone generators $\nefcn_j$ are given by
\begin{equation*}
\begin{gathered}
\begin{aligned}
\{\moricn_i\} &= \{(1, 0, 0, 0, 0, 0, 0)\,, ~(0, 1, 0, 0, 0, 0, 0)\,, ~(0, 0, 1, 0, 0, 0, 0)\,, ~(0, 0, 0, 1, 0, 0, 0)\,, ~(0, 0, 0, 0, 1, 0, 0)\,, ~ \\&~~~~(0, 0, 0, 0, 0, 1, 0)\,, ~(0, 0, 0, 0, 0, 0, 1)\,, ~(1, 2, 3, 1, -1, -3, -2)\,, ~(-1, -1, -1, 0, 1, 2, 1)\}\,,
\end{aligned}
\\
\begin{aligned}
\{\nefcn_j\} &= \{(0, 0, 0, 0, 1, 2, 1)\,, ~(1, 2, 3, 1, 1, 1, 0)\,, ~(2, 4, 6, 2, 1, 0, 0)\,, ~(0, 0, 0, 0, 2, 4, 2)\,, ~(1, 2, 3, 1, 0, 0, 0)\,, ~ \\&~~~~ (1, 2, 3, 3, 3, 3, 0)\,, ~(0, 0, 0, 1, 2, 3, 1)\,, ~(3, 6, 9, 6, 3, 0, 0)\,, ~(0, 0, 0, 2, 4, 6, 3)\,, ~(2, 4, 6, 3, 0, 0, 0)\,, ~ \\&~~~~ (0, 0, 0, 1, 2, 4, 2)\,, ~(0, 0, 2, 2, 2, 2, 0)\,, ~(0, 0, 1, 1, 1, 1, 0)\,, ~(1, 2, 6, 4, 2, 0, 0)\,, ~(0, 0, 2, 2, 2, 2, 1)\,, ~ \\&~~~~ (1, 2, 4, 2, 0, 0, 0)\,, ~(0, 0, 1, 1, 1, 2, 1)\,, ~(0, 1, 2, 1, 0, 0, 0)\,, ~(0, 1, 2, 1, 1, 1, 0)\,, ~(0, 2, 4, 2, 1, 0, 0)\,, ~ \\&~~~~ (0, 1, 2, 2, 2, 2, 0)\,, ~(0, 3, 6, 4, 2, 0, 0)\,, ~(0, 2, 4, 2, 0, 0, 0)\,, ~(0, 1, 3, 2, 1, 0, 0)\}\,.
\end{aligned}
\end{gathered}
\end{equation*}
The irreducible rigid divisors $\indiv_k$ are just the Mori cone generators.

\newpage

\providecommand{\href}[2]{#2}\begingroup\raggedright\endgroup

\end{document}